\documentclass[final, nomarks]{dmtcs-episciences}

\usepackage[utf8]{inputenc}
\usepackage{float}
\usepackage{graphicx}
\graphicspath{{./figures/},{./figures/}}
\usepackage{amsthm, amsmath}
\usepackage{comment}
\usepackage{enumerate}
\usepackage{color}
\usepackage{xspace}
\usepackage{xcolor}
\usepackage{url}
\usepackage{hyperref}
\usepackage{comment}
\usepackage{multirow}

\usepackage{amsfonts}
\usepackage{amssymb}
\usepackage{thmtools}
\usepackage{mathscinet}
\usepackage{svg}
\usepackage{thm-restate}
\newtheorem{theorem}{Theorem}
 \newtheorem{corollary}[theorem]{Corollary}

\newtheorem{lemma}[theorem]{Lemma}
 \newtheorem{claim}{Claim}

\usepackage{stmaryrd}

\newcommand{\Z}{\ensuremath{\mathbb{Z}}\xspace}
\newcommand{\Q}{\ensuremath{\mathbb{Q}}\xspace}
\newcommand{\R}{\ensuremath{\mathbb{R}}\xspace}
\newcommand{\N}{\ensuremath{\mathbb{N}}\xspace}
\newcommand{\A}{\ensuremath{\mathcal{A}}\xspace}
\newcommand{\F}{\ensuremath{\mathcal{F}}\xspace}
\newcommand{\Practical}{\ensuremath{\mathcal{P}}\xspace}
\newcommand{\Theoretical}{\ensuremath{\mathcal{T}}\xspace}
\newcommand{\Both}{\ensuremath{\mathcal{B}}\xspace}

\newcommand{\ER}{\ensuremath{\exists \mathbb{R}}\xspace}
\newcommand{\NP}{\ensuremath{\textrm{NP}}\xspace}
\newcommand{\vis}{\ensuremath{\textrm{vis}}\xspace}
\newcommand{\VIS}{\ensuremath{\textrm{VIS}}\xspace}
\newcommand{\disth}{\ensuremath{\textrm{dist}_\textrm{H}}\xspace}
\newcommand{\seg}{\ensuremath{\textrm{seg}}\xspace}
\newcommand{\short}{\ensuremath{\textrm{short}}\xspace}
\newcommand{\tangent}{\ensuremath{\textrm{tangent}}\xspace}
\newcommand{\power}{\ensuremath{\textrm{angular capacity}}\xspace}
\newcommand{\powerF}{\ensuremath{\textrm{capacity}}\xspace}
\newcommand{\POwer}{\ensuremath{\textrm{Angular Capacity}}\xspace}

\newcommand{\vertex}{\ensuremath{\textrm{vertex}}\xspace}
\newcommand{\node}{\ensuremath{\textrm{node}}\xspace}
\newcommand{\tafel}{\ensuremath{\textrm{table}}\xspace}
\newcommand{\face}{\ensuremath{\textrm{face}}\xspace}

\newcommand{\eps}{\ensuremath{\varepsilon}\xspace}
\newcommand{\ray}{\ensuremath{\textrm{ray}}\xspace}

\newcommand{\opt}{\ensuremath{\textrm{opt}}\xspace}

\newcommand{\representative}{\ensuremath{\textrm{representative}}\xspace}

\newcommand{\splittable}{\ensuremath{\textrm{splittable}}\xspace}
\newcommand{\interior}{\ensuremath{\textrm{int}}\xspace}
\newcommand{\chord}{\ensuremath{\textrm{chord}}\xspace}
\newcommand{\cw}{\ensuremath{\textrm{cw}}\xspace}
\newcommand{\visionstability}{vision-stability\xspace}
\newcommand{\Visionstability}{Vision-stability\xspace}
\newcommand{\visionstable}{vision-stable\xspace}
\newcommand{\Visionstable}{Vision-stable\xspace}
\newcommand{\enhanced}{enhanced\xspace}

\newcommand{\diminished}{diminished\xspace}

\newcommand{\weakVisPolyTree}{weak visibility polygon tree\xspace}
\newcommand{\WeakVisPolyTree}{Weak visibility polygon tree\xspace}

\newcommand{\const}[1]{\ensuremath{\left\llbracket \, #1 \,\right\rrbracket}\xspace}
\newcommand{\chordwidth}{chord-visibility width\xspace}

\definecolor{darkgreen}{rgb}{0.01, 0.93, 0.29}
\definecolor{lightbrown}{rgb}{0.91, 0.4, 0.11}

\author[Hengeveld and Miltzow]
{Simon B. Hengeveld\affiliationmark{1}
  \and Tillmann Miltzow\affiliationmark{1}\thanks{The second author is generously supported by the NWO Veni grant 016.Veni.192.250.}
 }
\title[Practical Art Gallery]{A Practical Algorithm with Performance Guarantees
for the Art~Gallery~Problem}
\affiliation{Department of Information and Computing Science, Utrecht University, The Netherlands}
  
\keywords{Computational Geometry, Existential Theory of the Reals, Art Gallery Problem, Discretization}

\publicationdata{vol. 25:2 }{2023}{18}{10.46298/dmtcs.9225}{2022-03-18; 2022-03-18; 2023-01-09; 2023-05-29; 2023-08-25}{2023-08-28}

\begin{document}
\maketitle

\begin{abstract}
    Given a closed simple polygon $P$, we say two points $p,q$ 
    see each other if the segment $\textrm{seg}(p,q)$ is  contained in $P$.
    The art gallery problem seeks a minimum size set $G\subset P$ 
    of guards that sees $P$ completely.
    The only currently correct algorithm to solve the art gallery problem
    exactly uses algebraic methods.
    As the art gallery problem is \ER-complete,
    it seems unlikely to avoid algebraic methods, 
    for any exact algorithm, without additional assumptions. 

    In this paper, we introduce the notion of \emph{\visionstability}.
    In order to describe \visionstability consider an \emph{\enhanced} guard that can see ``around the corner'' by an angle of~$\delta$  or 
    a \emph{\diminished} guard whose vision is by an angle of $\delta$ ``blocked'' by reflex vertices.
    A polygon~$P$ is \visionstable for an angle $\delta$ if the optimal number of \enhanced guards to guard~$P$ is the same as
    the optimal number of \diminished guards to guard~$P$.
    We will argue that most relevant polygons are \visionstable.
    We describe a {\em one-shot \visionstable} algorithm that computes an optimal guard set for \visionstable polygons using polynomial time and solving one integer program.
    It guarantees to find the optimal solution for every \visionstable polygon.
    We implemented an {\em iterative \visionstable} algorithm and show its practical performance is slower, 
    but comparable with other state-of-the-art algorithms. The practical implementation can be found at: \href{https://github.com/simonheng/AGPIterative}{\url{https://github.com/simonheng/AGPIterative}}. 
    Our iterative algorithm is inspired and follows closely the one-shot algorithm. 
    It delays several steps and only computes them when deemed necessary. 
    
    A \emph{chord} is a straight line within a polygon connecting two non-adjacent vertices.
    Given such a chord $c$ of a polygon, we denote by $n(c)$ the number of vertices visible from~$c$. 
    The {\em \chordwidth} ($\cw(P)$) of a polygon is the maximum $n(c)$ over all possible chords~$c$.
    The set of \visionstable polygons admit an FPT algorithm when parameterized by the \chordwidth.
    Furthermore, the one-shot algorithm runs in FPT time when parameterized by the number of reflex vertices.
\end{abstract}

\vfill

\begin{figure}[htbp]
\centering
    \includegraphics{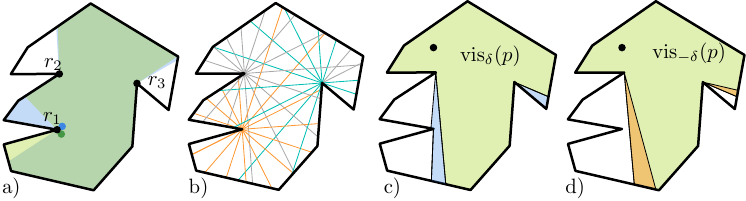}
    \caption{\small a) A small positional change largely influences how much visibility of the two guards is blocked by $r_1$, but only has a small effect on the way that $r_2$ or $r_3$ blocks the visibility of the guards. 
    b) Shooting rays from reflex vertices. 
    c) Enhanced visibility region. d) Diminished visibility region.
    }
\label{fig:SmallChanges}
\end{figure}

\newpage
\section{Introduction}

For most algorithmic problems, there is a \emph{practice-theory gap}.
That is, a gap between algorithms that perform best in practice and that perform best in theory.
A striking example is the euclidean traveling salesperson problem. It has long been known to be NP-hard, see~\cite{kisfaludi2020gap} for the most recent hardness results.
Despite those theoretical hardness results, in practice, new records of optimally solved large scale instances keep making the news~\cite{GermanTSP}.

In this work, we narrow this gap for the art gallery problem.
We present several closely related algorithms.
Some of them have provable performance guarantees under some
mild assumptions.
Other algorithms perform comparably well to the best state-of-the-art practical algorithms on the art gallery problem, however at the cost of losing these theoretical performance guarantees. 

\paragraph*{Motivation.}
The theoretical study~\cite{Beyond-Worst-Case} of any algorithm in general has three main motivations:
\\\textbf{ Prediction:} Given a specific algorithm, we want to predict how well it will perform.
    This can help us to decide if we want to implement it in the first place.
\\\textbf{ Explanation:} Assuming that we know how well an algorithm performs,
    we want to find an underlying reason for its performance.
\\\textbf{ Invention:} Can we design better algorithms in practice?

These goals are an important measure of success.
In an ideal world, theory and practice go hand-in-hand.
Theory researchers design and analyze algorithms and practically oriented researchers implement those algorithms and rigorously test them.
These tests help to adapt models, 
assumptions, and performance measures.
At the same time, practical researchers may find other
approaches fruitful and then in turn theoretical researchers need to find out why those 
approaches work.
This \emph{theory-practice cycle} should ideally lead to a very good 
theoretical understanding and very fast practical algorithms.
Unfortunately, as we will lay out, the study on the art gallery problem
has two almost independent lines of research. 
One solely theoretical and another purely practical one.
Thus, our theoretical study has limited value to predict,
 to explain, or to invent.
We believe that a solid theoretical understanding 
 will be also useful for experimental research.
Our main contribution is to  narrow this
gap and give a starting point to the theory-practice cycle.

\paragraph*{Related work.}
Let us start by considering practical research. 
Various researchers implemented algorithms and found the optimal solution on synthetic instances.
Among those algorithms with published code, the best can find the optimal solution of instances of up to~$500$ vertices {and sometimes even up to~$2500$ vertices}~\cite{tozoni, PracticalARTMasterFriedrich, PracticalARTbottino2011, PracticalARTkroller2012, PracticalARTcouto2011, PracticalARTbottino2008, PracticalARTcouto2008, PracticalARTamit2010, engineering, SoCGVideo}. 
All of these works rely on solving integer programs.
The idea is to discretize the problem and hope that the solution to the discretized problem also gives the optimal solution to the original problem.
To be more specific, the approach is to generate a \emph{candidate} set~$C$  and a \emph{witness} set~$W$, 
then compute the optimal way to guard~$W$ using the minimum number of guards in~$C$.
In an iterative manner, more candidates and witnesses are generated until the optimal solution is found. 
While these algorithms \emph{usually} find the optimal solution, 
there is a simple polygon on which these algorithms run \emph{forever}~\cite{abrahamsen2017irrational}, because 
the guards that make up the optimal solution for this specific polygon have irrational coordinates (we will refer to such a solution as irrational guards).
{We do know very few sufficient conditions under which these practical algorithms would give the optimal solution in a finite amount of time.
For instance, if the polygon is convex or witnessable (see below). 
However, those criteria are rarely met in practical instances.}

Let us continue by considering worst-case optimality.
We know that the art gallery problem is decidable using tools from real algebraic geometry. 
The idea is to encode guards by real numbers 
and use polynomial equations and inequalities to encode visibility~\cite{EfratH06}.
Currently, researchers working on solving polynomial equations repeatedly report that~$12$ is the maximum number of variables they could handle, using exact methods.
(The second author asked this at several conferences and workshops.)
Expressing the art gallery problem using polynomial equations has a considerable blow-up and needs existentially and universally quantified variables.
We think it is unlikely that those methods can be applied to decide if a polygon can be guarded with three guards,
but we leave this as an interesting open problem. 
To summarize, we would be surprised if these exact methods could even find an optimal solution of size~two for the art gallery problem.
As the art gallery problem is \ER-complete~\cite{ARTETR, stade2022complexity}, we know that methods from real algebraic geometry are unavoidable for any exact algorithm, without additional assumptions.
The \ER-completeness of the art gallery problem~\cite{ARTETR, stade2022complexity} is a very good explanation of why researchers were not able to find algorithms that avoid the aforementioned algebraic methods.
Furthermore, it explains why it is hard to prove that practical  discretization schemes work from a theoretical perspective.
To be precise, if there would be a discretization scheme with worst-case guarantees, then $\NP\neq \ER$.

To the reader not familiar with the \emph{existential theory of the reals} (\ER),
it may be insightful to get some background information.
It is defined analogously to \NP. 
Recall that a problem is in \NP if for every input there is a binary witness~$w\in\{0,1\}^*$
and a verification algorithm that runs on the word RAM in polynomial time.
We say a problem is contained in~\ER if for every input there is a \emph{real} witness~$w\in\R^*$ and a verification 
algorithm that runs on the \emph{real RAM} in polynomial time.
Note that the usage of witness in the context of complexity classes is very different from the usage of witness in the context of the art gallery problem~\cite{robustFramework}.
The complexity class~\ER is important as it gives a precise characterization of many important 
algorithmic problems.
Important algorithmic problems are related to graph drawing~\cite{AnnaPreparation,bienstock1991some,LindaPHD,cardinal2017recognition, AreasKleist, cardinal2017intersection, mcdiarmid2013integer}, the art gallery problem~\cite{ARTETR, stade2022complexity}, geometric packing~\cite{etrPacking}, linkages~\cite{abel}, polytopes~\cite{richter1995realization, NestedPolytopesER}, machine learning~\cite{etrNeurons,train-fully-neural-networks, Z92}, matrix factorization~\cite{Schaefer-ETR,   shitov2016universality, Shitov16a}, order types~\cite{shor1991stretchability, mnev1988universality} and various other topics~\cite{erickson2019optimal, garg2015etr,  kang2011sphere, 
Schaefer2010, schaefer2013realizability, deligkas2020square, SimplexER, GeoThicknessER, Hausdorf-UER, MS22, HMWW24}.
None of these algorithmic problems are known to be contained
in~\NP. 
The problem to show \NP-membership is that it is seemingly
impossible to describe a discrete witness of polynomial-size.
The complexity-class~\ER relates the issue for each of those algorithmic problems.
To be more specific, either all of those algorithmic problems admit a polynomial size witness or none of them do.
This also makes it difficult to discretize any of those problems
as any such discretization, usually, would also give rise to a 
polynomial sized witness.

Maybe the main approach to overcome the discretization problem for the art gallery problem was to 
consider variants.
Specifically, restricting guard placements to the vertices of the polygon discretizes the candidate set, and thus allows us to avoid this discretization problem, see~\cite{o1987art,PritamConstantFactor,ApproXKirkpatrick15, AlmostConvex,PerfectGraphApproach,ghosh2010approximation}. 
While many of these results are very intricate and innovative, they do not play an important role in practice.
We think that there are two main reasons for this.
In practice, the discretization problem is solved to a fairly large degree.
That is, the candidates and witnesses are usually sufficient to determine
an optimal solution, even if this cannot be proven, i.e., restricting
the guards to the vertices has only a small added benefit.
The second reason is that theoretical research focuses on the study of approximation algorithms,
which may not be so relevant in practice.

The study of approximation algorithms is mainly motivated by the
fact that the art gallery problem and many of its variants are \NP-hard~\cite{eidenbenz2001inapproximability,LeeLin86,SchuchardtH95,BonnetW1HARD}.
Thus, we cannot expect a polynomial-time algorithm, assuming P~$\neq$~\NP.
However, there are four important facts to consider when dealing with the practical performance of approximation algorithms:\\
    \textbf{Discretization:} While the concept of approximation relaxes the worst-case condition, it has not been shown to help sufficiently to find a nice discretization~\cite{BonnetM17Approx}. 
    (See below for a detailed discussion.)\\
   \textbf{ IP-solvers:} Once we have
    a discretization of the art gallery problem, IP-solvers often find the optimum
    for that discretized problem very fast in practice.
    (Often only $10\%$ of the total running time is spent on solving IPs~\cite{engineering}.)
    Thus, the aim of having a polynomial-time algorithm is not so important.\\
    \textbf{Visibility:} The real bottleneck in practical performance seems to be computing visibilities between candidates and witnesses. Approximation algorithms have to do these computations as well.
    Thus, they may just not be any faster.\\
    \textbf{Non-optimality:} By definition approximation algorithms do not
    give the optimal solution. 
    This may just be a too high price to pay.

\paragraph*{Contribution.}
In this paper, we introduce the notion of \visionstability.
We argue that most practical polygons are \visionstable.
Using this assumption, we give a theoretical solution to the 
discretization problem.
Based on this discretization, we develop the \emph{one-shot (\visionstable)} algorithm that is \emph{guaranteed} to find an optimal solution for every \visionstable polygon.
The algorithm takes polynomial preprocessing time and thereafter solves one integer program.
In particular, this shows that the art gallery problem is in~\NP for {simple} \visionstable polygons.
We refine the algorithm and present the \emph{iterative (\visionstable)} algorithm.
The iterative algorithm delays many steps of the one-shot algorithm.
Both algorithms are \emph{reliable}, in the sense that even if the input polygon is not \visionstable, the reported result is correct.
In order to make the iterative algorithm comparable
in performance to the state-of-the-art,
we implemented several practical improvements.
The downside of this approach is that we are not able to show theoretical performance guarantees for this practical version of the iterative algorithm.
On the other hand, these improvements make it so that our implementation runs as fast as the state-of-the-art algorithms. 

We believe that the concept of \visionstability contributed to all
three main goals, as mentioned above. 
It gives an \emph{explanation} of why the discretization
problem is solvable in practice.
It led to the \emph{design} of a new practical algorithm, that, as we will see, works quite differently from previous algorithms in many aspects.
And finally, it \emph{predicted} correctly that this algorithm is feasible,
which we verified in practice.

Let us emphasize here that there is still a large gap,
between theory and practice.
Our predicted running times are far worse than the 
observed practical running time.
Furthermore, in practice, we used many techniques for which 
it seems very hard to give a profound theoretical explanation
of why they work so well.
We hope our work contributes to a deepened theoretical understanding
of the art gallery problem.

\paragraph*{Theoretical Practical Work.}
{While the gap between theory and practice is indeed deep, we would like to emphasize that there are indeed connections between them that we are aware of.} 
Furthermore, there may be more examples that we are not aware of.
For instance, Ghosh's approximation algorithm~\cite{ghosh2010approximation} had also a large influence 
on designing heuristics and exact algorithms.
Another example is the algorithm that is capable to discretize the
art gallery problem, in case that it is witnessable~\cite{KnauerWitness}.
In that case the algorithm by Tozoni et al.~\cite{tozoni2013practical} is guaranteed to find the optimal solution.

\paragraph*{\Visionstability.}
Before we formally define the notion of \visionstability, which is the key concept of this paper, let us give some
basic intuition on discretizing the art gallery problem. 
The aim of the discretization process is to 
define a suitable \emph{candidate set} $C$ such 
that $ \opt[C]$, the optimum guard set  
restricted to $C$,  is ``pretty close'' to the actual optimum $\opt$.
After defining $C$ in a suitable fashion, we can
compute $\opt[C]$ by solving an integer program.
We know~\cite{BonnetM17Approx} 
that the grid $\Gamma = w\Z^2 \cap P$ with a small 
enough width $w$
is a good such candidate set in the sense that 
$|\opt[\Gamma]| \leq  10\, |\opt|$, under some mild general position assumptions.
Using smoothed analysis~\cite{ArxivSmoothedART,robustFramework}
and a suitable random model of perturbation, 
we even know that $|\opt[\Gamma]| = |\opt|$, with high probability.
In summary, a fine enough grid contains the optimal solution. {Under certain assumptions}, however, the size of the candidate set~$C$ is also important.
The grid~$\Gamma$ as described above is huge. 
Thus, computing $\opt[\Gamma]$ is 
infeasible in practice, 
making it very desirable to attain a candidate set $C$ with 
similar properties as $\Gamma$ and of polynomial size.

As a first step, we realize that a \emph{uniformly distributed} candidate set 
would either be too large or too coarse.
Thus, at some spots, we want the candidate set to be denser 
and at other places in the polygon, we want it to be more sparse. 
Let us say two 
guard positions $g,g' $ have almost the same visibility region.
Then, hopefully, it is not so important to keep track
of both positions, but keeping only one of the positions
is sufficient. 
However, if the visibility regions
of $g$ and $g'$ are very different, we should potentially 
include both of them in our candidate set $C$.
Another, almost trivial, observation is that a small movement
of $g$ towards $g'$ may dramatically change visibility regions, 
if $g$ and $g'$ are close to a \emph{reflex} vertex.
A reflex (or a concave) vertex is a vertex with an internal angle strictly greater than $\pi$ radians.
If $g$ and $g'$ are both very far 
away from all reflex vertices
their visibility regions change almost not at all, 
see Figure~\ref{fig:SmallChanges}~a).
Those two observations suggest that we may want to pick a candidate set that is denser closer to reflex vertices and more sparse farther away from reflex vertices.

This motivates the following approach. 
Shoot rays from every reflex vertex, such that the angle between any two rays is at most some 
given angle $\delta$, where $\delta$ is not too small.
This defines an arrangement~$\A$.
All intersection points of the rays within the polygon $P$
define our candidate set~$C$, 
see Figure~\ref{fig:SmallChanges}~b).
On an intuitive level, for every point $p\in P$, there is a candidate $c\in C$ such that the visibility regions of~$p$ and~$c$ are similar.
When $r$ denotes the number of reflex
vertices, we get a total number of rays upper bounded by
$O(\frac{r}{\delta})$. Thus, the candidate 
set has size~$O(\frac{r^2}{\delta^2})$.

We are now ready for the formal definition of \visionstability.
Given a simple polygon $P$ and a point $q$, the visibility region of $q$ is defined as $\vis(q) = \{x\in P: x \textrm{ sees $q$}\}$. 
Let $r$ be a reflex vertex of $P$ and we assume that $q$ sees $r$.
Then, given some $\delta>0$, we can define the \emph{visibility enhancing region} $A = A(q,r,\delta)$ as follows.
Rotate a ray $v$ with apex $r$ 
by some angle of $\delta$,
clockwise and counter-clockwise. 
At each time of the rotation, the ray $v$ defines a maximal segment inside~$P$ with endpoint~$r$.
In some cases, the segment is the single point $r$, see Figure~\ref{fig:AugmentationRegion}.
The region $A(r,\delta)$ is the union of all those segments.
\begin{figure}[tbp]
    \centering
    \includegraphics[page=12]{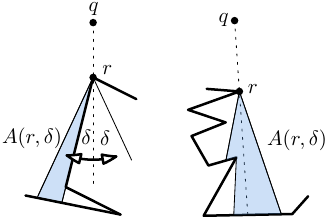}
    \caption{A ray is rotated around a reflex vertex $r$. It defines a region that is either added or removed from the visibility region.}
    \label{fig:AugmentationRegion}
\end{figure}
For some $\delta>0$, we define the $\delta$-\enhanced visibility region $\vis_{\delta}(q)$ of $q$ as $\vis(q)$ and, for every suitable reflex vertex, 
we add the region $A(r,\delta)$.
We define the $\delta$-\diminished visibility region $\vis_{-\delta}(q)$ of $q$ as $\vis(q)$ after we remove the regions $A(r,\delta)$,  for every applicable reflex vertex $r$.
To be precise, we define $\vis_{\delta}(q)$ to be a closed set, both for $\delta>0$ and $\delta\leq 0$.
Given a polygon $P$, we say that the point set{~$G$} is \emph{$\delta$-guarding} $P$ if 
$\bigcup_{g \in G} \vis_{\delta}(g) = P$.
We denote by $\opt(P,\delta)$ the size of the minimum 
$\delta$-guarding set.
For brevity, we denote $\opt(P,0)$, merely by $\opt(P)$ or, if $P$ is clear from the context, by~$\opt$.
We say that a polygon $P$ is \emph{\visionstable} or
 equivalently has \visionstability $\delta>0$,
 if $\opt(P,-\delta) = \opt(P,\delta)$.
Note that for $\delta' > \delta > 0$, it holds that $P$ has \visionstability $\delta'$ implies that $P$ also has \visionstability $\delta$. 
Thus, the smaller the \visionstability the larger we expect the candidate set to be.
To avoid confusion, at no time are we interested in actually computing~$\opt(P,x)$, for any value $x\neq 0$. 
The notion of \visionstability is purely a theoretical concept in 
order to formulate assumptions on the underlying polygon.

\begin{figure}[tbp]
    \centering
    \includegraphics[page=13]{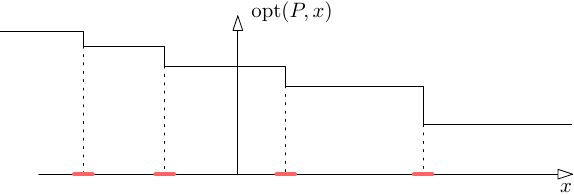}
    \caption{On the $x$-axis, we have the value by which we either diminish or enhance guards. On the $y$-axis we display the optimal number of guards. The function $\opt(P,x)$ only takes discrete values and is monotonically decreasing. Thus, it has a finite number of breakpoints.}
    \label{fig:SmoothArgument}
\end{figure}

Here, we will give a summary of the justification of \visionstability in the context of solving the art gallery problem. 
First, without any assumption, we cannot avoid algebraic methods unless~$\NP = \ER$.
Furthermore, there is an argument to be made related to smoothed analysis.
Consider $\opt(P,x)$ as a function~$f$ of $x$, see Figure~\ref{fig:SmoothArgument}. Clearly,~$f$ takes only
a discrete number of values, by definition. 
Furthermore, the function is monotonically decreasing, as 
$\vis_x(p) \subseteq \vis_{x'}(p)$, if $x\leq x'$.
Thus, the function~$f$ has only a finite number of breakpoints.
If a breakpoint happens to be at zero then~$P$ is not \visionstable.
Intuitively, this seems unlikely.

\paragraph*{Discretization.}

Using \visionstability, we exhibit a candidate set~$C$ of polynomial size. The idea is to use the vertices of arrangement~\A from Figure~\ref{fig:SmallChanges}~b). 
For technical reasons, we will not use the arrangement~\A, but a refinement of it.
Note that the smaller the \visionstability, the weaker the assumption on the underlying polygon, and thus the larger the required candidate set. 
\begin{restatable}[Candidate Set]{theorem}{Candidates}
    \label{thm:Candidates}
     Given a {simple} polygon with \visionstability~$\delta$ and~$r$ ($>0$) reflex vertices, it is possible to compute a candidate set~$C$ (of size $O(\frac{r^4}{\delta^2})$) in polynomial time on a real RAM.
     The candidate set~$C$ contains an optimal solution.
\end{restatable}

Although some of the details of the proof are tedious and involved, we see the main contribution on a conceptual level.
It is surprising to us that the \visionstability was 
not formulated earlier.
In particular, regarding the popularity of the art gallery problem within computational geometry 
and the problematic lack of algorithmic results.
Let us stress that while we use augmented and diminished visibilities as abstract concepts, 
we are only interested in solving the original art gallery problem.

Let us give an intuition of the proof idea of Theorem~\ref{thm:Candidates}. 
As mentioned before, we use all the vertices of a refinement of the arrangement~\A, as in Figure~\ref{fig:SmallChanges}~b), as our candidate set~$C$.
Consider a minimum size set~$G_0$ that is $(-\delta)$-guarding~$P$.
Replace each guard~$g\in G_0$ by a guard $g'\in C$ that is ``close by''. This gives a new solution $G_1\subseteq C$ of equal size.
Using that $G_0$ is $(-\delta)$-guarding, we can show that $G_1$ is guarding~$P$ in the usual sense.
As $\opt(P,-\delta) = \opt(P,0) = \opt(P,\delta)$, we can conclude that~$G_1\subseteq C$ is of minimum size. 
The technical demanding part of the proof is to show that for every point $p\in P$ there exists a point $c\in C$ such that $\vis_{-\delta}(p) \subseteq \vis(c)$.

Theorem~\ref{thm:Candidates} answers an open question posed by several authors of related works~\cite{PracticalARTamit2010, tozoni2013practical, Quest-Tozoni, engineering}.  For instance, De~Rezende~et~al.~\cite{engineering} states 
``\emph{Therefore, it remains an important open question whether there
exists a discretization scheme that guarantees that the algorithm always converges~[...].}''

In the next paragraph, we discuss how the discretization scheme leads to a correct algorithm that avoids algebraic methods.

\paragraph*{One-Shot \visionstable algorithm.}
Note that in practice there does not exist an algorithm which can be used to compute whether or not a polygon has \visionstability~$\delta$. Therefore, it is important that our algorithms work correctly even if the underlying polygon is not \visionstable. 
Specifically, our algorithms will either compute the optimal solution or report that the input polygon had lower \visionstability than specified by the user. 
We say our algorithms are \emph{reliable}, as they never return an incorrect answer.

\begin{restatable}[One-Shot \visionstable Algorithm]{theorem}{OneShot}
     \label{thm:OneShotAlgo}
    Let $P$ be an $n$ vertex {simple} polygon, with~$r$ reflex vertices. 
    We assume that a suggested value for~$\delta$ is given as part of the input.
    Then the one-shot algorithm has a preprocessing time of $O(\frac{r^{8}}{\delta^4}\log n + n\log n )$ on a real RAM and additionally solves exactly one integer program.
    The algorithm either returns the optimal solution or reports that the given suggested value for~$\delta$ is incorrect.
\end{restatable}

This is the first algorithm for the classical art gallery problem that 
avoids using algebraic methods and gives an exact solution.
Note that we do not improve over algebraic methods in terms of worst-case time complexity.

The core idea of the algorithm is to utilize the candidates 
from Theorem~\ref{thm:Candidates} and use them to build 
a set-cover instance, which can then be solved using integer programming.

Let us point out that next to vertex-candidates, we also use faces as candidates. This is a distinct feature of our algorithm. All previous algorithms used only points as candidates.
In this way, we can easily check that we have not missed a better solution. 
Say the algorithm returns a solution~$G$ of size $k$
and suppose for the purpose of contradiction that there would be a smaller 
solution~$G'$ of size~$k'<k$. 
Pick for each $g\in G'$ a face containing~$g$. 
(If~$g$ lies on an edge or a vertex of~\A make an arbitrary choice.)
This defines a set $\F$ of faces with $|\F| = k'$.
Now we arrive at a contradiction, as $\F$ is a valid guarding of the polygon~$P$.
The algorithm primarily minimizes the number of used guards (vertex or face guards), but among the minimum size solutions, it prefers vertex guards over face guards. 
We say that using vertex guards is a \emph{soft constraint}.
We will show later that if every guard is a vertex guard and everything is guarded, 
we have found an optimal solution. 
The idea of hard and soft constraints is that soft constraints are only relevant once all hard constraints are satisfied.

Similar to previous algorithms, we also use witnesses.
In the context of the art gallery problem, we require 
only that all the witnesses are seen, instead of the entire polygon.
The hope is that the computed guards will also see the entire polygon.
This is a second important step to discretize the problem.
In our case, the witness set~$W$ will be all \emph{faces} and
vertices of some arrangement~\A. 
This is a second feature that makes our algorithm unique.
As all the faces of~\A are covering~$P$, seeing all faces guarantees that~$P$ is completely seen.
We impose the \emph{hard constraint} that all point-witnesses are seen and the \emph{soft constraint} that each face-witness must be seen by at least one guard.

Due to the hard constraints, we know that our solution is at most the optimal size (theoretically it could be smaller, as face-guards are more powerful than point-guards). 
If we see all face-witnesses, then we know that the guards are actually guarding~$P$. Furthermore, if all guards are points, we know that we have found a minimum size point guard set.
In case that one of the soft-constraints is violated, we will be able to deduce that the underlying polygon was not \visionstable, with the suggested value~$\delta$ given in the input.
The last statement is the technically demanding part of the proof.
One of the central concepts of the proof is the \emph{\power of a face}. It measures how small a face 
is while taking into account the distance to reflex vertices.

\vspace{4pt}

\emph{For the remainder of the paper, whenever we refer to a candidate, witness or a guard, we could mean either a point or a face, if not further specified.}

\vspace{4pt}

{The choice of the \visionstability $\delta$ is somewhat arbitrary and is left to the user.
It is not trivial to upper bound the worst possible $\delta$. 
Clearly, if the polygon is convex then $\delta$ is unbounded, as there is no reflex vertex and no visibility polygon will be altered.
If there is at least one reflex vertex then 
$\delta$ must be smaller than $2\pi$ trivially. 
However, due to the way that the algorithm works internally, we think that $\pi/2$ is maybe the largest reasonable choice. 
}

\paragraph *{Parameterized algorithms.}
As the size of the integer program of the one-shot algorithm only depends on $r$ and $1/\delta$, it exhibits an FPT algorithm, with respect to the number of reflex vertices~$r$, for every fixed~$\delta$.
The most natural parameter for the art gallery problem is the solution size. 
As the art gallery problem is W[1]-hard, when parameterized by the solution size~\cite{BonnetW1HARD},
 research focused on other parameters~\cite{agrawal2020parameterized, AlmostConvex, agrawal2020parameter, ashok2019efficient, khodakarami2015fixed, khodakarami20171, TerrainConflictFreeFPT}. 
Specifically, Agrawal~et al.~\cite{AlmostConvex} described an elegant FPT algorithm for the art gallery problem.
They considered three variants of the art gallery problem defined by restricting guard positions and the part of the polygon that needs to be guarded.
By considering the number of reflex vertices as the parameter,
they gave a positive answer to a question posed by Giannopoulos, for those variants.
``\emph{Guarding simple polygons has been recently shown
to be W[1]-hard w.r.t. the number of (vertex or point) guards. 
Is the problem
FPT w.r.t.\, the number of reflex vertices of the polygon?}''~\cite{PanosQuestion}.
We answer the same question, with respect to \visionstable polygons and the classic variant of the art gallery problem.

\begin{restatable}[Reflex-FPT Algorithm]{corollary}{ReflexFPT}
\label{cor:ReflexFPT}
    Given a {simple} \visionstable polygon, with any fixed \visionstability.
    The one-shot algorithm is FPT with respect to the number of reflex vertices.  
\end{restatable}

\paragraph*{Practical algorithm.}

Although the one-shot algorithm does not require algebraic methods and only polynomial preprocessing time, it is still way too slow to be considered practical. 
Note that the performance bottleneck is not solving the integer program, but computing visibilities.
As a next step, we develop the \emph{iterative \visionstable} algorithm.
It is practical and follows the basic principle of the one-shot algorithm.
Ideally, we would also like to show provable performance guarantees for the iterative algorithm.
Note that the statement that an algorithm is both practical and has theoretical guarantees must be taken with caution.
Once we have a \emph{practical} algorithm~\Practical
and a \emph{theoretical} algorithm~\Theoretical,
we can easily get a third algorithm~\Both that 
has \emph{both} properties as follows.
Run algorithm~\Practical within the running time bound of~\Theoretical.
If~\Practical does not return a solution abort and run~\Theoretical.
Here, \Theoretical serves as a \emph{safe guard} for~\Practical.
Clearly \Both performs as well in practice as~\Practical and has the theoretical bounds of~\Theoretical. 
Thus, when we say that we show theoretical performance guarantees of some algorithm, we should really ask ourselves, if we show those guarantees for the algorithm that we actually use or for some safe guards that aren't ever used in practice.

The core idea of the iterative algorithm is as follows.
We start with a very coarse arrangement~\A. 
Using the faces and the vertices of~\A, we define a candidate set $C$ and a witness set~$W$.
We compute which candidates see which witnesses.
This enables us to build an integer program, that tries to find a minimum guard set $G\subseteq C$, as described for the one-shot algorithm. 
Note that vertices and faces may serve as guards and some face-witnesses may be unguarded.
As a secondary objective function, the integer program tries to minimize the number of face-guards used in $G$ and the number of unseen face-witnesses.
If the guard set $G$ contains only point guards 
and sees all face-witnesses, the iterative algorithm reports the optimal solution.
Otherwise, it refines the arrangement~\A and goes to the next iteration.

\begin{figure}[p]
    \centering
    \includegraphics[width=0.94 \textwidth ]{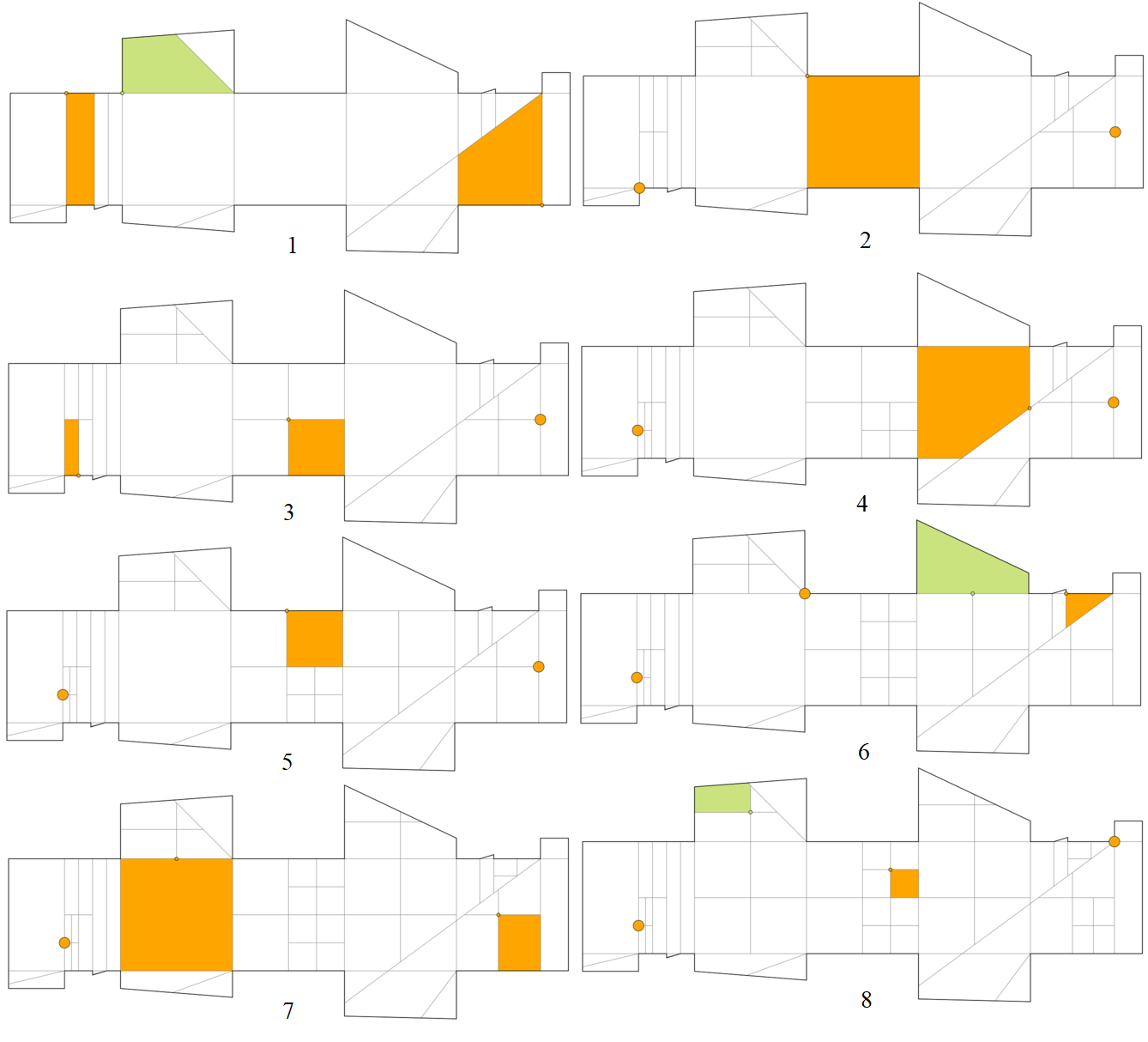}
    \caption{The first 8 iterations of the iterative algorithm on the Irrational-Guard polygon}
    \label{fig:first8p5}
\end{figure}
 
 \paragraph*{Irrational-Guard Polygon.}
In Figure~\ref{fig:first8p5} we show the first $8$ iterations of the Irrational-Guard polygon~\cite{abrahamsen2017irrational}. 
The orange points and faces represent point- and face-guards in the intermediate solution. The green faces represent faces not fully seen by the current candidate solution. Both orange and green faces are split in the next iteration. Note that for each of the orange and green faces, we draw a random vertex in the same colour,
to make also very small faces visible.

\paragraph*{Local complexity.}
    One of the bottlenecks is the large number of  possible visibilities between candidates and witnesses that we have to compute.
    Due to the low local complexity of the input polygons, most of those pairs are not seeing each other. 
    We exploit this by building a so-called \emph{\weakVisPolyTree}.
    In this tree, any point~$p$ in node $n(p)$ can see a point~$q$ in node $n(q)$, if the two nodes are siblings or in a parent-child relationship, see Figure~\ref{fig:WeakVisPolTree}.
    We give a detailed definition in Section~\ref{sub:WeakVisTree}.

Interestingly, this inspired a new structural parameter, which we call the \chordwidth. 
Given a chord~$c$ of~$P$, we denote by $n(c)$ the number of vertices visible from $c$. 
The {\em \chordwidth} ($\cw(P)$) of a polygon is the maximum $n(c)$ over all possible chords $c$.

We show that the art gallery problem is FPT with respect to the \chordwidth.
\begin{restatable}[Chord-Width-FPT]{theorem}{ChordFPT}
	\label{thm:FPT-chordwidth}
	Let $P$ be a simple polygon with \visionstability at least some fixed~$\delta$.
	Then there is an FPT algorithm for the art gallery problem with respect to the \chordwidth.
\end{restatable}
The core idea is to use dynamic programming along the \weakVisPolyTree,
similar to dynamic programming algorithms for tree-width.
The challenge here is to find bounds on the number of candidates per 
node.

 \begin{figure}[H]
    \centering
    \includegraphics[width=0.8\textwidth]{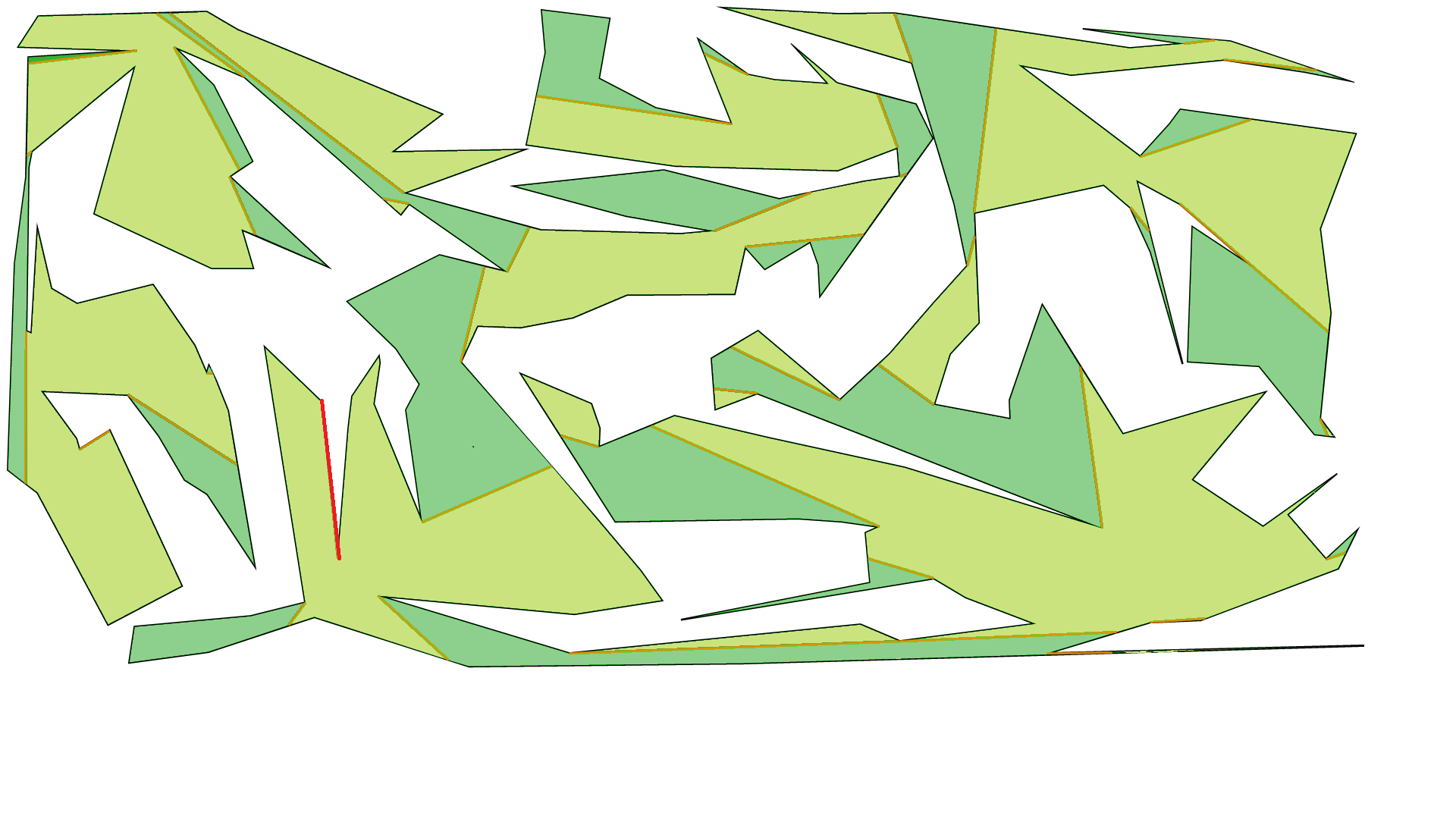}
    \caption{The polygon together with a \weakVisPolyTree. 
    The polygon has~$200$ vertices, but each node in the \weakVisPolyTree has only about~$20$ vertices. The red segment indicates the starting edge of the \weakVisPolyTree.
    }
    \label{fig:WeakVisPolTree}
\end{figure}

\paragraph*{Critical witnesses.}
Another practical idea is to reduce the total number of witnesses that we use. 
Instead of using all faces and vertices of~\A as witnesses, 
we only use some random selection, which we call \emph{critical witnesses}.
We use heuristics to update this critical witness set. 
In deciding whether to add a critical witness or not, we
are faced with some trade-offs. 
If we add very few critical witnesses, we have to make more loops until the IP solution returns a guard set that sees the entire polygon.
If we add too many critical witnesses, the set of critical witnesses grows unnecessarily fast and we have to compute many more visibilities.

\paragraph*{Visibility queries.}
One of the major bottlenecks at the beginning of this project was
the computation of weak visibility queries. That is, we often need to decide if a given face sees a given point or another face.
In a follow-up project~\cite{ConvexExpansion}, we developed a fast practical algorithm to compute weak visibility polygons.

\paragraph*{Losing performance guarantees.}
The main reason that the iterative algorithm does not have the same performance guarantees as the one-shot algorithm is as follows. It is possible that the iterative algorithm keeps splitting a certain face (and its children) many times, only to conclude much later that it was misled.
It could have found a solution much earlier by splitting
one of the larger faces.
In practice, it seems usually a very good idea to split the faces that were selected by the Integer program solution.
Especially, if those faces are small, we made progress and avoided usually unnecessary splits of big unimportant faces.
It is easily possible to design an algorithm in a way that it will not split too small faces, before also splitting occasionally bigger faces, that might be useful.
In this way, we are still able to ensure theoretical performance guarantees, see Theorem~\ref{thm:iterative}.
However, this theorem adds little to the goals of explanation, prediction, and invention.
Quite the opposite, this analysis suggests that faces should be split in a way that is harmful to practical performance.
In this paper, we describe several different versions of the iterative algorithm. 
When we use the safeguard version of the algorithm, we get the following theorem.

\begin{restatable}[Iterative Algorithm]{theorem}{ThmIterative}
     \label{thm:iterative}
    Let $P$ be {a simple} $n$ vertex polygon, with  \visionstability $\delta$.
    Then, the iterative algorithm returns the optimal solution
    to the art gallery problem. 
    It has a running time of $(\frac{n}{\delta})^{O(1)} + T$ per iteration and takes at most $(\frac{n}{\delta})^{O(1)}$ iterations.  Here $T$ denotes the time it takes to solve one integer program.
\end{restatable}

\paragraph*{Experimental results.}
We implemented and tested the iterative algorithm with a 64-bit Windows 10 operating system, an 8-core Intel(R) Core i7-7700HQ CPU at 2800 Mhz and 16 GB of main memory.
The practical implementation makes heavy use of version 4.13.1 of CGAL~\cite{cgal:eb-20a}. 
The IP solver used was IBM ILOG CPLEX version 12.10~\cite{cplex}.

We compared our implementation directly with the algorithm from Tozoni et~al.~\cite{tozoni}, as this is the currently best algorithm for which there was freely accessible code available.
{The algorithm from Tozoni et~al. was tested on the same machine described above, but on a Linux Mint operating system (using a dual-boot set-up).}
In order to get a deeper understanding, we analyzed various aspects of the running time.
One of them is the distribution of the running time w.r.t.~different subroutines.
Furthermore, we studied the influence of \visionstability on the running time. 
We also study the effect of our speed-up methods.
Lastly, we study the iterative algorithm on the irrational-guard polygon~\cite{abrahamsen2017irrational}.

\paragraph*{Comparison.}
We tested our algorithm and the algorithm of Tozoni~et al. on $5$ sets of~$30$ polygons with 60, 100, 200, and 500 vertices.
The results can be seen in Table~\ref{tab:running-times1}. 
We see that, except for size $60$, our algorithm is slightly faster. 
However, we want to point out that comparing those running times should be taken with a grain of salt.
Both algorithms rely on a software environment that is not identical. To some degree, the differences in the running time may come from performance differences from this software environment.
First, note that Tozoni~et al.~implemented their algorithm in Linux,
whereas we implemented our algorithm in Windows.
Interestingly, CGAL runs between factor~$2-3$ faster on Linux compared to Windows~\cite{CGAL-Bad-Windows}.
Secondly, our algorithm has to compute visibilities of faces instead of point visibilities. 
{However, the algorithm of Tozoni~et al. is purely sequential, while we perform these visibility computations in parallel.} 
Thirdly, while testing the implementation of Tozoni~et al., we could not use the best available IP solver, which might skew the results. 
{We used the freely available GLPK solver, while the algorithm of Tozoni~et al. was reported to have much better results with the XPRESS solver (\cite{tozoni} reports speed-ups of 2.56).}

Overall, from these experiments, our algorithm seems faster for larger polygons, but does not make a very large improvement. {However, for the above reasons, the comparison was not entirely fair.}

We left out a comparison to~\cite{engineering}, as their code is not available to us.

\def\arraystretch{1.15}
\begin{table}[htp]
\centering
\begin{tabular}{|c|c|c|c|} 
    \hline
\multirow{2}{*}{\textbf{Sizes}} & \multicolumn{3}{c|}{\textbf{Average time (s)}} \\ \cline{2-4}
 & \textbf{Tozoni et al.} & \textbf{Tozoni et al. (Our hardware)} & \textbf{The Iterative Algorithm}  \\ \hline
60 & 0.26 & 0.18 & 0.39 \\
100 & 0.94 & 0.68 & 0.52  \\
200 & 3.77 & 2.54 & 2.02 \\
500 & 35.04 & 22.34 & 18.2  \\\hline
\end{tabular}
 \caption{A comparison of the iterative algorithm without safe guards with the results from Tozoni et al.~\cite{tozoni}, both the results reported by Tozoni et al. themselves~\cite{tozoni} and results found using their implementation on our hardware. Tests were ran on 150 polygons, but our algorithm could not find the optimal solution for one polygon ({Polygon \#25 in the set of polygons of size 100}) within the time limit, so the times of 149/150 polygons are displayed in this table. 
}
    \label{tab:running-times1}
\end{table}

\paragraph*{CPU distribution.}
We analyzed the distribution of the CPU time of the iterative algorithm. The results are shown below in a pie-chart in Figure~\ref{fig:distrib-chart1}. We see that solving integer programs is the dominating factor of the running time. This shows that to improve the running time of the algorithm, we must reduce the total number or the size of the IPs. Alternatively, we can optimize the IP solver, or perhaps experiment with different IP solvers.

\begin{figure}[H]
    \centering
    \def\svgwidth{\textwidth}
    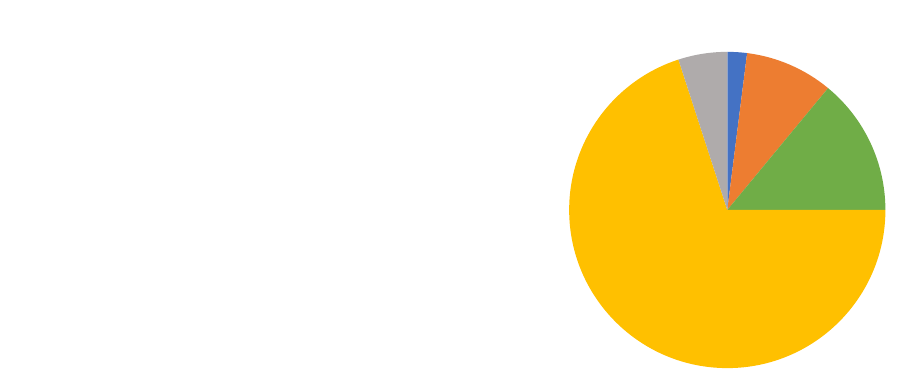
     \caption{The chart shows the CPU distribution of the Iterative Algorithm implementation for solving $30$ polygons of size $200$. Each slice in the chart represents the total CPU time spent on that part of computation for all 30 polygons.}
    \label{fig:distrib-chart1}
\end{figure}

It may appear strange that we spend so much energy on reducing the CPU time spent on speeding up visibility queries when the CPU usage is dominated by solving IPs. 
The reason is that before we implemented all of the improvements described before, the running time was dominated by weak-visibility queries.

In general, the CPU distribution needs to be regarded with a grain of salt.
It seems that there is usually one subroutine that dominates the running time.
Often conceptual improvements decrease the running time by several factors, 
which in turn makes a different subroutine appear to dominate the running time.
Thus, the fact that solving IPs says more about the components that we optimized rather than which parts are inherently more difficult.
Therefore, we consider it a success that the running time is now dominated by solving IPs.

\paragraph*{\Visionstability.}
Unfortunately, we cannot measure the \visionstability
of a polygon nor approximate it. 
To get a vague idea of the influence of the 
\visionstability on the running time, we 
define the granularity of the subdivision at the
end of the iterative algorithm.
The granularity is related to the smallest face
in the final subdivision.
Interestingly, even with polygons of the same size,
the running times vary widely.
The large correlations between the
observed running times and the granularity (shown in Table~\ref{tab:Correlations1}) indicate
that \visionstability may have a significant influence on
the total running time.

\begin{table}[ht]
    \centering
    \begin{tabular}{|c|cccc|}
        \hline
         \textbf{Size} &  $60$ & $100$ & $200$ & $500$  \\
         \hline
        \textbf{Correlation} & $0.07$ & $0.3$ & $0.1$ & $0.6$ \\
        \hline
    \end{tabular}
    \caption{The correlation coefficients between the measured granularity  and the running time, computed per size. }
    \label{tab:Correlations1}
\end{table}

\paragraph*{\WeakVisPolyTree.}
To verify the amount of visibility queries that we save by using the \weakVisPolyTree, {we measured the characteristics of the weak visibility polygon trees in the experiments done with the iterative algorithm without safe guards}. Note that computing the \weakVisPolyTree requires the computation of weak visibility polygon. In practice, this was achieved by using an efficient new algorithm, about which a follow-up paper will be published~\cite{ConvexExpansion}.
The precise percentage highly depends on the type of polygon.
See Table~\ref{tab:weakvisstats-intro} for a detailed overview.
Interestingly, the number of reflex vertices per node of the weak visibility polygon grows much slower than the input size.
This indicates that \chordwidth may be a useful practical parameter to study for geometric algorithms in polygons.
\begin{table}[ht]
    \centering
    \begin{tabular}{|c|cccc|}
        \hline
         \textbf{Size} &  $60$ & $100$ & $200$ & $500$  \\
         \hline
         \textbf{Tree size} &
         14.2 & 23.0 & 46.3 & 115.0  \\
          \textbf{Largest polygon} & 20.5 & 23.3 & 26.2 & 28.4 \\
          \textbf{Largest number of reflex vertices} & 5.9 & 6.2 & 7.0 & 9.2 \\
        \textbf{Percentage of queries saved} & 16.7\% & 35.4\% & 63.5\% & 87.3\% \\
        \hline
    \end{tabular}
    \caption{We tested 30 input polygons from the AGPLIB library~\cite{art-gallery-instances-page} of four sizes. For each size class, we see the averages of characteristics of the \weakVisPolyTree{s}.  
  }    \label{tab:weakvisstats-intro}\end{table}

\paragraph{Pseudocode.}
{The algorithm(S) consists of many components, it may help to have pseudocode to guide the reader. 
We start with the pseudocode of the one-shot algorithm.
\begin{center}
\includegraphics[page=2]{figures/pseudocode.pdf}    
\end{center}

The iterative algorithm goes into a loop that alternates
between geometric computations and solving an~IP.
\begin{center}
\includegraphics[page=3]{figures/pseudocode.pdf}    
\end{center}
}

\paragraph*{Irrational guards.}
We show that the implementation of our algorithm provides a rapidly improving solution even for polygons that are not \visionstable.
Specifically, Abrahamsen~et~al.~\cite{abrahamsen2017irrational} introduced a small and simple polygon
which requires irrational guards for an optimal guarding using point-guards.
Although the iterative algorithm avoids irrational numbers, it still returns a guard set~$G_i$ for each iteration $i=1,2,3,\ldots$. 
Recall that $G_i$ consists of faces and points.
As we know the optimal solution $G^*$, we can compute $d_i = d(G^*,G_i)$,
see Figure~\ref{fig:irrational-convergence-intro}.
{(Note that $d(G^*,G_i)$ denotes the Hausdorff distance between $G^*$ and $G_i$. 
In order to compute it, it is sufficient to compute the distance of each pair of vertices of $G_i$ with $G^*$.)}

\begin{figure}[H]
    \centering
    \def\svgwidth{\textwidth}
    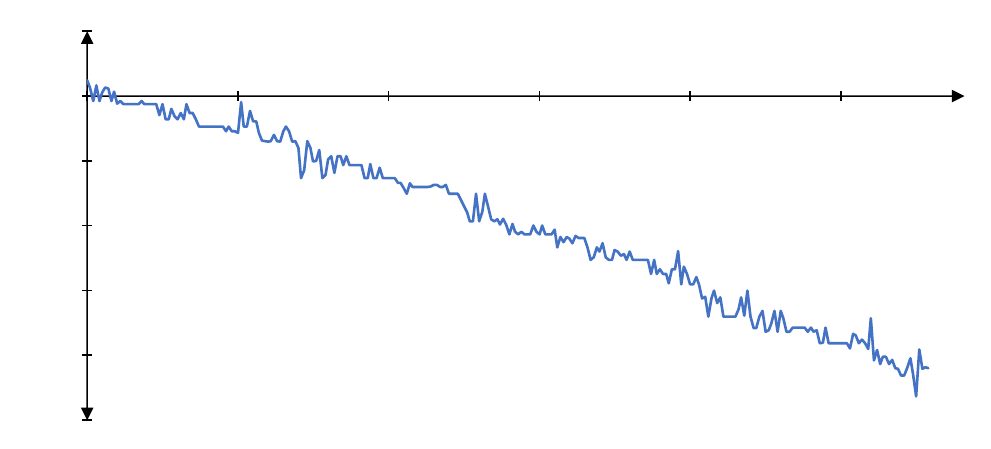
     \caption{The iterative algorithm based on the notion of \visionstability reports a sequence of solutions. The Graph  shows on the $x$-axis the iterations from~$1$ to about~$300$ and on the $y$-axis, the $\log_{2}$ of the Hausdorff distance to the optimal solution.  
     }
    \label{fig:irrational-convergence-intro}
\end{figure}

\paragraph*{Future research.}
The vast majority of the work on the art gallery problem focused on variants of the classic question. 
There are almost no positive theoretical algorithmic 
results on the original art gallery problem, with some exceptions~\cite{belleville1991computing, EfratH06, BonnetM17Approx}.
We believe that the main reason for this focus on variants is the fact that the art gallery problem is inherently continuous, as is reflected by its \ER-completeness~\cite{ARTETR,stade2022complexity}.
Now that we arguably broke that barrier, we hope that more progress will be made on the original problem.

\begin{itemize}
    \item Can we adapt the algorithm to polygonal domains with holes?
    Here, the main bottleneck seems to be adapting the visibility queries to polygons with holes.
    \item Does the iterative algorithm always converge towards the optimal solution, even if the underlying polygon is not \visionstable?
    Our experimental results suggest that we converge exponentially fast 
    to an optimal solution. 
    It is intriguing to see if this also holds true in general. 
    \item One simple way to make the one-shot algorithm more 
    practical would be to find a smaller witness and candidate set.
    What is the smallest integer program that guarantees to give
    the optimal solution for \visionstable polygons?
    The bound we gave seems to have plenty of room for improvement.
    \item It would be interesting to test the algorithm on a wider range of polygons, like orthogonal polygons, von Koch polygons, and spike polygons, as described in previous work~\cite{engineering}.
    This would clarify if our algorithm is applicable to a wider range of polygons.
    \item {The implementation is lacking for larger polygons, as running times get very large. There are several avenues for improving the algorithm so that experiments with larger polygons may be conducted. }
    
\end{itemize}

\section{Vision Stability}
\label{sec:Visstability}
In this section, we define the notion of \emph{\visionstability},
and give a justification, why we believe that typical polygons are
\visionstable. At last, we will show that, under very specific circumstances, the visibility region of a point contains the visibility region of a face. 

\begin{figure}[H]
	\centering
	\includegraphics[page=1]{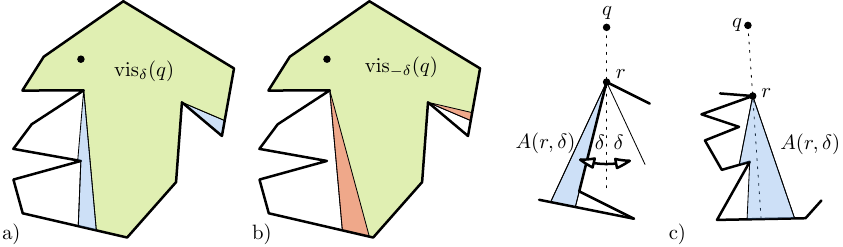}
	\caption{a) The visibility of the point $q$ in green. The blue region \enhanced the visibility. b) The red region diminishes the visibility of $q$.  c) A ray is rotated around a reflex vertex $r$. It defines a region that is either added or removed from the visibility region.}
	\label{fig:AugmentationRegionAppendix}
\end{figure}

\subsection{Definition}
In a nutshell, \enhanced and \diminished vision are artificial ways of vision where we can either ``look around a corner'' or are ``more blocked'' by a corner than we would expect.
The notion \visionstability entails that either enhancing or diminishing the vision does not change the optimal number of guards.

Given a simple polygon $P$ and a point $q$, the visibility region of $q$ is defined as $\vis(q) = \{x\in P: x \textrm{ sees $q$}\}$. 
Let $r$ be a reflex vertex of $P$ and we assume that $q$ sees $r$.
Then, given some $\delta>0$, we can define the \emph{visibility enhancing region} $A = A(q,r,\delta)$ as follows.
Rotate a ray $v$ with apex $r$ 
by some angle of $\delta$,
clockwise and counter-clockwise. 
At each time of the rotation the ray $v$ defines a maximal segment inside~$P$ with endpoint~$r$.
In some cases, the segment is the single point $r$, see Figure~\ref{fig:AugmentationRegionAppendix}~c).
The region $A(r,\delta)$ is the union of all those segments.

For some $\delta>0$, we define the $\delta$-\enhanced visibility region $\vis_{\delta}(q)$ of $q$ as $\vis(q)$ and for every suitable reflex vertex, 
we add the region $A(r,\delta)$.
We define the $\delta$-\diminished visibility region $\vis_{-\delta}(q)$ of $q$ as $\vis(q)$ after we remove the regions $A(r,\delta)$,  for every applicable reflex vertex $r$.
To be precise, we define $\vis_{\delta}(q)$ to be a closed set, both for $\delta>0$ and $\delta\leq 0$.

Given a polygon $P$, we say that $G$ is \emph{$\delta$-guarding} $P$ if 
$\bigcup_{g \in G} \vis_{\delta}(g) = P$.
We denote by $\opt(P,\delta)$ the size of the minimum 
$\delta$-guarding set.
For brevity, we denote $\opt(P,0)$, merely by $\opt(P)$ or, if $P$ is clear from the context, by~$\opt$.
We say that a polygon $P$ is \emph{\visionstable} or
equivalently has \visionstability $\delta>0$,
if $\opt(P,-\delta) = \opt(P,\delta)$.

\begin{figure}[H]
	\centering
	\includegraphics[page=2]{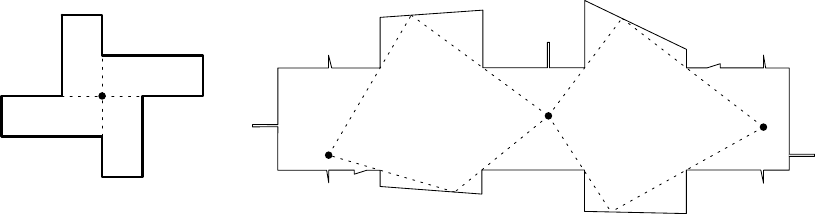}
	\caption{Left: The polygon has a unique guarding that relies on simple collinearities. Right: This polygon has a unique irrational guarding with three guards~\cite{abrahamsen2017irrational}.}
	\label{fig:NotVisionStable}
\end{figure}

\subsection{Justification}
\label{sub:justification}
In this section, we try to reflect on different aspects 
of the definition of \visionstability. 

Consider polygon $P_1$, in the left part of Figure~\ref{fig:NotVisionStable}.
It is not \visionstable.
That $P_1$ is not \visionstable
is clearly a weakness of the concept.
Note that this example relies on collinearities.
Many computational geometry papers assume general position of the underlying point set. Those assumptions are made for two reasons. 
The first is that collinearities are very unlikely, if we 
think about some random model that generates our point set.
The second reason is that collinearities 
can often be handled in practice, and the details are left out as they 
are very tedious, but add little to nothing to the underlying algorithmic concepts.
Our algorithms are also able to find the optimum 
for these types of situations in practice.

\vspace{0.1cm}

Let us now give another example of a polygon $P_2$ 
that is not \visionstable, see to the right of Figure~\ref{fig:NotVisionStable}.
See Abrahamsen et al.~\cite{abrahamsen2017irrational} for a detailed descritpion of the polygon. Here, we only note that all its vertices have rational coordinates.
Yet, the polygon $P_2$ has a unique optimal guarding with irrational guards.
Although it has plenty of collinearities, 
it can be modified not to have those collinearities. 
Polygons similar to $P_2$ exist and the irrational coordinates
required for an optimal guarding might need arbitrarily large
algebraic degree~\cite{ARTETR, stade2022complexity}.
We see the existence of such polygons as a strong argument for the need of additional assumptions, for instance about the \visionstability of the input polygon. 

In particular, those polygons 
show that \emph{without} additional assumptions 
algebraic methods are unavoidable.
Furthermore, they show that the art gallery problem is 
\ER-complete~\cite{ARTETR, stade2022complexity}.

\vspace{0.1cm}

As a third aspect for \visionstability, 
we want to argue that polygons similar to $P_2$ are very rare in practice.
Note that polygon $P_2$ was only found after four decades of research on the art gallery problem. 
It took the right approach, computer assistance and tedious trial and error to find this polygon. 
The geometric simplicity of the polygon may suggest that such a polygon would be relatively easy to construct, but this is far from true.
Thus, it is no surprise that as of this writing no second similar polygon is known yet.
There are polygons that are slight modifications of $P_2$ and the polygons that stem from the \ER-hardness proofs are the only known exceptions~\cite{ARTETR, stade2022complexity}.
The smoothed analysis by Dobbins, Holmsen and Miltzow~\cite{ArxivSmoothedART}
gives a theoretical argument why those polygons should be very rare in practice.

There is a fourth consideration. 
We will present an algorithm based on \visionstability in Section~\ref{sec:IterativeAlgo}.
As we will see in our test results in Section~\ref{sec:Tests}, the algorithm gives a sequence $(G_i)_{i\in \N}$ of guarding sets.
We observe that this sequence empirically approaches the actual optimal guarding for $P_2$ (see Figure~\ref{fig:irrational-convergence}). 
In other words, it could be that the algorithm described hence forth has the property that it always converges to the optimal solution, even if the underlying polygon is not \visionstable. 
It is a tantalizing open question, if this always happens or whether it is just a lucky coincidence that happened with~$P_2$.
Unfortunately, we don't know another polygon that we could potentially use to test this conjecture empirically. 

\begin{figure}[H]
	\centering
	\includegraphics[page=3]{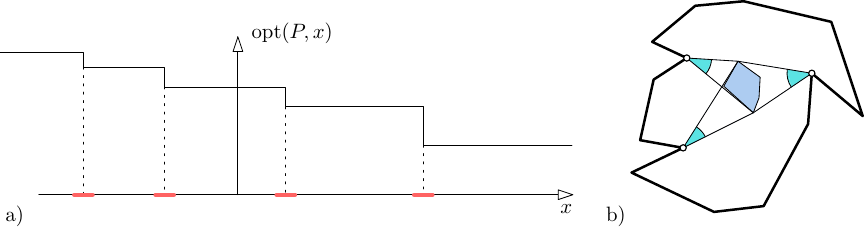}
	\caption{a) There are at most $n$ intervals of length $2\delta $ each. b) Illustration of the \power of a face.}
	\label{fig:BreakPoints}
\end{figure}

A fifth argument can be made that resembles ideas from smoothed analysis. Fix some polygon $P$ for the following discussion. Consider a world, where people would actually care for computing $\opt(P,x)$, for values of $x$ other than zero.
We define the event that 
$\opt(P,x-\delta) = \opt(P,x+\delta)$ by $E(x,\delta)$.
In this case, we say that $P$ is $x$-\visionstable, with \visionstability $\delta$.
In other words, the event $E(x,\delta)$ represents \visionstability for the task of computing $\opt(P,x)$, instead of $\opt(P,0)$.
In particular $E(0,\delta)$ corresponds to $P$ being \visionstable in the usual sense.
We show the following lemma.
\begin{lemma}\label{lem:visionstableprobability}
	Let $P$ be any simple polygon on~$n$ vertices.
	Choose $x\in [-1/2,1/2]$ uniformly at random. Then it holds that
	$\Pr(E(x,\delta)) \geq 1- 2\delta n$. 
\end{lemma}
This lemma says that any simple polygon $P$ is $x$-\visionstability with high probability, for $\delta = O(1/n)$.
We want to argue that as $E(x,\delta)$ has high probability, 
then on an \emph{intuitive level}, the same should be true for $E(0,\delta)$.
The reason being that $E(x,\delta)$ and $E(0,\delta)$ may be regarded as mathematically equally valuable. 

\begin{proof}[of Lemma~\ref{lem:visionstableprobability}]
	We consider the function \[f: [-1/2,1/2] \rightarrow \N, \quad
	x \mapsto \opt(P,x).\]
	See Figure~\ref{fig:BreakPoints}~a) for an illustration of~$f$.
	The function~$f$ is monotone as visibility regions only get larger with larger~$x$, thus it becomes easier to guard.
	We need at least one guard and most~$n$ guards, if~$n$ is the number of vertices of~$P$. This is because~$n$ guards are always sufficient, even for $(-1/2)$-diminished visibility. 
	Thus~$f$ has at most $n-1$ breakpoints. 
	Given some $x$, the event $E(x,\delta)$ is equivalent to the fact that there is no breakpoint within distance $\delta$ of $x$.
	Taking the union of all the intervals of length $2\delta$ centered at the breakpoints yields
	$\Pr(E(x,\delta)) \geq  1 - 2\delta n$. \qedhere
\end{proof}

\vspace{0.1cm}

Let us mention a sixth aspect of \visionstability.
Our practical algorithm computes on the fly a crude estimate
on the \visionstability of the underlying polygon.
This estimate has a large explanatory value 
in understanding practical running times, as we will explain further in Section~\ref{sub:PractRunning}.
Interestingly the running time of equally sized polygons vary easily within a factor of~$10$.
In many cases, it seems to be the case that polygons with a 
high running time have a low \visionstability, see Section~\ref{sub:PractRunning}.

In summary additional assumptions should always be treated with some amount of caution. We consider different aspects in favor of the usage of our new assumption.
We are looking forward to a lively discussion in the research community.

\subsection{\POwer of a Face}
One of the key concepts of our algorithms is the \power of a face. 
Assume, we are given a point $g$ contained in a convex set $f$.
Clearly, $f$ sees at least as much as $g$. 
If $g$ and $f$ see the same set of reflex vertices then we can think of $f$ as ``seeing around the corner'' a little bit more than $g$. 
The \power of $f$ is a simple to calculate approximation of the degree to which $f$ ``sees around the corner''.

For the following description consider Figure~\ref{fig:BreakPoints}~b).
For each reflex vertex $r$, we define the angle $\alpha(r,f) = \alpha(r)$,
as the angle of the minimum cone with apex $r$ that fully contains~$f$.
The \emph{\power}  of the face~$f$  ($\powerF(f)$ in short) is the maximum of all the $\alpha(r)$, for $r$ visible from $f$, i.e.,
$\powerF(f) = \max_{r\in \vis(f)} \alpha(r)$.

We denote by $\chord(a,b)$ the chord in $P$ that contains the 
two distinct points $a,b$, if it exists.
We define the set $\chord(A,B) = \{ \chord(a,b): a\in A, b\in B, a\neq b, \textrm{$a$ sees $b$}\}$.
Let $R$ denote the set of reflex vertices of~$P$.
We refer to a chord $c \in \chord(R,R)$ as a \emph{reflex chord}.
Reflex chords play a major role in proving correctness of our algorithms.
One of its first appearances can be seen in the following lemma.
We say a line~$\ell$ \emph{properly intersects} a convex set~$f$ if $\ell \cap \interior(f) \neq \emptyset$.
We denote by $\interior(f)$ the interior points of~$f$.
\begin{lemma}
	\label{lem:reflex-sees-all}
	Let $P$ be a simple polygon and let $f\subseteq P$ be a convex set that is not
	properly intersected by any reflex chord.
	Furthermore let $r$ be a reflex vertex that sees at least one
	interior point of $f$.
	Then it holds that~$r$ sees the entire convex set~$f$.
\end{lemma}
\begin{proof}
	As $\vis(r)$ is a closed region, we know that $r$ sees $f$ if and only if it sees all its interior points.
	For the purpose of contradiction assume that there is
	an interior point that is not seen by~$r$.
	Consider the boundary $b$ of $\vis(r) \cap f$.
	This segment $b$ is part of a reflex chord, which properly intersects $f$.
	This is a contradiction to the assumption.
\end{proof}

Given a convex polygon $f$ in some polygon~$P$, we denote by $\representative(f)$ a point of $f$ to represent the face.
In case that $f$ has a reflex vertex $r$, we set
$\representative(f) = r$. 
Otherwise, we choose 
the  lexicographically smallest vertex in~$f$ that is not a convex vertex of~$P$.
(In principle, we could choose any point in~$f$ arbitrarily.
We exclude convex vertices of~$P$, as this will allow
us to describe an FPT algorithm later most conveniently.
For concreteness, we pick the lexicographically smallest.
We pick a vertex of~$f$ as those will be later be part of our candidate set, as we will see later {in Section~\ref{sec:OneShotAlgo}}.)
If we are given a set $F = \{f_1,\ldots,f_k\}$ of convex polygons, then we define $\representative(F) = \{\representative(f) : f\in F\} = \{\representative(f_1) ,\ldots, \representative(f_k)\} $.

In a simplified way the following lemma states that the visibility of a convex polygon~$f$ can be replaced by the visibility of a single point~$p\in f$ with enhanced vision.
{However, there are two special cases that we need to be aware of. 
Firstly, we need to make sure that $f$ contains at most one reflex vertex. 
The reason is that every reflex vertex of $f$ may see much more than any other point of $f$.
We can deal with this by defining the representative of $f$ to be a reflex vertex if $f$ has at most one reflex vertex.
Secondly, it could be that $f$ sees only one boundary point of $f'$ and enhancing the visibility of $f$ will not make $f'$ visible at all. 
We can exclude this situation by demanding that $f$ sees an interior point of $f'$.
}
See Figure~\ref{fig:Face-Point-Degenerate} for an illustration of both special cases.

We define the visibility region $\vis(f)$ of a convex set $f$ as the union of the visibility regions of all its points, i.e., $\vis(f) = \bigcup_{q\in f} \vis(q)$.

\begin{lemma}[Face-Point-Replacement]
	\label{lem:Face-Point-Replacement}
	Assume the following conditions are met.
	\begin{itemize}
		\item We are given a simple polygon $P$.
		\item  Two closed convex regions $f,f'$ with $\powerF(f),\powerF(f') \leq \delta/2$ are given.
		\item Neither~$f$ nor~$f'$ are properly intersected by a reflex chord.
		\item The region~$f$ has at most one reflex vertex of~$P$ on its boundary.
		\item $\vis_\gamma(f)\cap \interior(f') \neq \emptyset$, for some $\gamma \in [-\delta,0]$.
		\item We denote $p = \representative(f)$.
	\end{itemize}
	Then it holds that $f'\subseteq \vis_{\gamma+\delta}(p)$.
\end{lemma}

Admittedly, the lemma is already technical.
Before we prove the lemma let us point out a simple corollary,
which is slightly less technical. 
Unfortunately, Corollary~\ref{cor:Vision-Enclosure} will not be strong enough for our purposes.
\begin{corollary}
	\label{cor:Vision-Enclosure}
	Let $P$ be a simple polygon.
	Assume the following conditions are met.
	\begin{itemize}
		\item A convex region $f\subseteq P$ with $\powerF(f) \leq \delta/2$ is given and no reflex chord properly intersects~$f$.
		\item The region~$f$ has at most one reflex vertex of~$P$ on its boundary.
		\item We denote $p = \representative(f)$.
		\item A number $\gamma \in [-\delta,0]$ is given.
	\end{itemize}
	Then it holds that $\vis_{\gamma}(f) \subseteq \vis_{\gamma+\delta}(p)$.
\end{corollary}
\begin{proof}[of Corollary~\ref{cor:Vision-Enclosure}]
	Construct an arrangement~\A such that every face in~\A has small \power 
	and is not split by a reflex chord.
	Let $f'$ be a face of \A. If $\vis_{\gamma}(f)$ contains an interior point of $f'$, then by Lemma~\ref{lem:Face-Point-Replacement} it holds that 
	$f'\subseteq \vis_{\gamma+\delta}(p)$. 
	This finishes the proof.
\end{proof}

\begin{figure}[H]
	\centering
	\includegraphics[page=8]{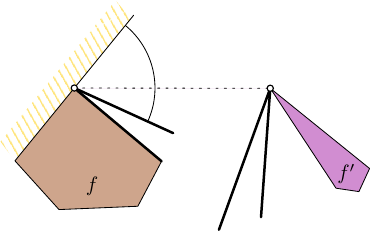}
	\quad
	\includegraphics[page=9]{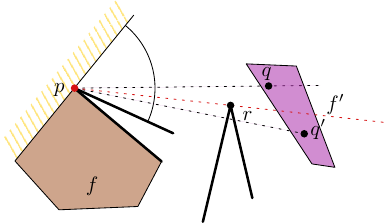}
	\caption{Left: It is not enough for one point of~$f$, seeing one point of~$f'$. We need the stronger assumption that an interior point is seen. Right: If the interior point is seen by the reflex vertex of $f$, then this reflex vertex sees $f'$ entirely.}
	\label{fig:Face-Point-Degenerate}
\end{figure}

\begin{proof}[of Lemma~\ref{lem:Face-Point-Replacement}]
	Let $q\in \interior(f')$ be an arbitrary interior point of $f'$. 
	We will show that $q \in \vis_{\gamma+\delta}(p)$.
	This is sufficient as visibility regions are closed by definition and thus $\vis_{\gamma+\delta}(p) \supseteq \interior(f') \Rightarrow \vis_{\gamma+\delta}(p) \supseteq f'$.
	Let $a\in f$ be a point that sees some other point $b\in \interior(f')$.
	Such a pair~$(a,b)$ exists by the assumption 
	$\vis_\gamma(f)\cap\, \interior(f')\neq \emptyset$
	and the fact that $\vis_\gamma(f)\subseteq \vis(f)$, for $\gamma\leq 0$.
	We handle first the special case that $a = p$ 
	and no other point of $f$ sees any point in the interior of~$f'$.
	\begin{claim}
		\label{clm:a-is-reflex}
		In this case, $a = p$ must be a reflex vertex of~$P$.
	\end{claim}
	\begin{proof}[of Claim~\ref{clm:a-is-reflex}]
		Consider the chord $\ell= \chord(a,b)$. The chord
		$\ell$ does not intersect~$f$ in another point~$a'$,
		as this would imply that $a'$ sees $b$ as well.
		Thus, $a=p$ is a vertex of $f$ and let us assume without loss of generality that $\ell$ is horizontal, and $f$ is below~$\ell$.
		There must be a reflex vertex on $\ell$ as otherwise, we could slide down $\ell$ and detect a new visibility pair  $(a',b')\in f\times \interior(f')$. 
		(Recall that $b$ is an interior point of~$f'$.)
		The chord~$\ell$ properly intersects~$f'$ and 
		thus contains exactly one reflex vertex~$r$.
		This reflex vertex~$r$ must be equal to $a=p$ as otherwise, 
		we can rotate~$\ell$ around~$r$ and get a new visibility pair,
		which does not exist by assumption.
		This finishes the proof of the claim.
	\end{proof}
	Due to Claim~\ref{clm:a-is-reflex}, Lemma~\ref{lem:reflex-sees-all} implies that $\vis_{\gamma+\delta}(p) \supseteq \vis(p) \supseteq f'$.
	
	\vspace{0.1cm}
	
	Now we consider the case that $a \neq p$.
	In particular, this implies that $a$ is not a reflex vertex of~$f$.
	Recall that we want to show for every point $q \in \interior(f')$ 
	that $q \in \vis_{\gamma+ \delta}(p)$.
	First note that if the shortest path from $p$ to $q$ 
	($\short(p,q)$) is the line-segment 
	$\seg(p,q) \subseteq \vis(p) \subseteq \vis_{\gamma+ \delta}(p) \subseteq P$,
	then we are done.
	Thus, we consider the case that $\short(p,q)$
	contains at least one
	reflex vertex $r$ in its interior.
	We will show the following claim.
	\begin{claim}
		\label{clm:ShortestPath}
		The shortest path $\short(p,q)$ makes at most one bend~$r$.
		There is no reflex vertex on $\short(p,q)$ after~$r$.
	\end{claim}
	Given a polygonal path~$w$, we say~$w$ makes a \emph{bend} at $v$,
	if~$v$ is an interior vertex of~$w$ and the interior angle at~$v$ is not equal to~$\pi$.
	Note that every bend on a shortest path must correspond to a reflex vertex of~$P$. The reverse does not hold. 
	It can be that there are reflex vertices on a shortest path that do not cause a bend.
	
	\begin{proof}[of Claim~\ref{clm:ShortestPath}]

		We consider two sub cases.
        First, we consider the case that $\short(p,q)$ 
		is not properly intersecting the line $\ell(a,b)$.
		In the second case, to be handled later, $\short(p,q)$
		properly intersects the line $\ell(a,b)$ once. 
		Note that $\short(p,q)$ cannot properly intersect $\ell(a,b)$ more than once. 
		Otherwise, we could shorten the shortest path by using part of $\ell(a,b)$.

\begin{figure}[H]
			\centering
			\includegraphics[page=11]{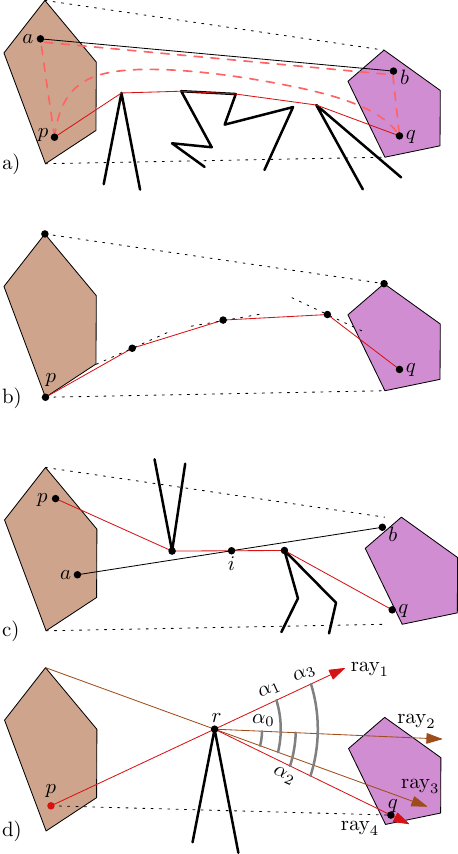}
			\caption{The brown face represents face~$f$ and the purple face represents face~$f'$.
				a) The shortest path from~$p$ to~$q$ is a convex chain.
				b) The tangent almost always intersects either~$f$, or~$f'$. 
				An exception is potentially the tangent through the point~$p$. 
				Recall that~$q$ is an inner point. 
				c) In this scenario the shortest path from~$p$ to~$q$ 
				consists of two convex chains. 
				d) Various rays and angles are illustrated.}
			\label{fig:One-Reflex}
		\end{figure}
  
		\textbf{Subcase~1. \ }
		Consider the path $pabq$. 
		This path is fully contained inside $P$, as $f$ and $f'$ 
		are convex and $a$ sees $b$.
		When we keep locally shortening this path, we 
		converge to the shortest path. 
		This mental experiment shows that the shortest
		path forms a convex chain, see Figure~\ref{fig:One-Reflex}~a).
        {(A polygonal path $P$ is a convex chain if $P$ is a subset of the boundary of the convex hull of $P$.)}

		For every point~$t$ on $\short(p,q)$, we can define a tangent $\tangent(t)$. 
		It holds for every point~$t$ that $\tangent(t)$ properly intersects~$f$ or~$f'$, except at the very beginning, see Figure~\ref{fig:One-Reflex}~b).
		Now, suppose for the purpose of contradiction that $\short(p,q)$ makes two (consecutive) bends. Say at reflex vertices $r_1,r_2$.
		Then  the line $\ell(r_1,r_2)$ is a tangent for some~$t$ and thus properly intersects either~$f$ or~$f'$.
		By assumption of the lemma, there is no reflex chord 
		that properly intersects either~$f$ or~$f'$. 
		The line $\ell(r_1,r_2)$ defines a reflex chord naturally.
		This finishes the proof of this subcase.
		
		\textbf{Subcase~2. \ }
		Now, we consider the second case{, when} the shortest path $\short(p,q)$ properly intersects $\ell(a,b)$
		in the unique point~$i$.
		Note that we can now decompose $\short(p,q)$ into $\short(p,i)$ and $\short(i,q)$. 
		By the same argument as in Subcase~1, we conclude that $\short(p,i)$ and $\short(i,q)$ are convex chains.
		We define the path $\alpha$ as the concatenation of $\short(p,i)$ with $\seg(i,b)$.
		In the same way we define $\beta$ as the concatenation of $\seg(a,i)$ and $\short(i,q)$.
		At first note that $\alpha$ and $\beta$ are fully contained in $P$ and are convex chains.
		Again, by the same argument as in Subcase~1, it holds that every tangent to $\alpha$ and $\beta$ are properly intersecting either~$f$ or~$f'$.
		(Again with the exception at the first segment of $\alpha$.)
		If all bends of $\short(p,q)$ are on the same side of $\ell(a,b)$, then we
		can literally repeat the argument from Subcase~1.
		So we assume that there is at least one bend on either side of $\ell(a,b)$.
		Let $r_1,r_2$ be the two consecutive bends of $\short(p,q)$ such that $\seg(r_1,r_2)$ properly intersects~$\ell(a,b)$, see Figure~\ref{fig:One-Reflex}~c).
		Now, note that $\chord(r_1,r_2)$ is a reflex chord that properly intersects either~$f$ or~$f'$, as $\ell(r_1,r_2)$ is a tangent to both~$\alpha$ and~$\beta$. This is a contradiction to the assumption that neither~$f$ nor~$f'$ are properly intersected by a reflex chord as stated in the assumption of the lemma. This finishes the proof of Claim~\ref{clm:ShortestPath}.
	\end{proof}
	
	Due to Claim~\ref{clm:ShortestPath}, we restrict our attention
	to the case that $\short(p,q)$ contains \emph{exactly} one bend, denoted by~$r$.
	Recall that~$r$ must be a reflex vertex of~$P$.
	By Lemma~\ref{lem:reflex-sees-all}, we know that~$r$ sees~$f'$ completely,
	as~$r$ sees~$q$, which is an interior point of~$f'$.
	For an illustration for the remainder of the proof, consider Figure~\ref{fig:One-Reflex}~d).
	We define the rays $\ray_1,\ray_2,\ray_3,\ray_4$, which will help us to prove our claim.
	All rays have apex $r$.
	The ray $\ray_1$ has the direction $r-p$.
	The ray $\ray_2$ describes the boundary of $\vis_{\gamma}(f)$ w.r.t.~$f'$. 
	The ray $\ray_3$ describes the boundary of $\vis(f)$ w.r.t.~$f'$.
	The ray $\ray_4$ has the direction~$q-r$.
	
	Now, we define the angles $\alpha_0, \alpha_1,\alpha_2,\alpha_3$ as follows. 
	The angle $\alpha_0$ is the angle between $\ray_2$ and $\ray_3$.
	By definition of diminished visibility regions $\alpha_0 = -\gamma$.
	(Recall that $\gamma$ is negative or zero by definition.)
	The angle $\alpha_1$ is the angle between $\ray_1$ and $\ray_3$.
	It holds that $\alpha_1 \leq \delta/2$, as $\powerF(f)\leq \delta/2$.
	The angle $\alpha_2$ is the angle between $\ray_2$ and $\ray_4$.
	It holds that $\alpha_2 \leq \delta/2$, as $\powerF(f')\leq \delta/2$.
	The angle $\alpha_3$ is the angle between $\ray_1$ and $\ray_4$.
	A simple calculation yields $\alpha_3 = \alpha_1+\alpha_2 - \alpha_0 \leq \delta + \gamma$. 
	This implies that $q \in \vis_{\gamma + \delta}(p)$ and finishes the proof of the lemma.
\end{proof}

\section{One-Shot \visionstable Algorithm}
\label{sec:OneShotAlgo}
In this section, we will describe an algorithm to solve 
the art gallery problem. We will show the following theorem.

\OneShot*

We call the algorithm from the previous theorem the \emph{one-shot \visionstable} algorithm.
In case that the algorithm does not return the optimal solution, the underlying polygon had \visionstability smaller than $\delta$.
In that case, we can simply half $\delta$ and repeat the algorithm.
In this way, we can find the optimal solution in the same polynomial running time solving additional $O(\log 1/\delta)$ integer programs in the worst case.
{To see this, assume we start with $\delta_0 = 1$ and the final value is $\delta_k = 2^{-k}$. 
Thus, we will have used $k = \log (1/\delta)$ halving steps.
}

We note that the integer program has size only dependent on $r$ and $\delta$, independent of the input size $n$.
Thus, for fixed $\delta$, it holds that the one-shot algorithm is fixed parameter tractable with respect to the number of reflex vertices.
\ReflexFPT*

The one-shot algorithm is not actually practical{, as demonstrated by our experiments in Section~\ref{sec:Tests}}. 
However, we will describe a second algorithm later, with similar theoretical performance guarantees, which we prove to be able to solve large instances of the art gallery problem practically.
Furthermore, we want to point out that, while solving an integer program may take an exponential amount of time in theory, in our practical experiments it played a minor role for the total running time (see Section~\ref{sec:Tests}). 
This corresponds to experiences that other groups of authors report as well~\cite{engineering}.

\paragraph*{Description of the one-shot \visionstable algorithm.}

\begin{figure}[H]
	\centering
	\includegraphics[page=3]{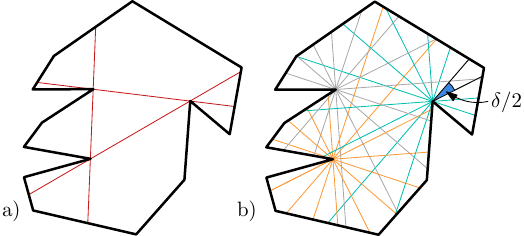}
	\caption{a) We include all chords of all pairwise visible reflex vertices.
		b) We shot rays from every reflex vertex. The angle between two consecutive rays is at most $\delta/2$.
		The two sets of segments define together the arrangement $\A$.}
	\label{fig:RayArrangement}
\end{figure}

As a first step, we construct an arrangement~$\A$.
We shoot rays from every reflex vertex and include all those rays
to~$\A$. 
The angle between any two consecutive rays is at most~$\delta/2$. 
We also add all reflex chords to $\A$.
At last, we subdivide all faces that are incident to more than one reflex vertex.
Note that we add at most {$O(r^2 + r/\delta) =  O(r^2/\delta)$} segments.
This finishes the description of~$\A$, see Figure~\ref{fig:RayArrangement}.
All faces of $\A$ have an \power of at most $\delta/2$.
Note that $\A$ has 
$O((r^2/\delta)^2 ) = O(r^4/\delta^2) $ vertices, edges and faces.
This does not account for the convex vertices of~$P$.

We define the candidate set~$C$ as the set of all 
vertices ($\vertex(C)$) {of \A}
that are not convex vertices of~$P$
and
all faces ($\face(C)$) of~$\A$.
We define the witness set~$W$ {as the} set of vertices ($\vertex(W)$) and a set of faces  ($\face(W)$).
For every face of \A, pick an interior point to define the \emph{vertex-witnesses}~$\vertex(W)$.
The \emph{face-witnesses} are composed of all faces of~\A.
We call candidates $c\in \face(C)$ \emph{face-candidate}s and
$c\in \vertex(C)$ \emph{vertex-candidate}s.
Similarly, we distinguish \emph{face-witness}es and \emph{vertex-witness}es.

Next, we compute for every candidate-witness pair $(c,w) \in C\times W$ whether the candidate sees the witness completely.
There are $O(r^8/\delta^4)$ such pairs in total. 
Using appropriate data structures, this can be computed in $O(\log n)$ time per pair~\cite{guibas1987linear}, where $n$ denotes the number of vertices of the underlying polygon.

We are now ready to build an integer program.
We call this integer program the \emph{one-shot~IP}.
For every candidate $c$, we create a variable $\const{c}$.
For every face-witness $w$, we create a variable $\const{w}$.
We denote variables with multi-character symbols.
As optimization function, we sum all the variables. 
However, the face-candidates receive a higher weight and the face-witnesses receive a lower weight. 
\[f = \sum_{c\in \vertex(C)} \const{c} \quad + \quad (1+\eps) \sum_{c\in \face(C)} \const{c} \quad + \quad \eps \sum_{w\in \face(W)} \const{w}.\]
We try to minimize the function $f$.
Choosing $1/\eps = |C|+|W| + 1$ is sufficiently small. 
The function $f$ counts primarily the number of guards that we use. 
The factor $(1+\eps)$ ensures that vertex-candidates are preferred over face-candidates.
The sum over the face-witnesses counts the number of face-witnesses that we do not see by a single guard in the solution.

For every witness $w$, we denote by $\VIS(w)$ the set of candidates that see $w$ completely.
We add the constraints
\[ \sum_{c \in  \VIS(w) } \const{c} \geq 1 \quad, \forall w\in \vertex(W).\]
For the face-witnesses we add the constraints
\[ \const{w} + \sum_{c\in \VIS(w) } \const{c}   \geq 1 \quad, \forall w\in \face(W).\]
Thus, the constraint to see all face-witnesses is relaxed. However, we need to pay in the objective by \eps, for every witness face that we do not see.
We also need to add constraints that ensure that every variable is in the set $\{0,1\}$.
The algorithm solves the integer program. 

The variables that are set to $1$, by the integer program give
rise to three sets. A set of \emph{vertex-guards} ($\const{c} =1,\  c\in \vertex(C)$), a set of face-guards ($\const{c} =1,\  c\in \face(W)$) and a set of unseen face-witnesses ($\const{w} =1,\  w\in \face(W)$).
If all guards that are found are vertex-guards 
and all the face-witnesses are seen by those guards
the algorithm reports those point-guards. 
Otherwise, 
the algorithm reports that it has not found an optimal solution.
In particular, this implies, as we will show, that the input polygon did not have \visionstability~$\delta$.

\paragraph*{Two stage integer program.}
While the IP that we describe here works theoretically, when we use a very small value of $\eps$, solvers may have trouble in practice.
In other words, they cannot solve the ILP.
Luckily, there is a simple trick to overcome this, by solving the IP in two stages.
The \emph{stage~$1$ IP} is the one-shot IP, with $\eps=0$.
In other words, the integer program is impartial towards face-guards and vertex-guards. 
Furthermore, it does not require any face-witness to be seen. 
Let $s$ be the value of stage~$1$~IP.
In the \emph{stage~2 IP}, we add the constraint that the total number of guards should be~$s$.
We define a new objective function~$f$ that counts all the face-guards and unseen face-witnesses. 
We aim to minimize the objective function.

Alternatively, we can also multiply the objective function with~$ 1 / \eps$. 
{In this way, the solver does not have to deal with small numbers, but instead with big numbers.
It seems IP solvers are more attuned to working with 
big integer numbers than small rational numbers.}

\paragraph*{Correctness.}
First, note that, by the description that we gave, the algorithm runs indeed in $O(\frac{r^{8}}{\delta^4}\log n )$ time and solves exactly one integer program as claimed.

\vspace{0.1cm}

We will first show that the result of the one-shot algorithm is \emph{reliable}. That is, whether the underlying polygon is \visionstable or not, we can trust that the algorithm returns the
correct result.
Thereafter, we will show that the algorithm will indeed report the optimal guarding for \visionstable polygons.

\vspace{0.1cm}

We show that if all guards $G$ found by the integer program are point-guards and all face-witnesses are seen, then the algorithm indeed returned the optimal solution.
Here, we are not assuming that $P$ is \visionstable.
As $G$ guards the entire polygon it holds that
$|G|\geq \opt$.
Let us now consider a set $F$ of point-guards of optimal size.
For every $p\in F$, we denote by $\face(p)$ the face of~\A that contains~$p$.
For some set $S$ of points, 
we denote by $\face(S) = \{\face(p): p \in S\}$.
Note that $\face(F)$ sees all vertex-witnesses.
It may be that there is one or several face-witnesses~$w$ that require several faces of $\face(F)$ to be completely guarded.
Such a face-witness~$w$, would be counted as unseen face-witness.
Still the one-shot IP could use exactly those face-guards $\face(F)$.
As the one-shot IP computed a set of point guards~$G$,
it implies $|G| < |F| + 1  = \opt + 1$.
Note that $\eps ( |G| + |W|) < 1$, by definition of~$\eps$.  
Thus, if the one-shot IP returns a solution using point-guards only and all face-witnesses are seen, we have~$|G| = \opt$.
This shows that the algorithm is reliable.

\vspace{0.1cm}

Now, we show that if $P$ has \visionstability $\delta$ then the one-shot algorithm will return the optimal solution.
Let us denote by $s\in \Q$ the value of the one-shot IP.

First, we will show that $s\leq \opt$.
As we know that $\opt = \opt(P,-\delta)$ there exists a finite set $G$,
with $|G| = \opt$ such that $G$ is $(-\delta)$-guarding $P$.
The set $G$ is not necessarily among the vertex-candidates.
Using the definitions above, $\representative(\face (g))$  denotes the representative vertex of the face of $\A$ that $g$ is contained in. 
Note that $\face(G)$ is also $(-\delta)$-guarding $P$.
(Recall that the visibility of a face~$f$ is defined as
$\vis_{\gamma}(f) = \bigcup_{p\in f} \vis_{\gamma}(p)$.)
Thus, in particular for every face-witness~$w$, there exists a 
face $f\in \face(G)$, which sees an interior point of~$w$.
By Lemma~\ref{lem:Face-Point-Replacement}, it holds that there
is a vertex $v\in \representative(\face(G))$, such that  $w\subseteq \vis(v)$. 
We apply the lemma with $\gamma  = -\delta$.
An interior point of~$w$ is visible by~$G$, as~$G$ is $(-\delta)$-guarding $P$.
This implies $s\leq \opt$.

The second step is to show that $s\geq \opt$.
Let $G$ be the guards returned by the one-shot IP.
We construct a new set of point-guards $S$ of size at most $s$ that is $\delta$-guarding $P$. As $\opt(P,\delta) = \opt$, it holds that 
$s \geq \opt$.
The guards~$G$ contains potentially face-guards and 
may not guard all face-witnesses.
However all the vertex-witnesses are guarded.
We define $S = \representative(\face(G))$.
In other words, for every guard~$g$, consider the face $g'= \face(g)$ that~$g$ is contained in. 
This is $g$ itself, in case that~$g$ is a face.
In case that $g$ is contained in several faces, pick an arbitrary one.
Then consider the representative of that face~$g'$.
The set~$S$ is the union of all of those representatives. 
We claim that~$S$ is $\delta$-guarding~$P$.
To this end let $w$ be a face-witness.
Then there is at least one guard in $G$ that sees at least one \emph{interior} point of $w$, as $G$ sees all vertex-witnesses.
(Recall that every face-witness contains one vertex-witnesses by definition.)
By Lemma~\ref{lem:Face-Point-Replacement}, it holds that there
is a vertex $v\in S$, 
such that  $w\subseteq \vis_\delta(v)$.
(Here, we apply the lemma with $\gamma = 0$.)
Thus, $\representative(\face(G))$ is $\delta$-guarding $P$.
Thus, $s\geq \opt(P,\delta) = \opt$ implies that 
$s \geq \opt$.

This finishes the proof of Theorem~\ref{thm:OneShotAlgo}. \hfill \qed

\section{Iterative \Visionstable Algorithm}
\label{sec:IterativeAlgo}
The practical bottleneck of the one-shot algorithm is not solving the integer program, but the polynomial-time preprocessing step.
Our iterative algorithm has several \emph{ideas}, which give huge improvements in practice, while maintaining the theoretical 
performance guarantees.

The first idea is to maintain a coarse arrangement~$\A$ and only refine $\A$ where needed in an iterative fashion.
Thus, the name iterative algorithm.
The refinement is by splitting appropriate faces.
In order to identify faces to be split, 
we will need to modify the integer program. 
We will modify the integer program such that we will either find a face to be split, find the optimal solution, or conclude that the polygon has not \visionstability~$\delta$.
In the last case, we can continue with $\delta/2$ as \visionstability.
Splitting faces only where needed gives a huge speed-up, as we show in our experiments in Section~\ref{sub:speedupeffects}.
The next idea is to reduce the witness set.
When we think about witnesses, we noticed that we actually do not necessarily need all witnesses for the integer program. 
We identify a much smaller set of \emph{critical witnesses}. In the integer program formulation, we require only the critical witnesses to be seen. 
At a later stage, we check if the guards that we compute in this way are actually guarding everything.
We are dynamically increasing the size of the critical witness set.
The main advantage is that this does not require us to compute all pairs of visibilities, but only between the guards and the arrangement as well as between the critical witnesses and the arrangement.
See Figure~\ref{fig:iterations-critical} for an illustration.

\begin{figure}[H]
	\centering
	\includegraphics[page=2]{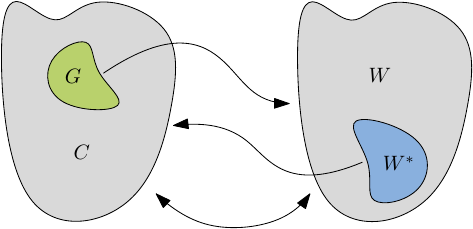}
	\caption{We compute all visibilities between the critical witnesses $W^*$ and the candidates $C$ and then all visibilities between the guards $G$ and the witnesses $W$. {These} are much fewer than the visibilities between all candidates and all witnesses.}
	\label{fig:iterations-critical}
\end{figure} 
Using critical witnesses reduced the number of visibilities that we had to compute largely. 
However, after this improvement, there were still a large number of visibilities that we had to compute, which slowed down the algorithm.
We noticed that a large percentage of these queries were between pairs of vertices or faces that were quite far from each other in the polygon.
This lead to our last idea. 
The idea is to compute a subdivision of the polygon, which we call the {\em \weakVisPolyTree}, see Figure~\ref{fig:WeakVisPolTree}. 
It captures locality of the polygon. 
Using the \weakVisPolyTree, we can prevent a sizeable fraction of the visibility queries to be even asked.

Let us point out that the iterative algorithm maintains the theoretical performance guarantees as it would, in the worst-case, subdivide $\A$ until it is as fine as the arrangement from Section~\ref{sec:OneShotAlgo}.
However, in practice most faces will be split much fewer times.
We will show the following theorem.
\ThmIterative*

In the following {subsections}, we describe the iterative algorithm in more detail. 
Often there are several possible choices.
We always describe a ``normal protocol'', which is the choice that we consider for Theorem~\ref{thm:iterative}.
We also define other protocols, as those give huge practical improvements, albeit we cannot prove theoretical performance for them anymore.

\label{sub:Initialization}
\subsection{Initialization}
In the initialization phase, we are constructing an arrangement~\A, 
which we will refine in subsequent steps.
At first, we add all defining chords~$c$ of the \weakVisPolyTree to~\A.
Furthermore, we shot horizontal and vertical rays from each reflex vertex. We stop the ray as soon as it hits the boundary of the polygon or another edge of~\A. 
Note that there could have been other choices been made. 
In practice, it is easy particularly easy to shoot horizontal and vertical rays.
See Figure~\ref{fig:Saved-Vis}~c).
We shoot the rays in arbitrary order.
There are $O(r)$ defining chords, where $r$ is the number of reflex vertices.  This is because each chord, other than the first one, is incident to at least one reflex vertex and no two chords are incident to the same reflex vertex. Furthermore, we are shooting $O(r)$ rays, in horizontal and vertical directions.
As, we are not introducing any crossings, 
we conclude that \A has a complexity of $O(r)$,
where $r$ is the number of reflex vertices of the polygon.
And all faces of \A are convex.

\begin{figure}[H]
	\centering
	\includegraphics[page=3]{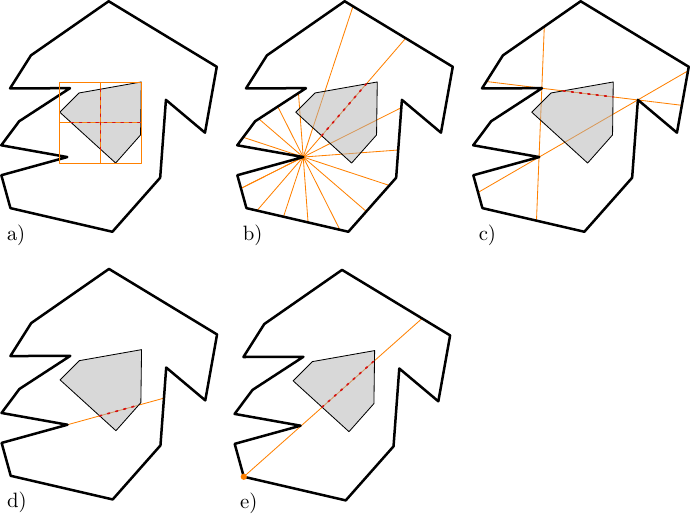}
	\caption{a) square split,
		b) angular split,
		c) (reflex) chord split,
		d) extension split,
		e) visibility line split}
	\label{fig:splits}
\end{figure}

\subsection{Splits}
\label{sub:splits}
In this {subsection}, we assume that some other part of the algorithm identified some appropriate faces to be split.
We choose from five different types of splits, called
\emph{square split, angular-split, (reflex) chord split, extension split} and \emph{visibility line split}, see Figure~\ref{fig:splits}. 
We choose randomly among the different type of splits, where the probabilities {were} chosen empirically based on the experiments in Section~\ref{sec:Tests}.

\begin{figure}[H]
	\centering
	\includegraphics[page = 6]{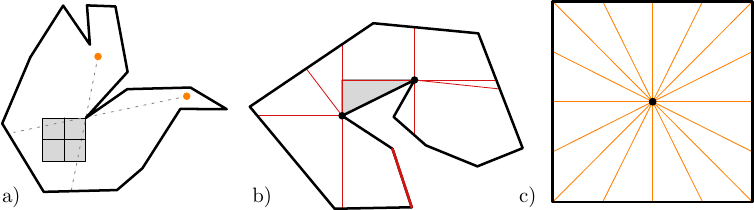}
	\caption{a) The two orange guards are guarding the entire polygon and the polygon is \visionstable. Still, even after an arbitrary number of 
		square-splits, there remains a face that is neither from the top nor from
		the right guard completely visible. 
		b) If a face is incident to two reflex vertices, it may guard  an entire polygon that otherwise can only be guarded with two guards. Square splits are efficient to  ensure that the face is only incident to at most one reflex vertex.
		c) Directions that we use for angular splits.}
	\label{fig:square-split}
\end{figure}

\paragraph*{Square splits.}
Square splits divide a face using a horizontal and vertical segment,
such that the height and width of the new faces are halved.
These type of splits have several advantages.
They are really easy to compute. 
They ensure that each new face is only incident to at most one reflex vertex.
They are efficiently reducing the \power of the new faces, in the case that the face was not adjacent to a reflex vertex. 
When a face is adjacent to a reflex vertex, its \power is not reduced by a square split.
This is because one of the newly formed faces will still be adjacent to the same reflex vertex.
Square-splits also have another disadvantage: it seems non-trivial to upper bound the number of square splits needed to ensure that every face has \power at most~$\delta$.
For an example, see Figure~\ref{fig:square-split}~a).

\paragraph*{Angular splits.}
In angular splits, we shoot rays from a reflex vertex, 
as to reduce the \power of the face. 
One fundamental problem with angular splits is that we are not allowed to use trigonometric functions (sinus, cosinus) in the real RAM model of computation.
But those functions are also problematic in practice, as we inherently need to use rounding. 
As a work around, we consider the $2^k$ rays  as indicated in Figure~\ref{fig:square-split}~{c}).
We define the \emph{granularity} of the angular splits as $\lambda = 2^{-k}$.
Recall that $\alpha/2 \leq \sin(\alpha) \leq \alpha$ for $\alpha \in [0,\pi/4]$. 
Thus, the angle $\alpha$ between any two rays is  bounded by 
\[\pi \lambda \leq  \alpha  \leq 4\pi \lambda.\]
We use the granularity to get an estimate of the \power of a face.
We initialize $k$ as $4$ and thus the granularity with $1/16$. Afterwards, we increment $k$ by $1$ when necessary.

The biggest advantage of using angular splits is that we can easily upper bound how many of them we use in the worst-case.
The reason being that angular splits are 
part of the arrangement from the one-shot~algorithm.

\paragraph*{Reflex chord splits.}
We check if there exists a reflex chord that intersects the face. If so, we use it to split the face in two (see Figure~\ref{fig:splits}~c).
If no such reflex chord exists, we try a different type of split.
Reflex chords are very important for the correctness of Lemma~\ref{lem:Face-Point-Replacement}. 
Thus, we need to include them.
It is not very efficient to compute them though,
and we have not encountered a practical situation where missing
them would be relevant. Thus, we choose those splits only with low probability.
However, we do them, if no other type of split is possible.

\paragraph*{Extension splits.}
Given a reflex vertex~$r$, we can consider the rays with apex~$r$ parallel to the two incident edges of~$r$.
These rays are commonly called \emph{extensions}.
If we want to make an extension split, 
we check if there is an extension that properly intersects the given face.
Those splits can be useful, if an optimal guard happens to lie precisely on an extension.

\paragraph*{Visibility line splits.}
The way we split a face when doing a visibility line split depends on the nature of the face~$f$ that  we are splitting. 
If we are using~$f$ as a face-guard, we try to find a witness point~$w$ that is seen by~$f$, but not seen by any of the other candidates in the current solution set. 
Furthermore, while~$f$ can see $w$, we find~$w$ such that~$w$ cannot fully see~$f$. 
Then we intersect $\vis(w)$ with~$f$, and use this intersection to split~$f$. 
If no such $w$ exists, we try a different type of split instead.
When~$f$ is an unseen face-witness, we find the candidate guard~$g$ from the current solution set that can only see part of~$f$. 
Then similarly to before, we split~$f$ using the intersection with $\vis(g)$.
As is discussed in Section~\ref{sub:PractRunning}, the reason that we include this type of split is that in specific cases, this split will massively reduce the amount of iterations necessary to find the solution for a polygon with very low \visionstability.

\paragraph*{Unsplittable faces.}
In order to show performance guarantees, we need to make sure that we are not doing too many splits. For that purpose, we declare faces \emph{unsplittable}, if all of the following conditions are met.
\begin{itemize}
	\item the face is incident to at most one reflex vertex.
	\item no reflex chord or extension split is possible.
	\item the face is not splittable by angular splits of granularity $\lambda$.
	(Here, $\lambda$ is the current granularity.)
\end{itemize}
The last item, implies that the \power of the face is at most $4\pi \lambda$.
In particular, If a face is unsplittable then all conditions of Lemma~\ref{lem:Face-Point-Replacement} and Corollary~\ref{cor:Vision-Enclosure} on the faces are met for $\delta \geq  8\pi \lambda $.
(Recall that Lemma~\ref{lem:Face-Point-Replacement} and Corollary~\ref{cor:Vision-Enclosure} give sufficient conditions for certain visibility regions to be contained in one another.)

\paragraph*{Protocol.}
There are many different ways that we can decide which split to make. Here, we explain two protocols, the normal protocol and the square split protocol. 
In the \emph{normal protocol}, we do square splits only when the face is incident to more than one reflex vertex.
This will happen only very few times and usually only at the beginning of the algorithm.
We choose the angular split with probability $0.6$. Then, with probability $0.2$, we choose the visibility line split. Finally, we choose the other two remaining split types {(chord and extension split)} with equal probability.
In case a split is impossible, we will first try the other splits (except for square split), before deciding that a face is unsplittable.
In the \emph{square split Protocol}, we only use square splits.
We know some situations, where square-splits are not good and will make the algorithm run into an infinite loop, see Figure~\ref{fig:square-split}~a).
However, for {many} polygons, square splits give a significant performance boost. {However, the algorithm is not very robust due to the possibility of infinite loops.}

\subsection{Critical Witnesses}
\label{sub:CriticalWitnesses}
The arrangement~\A has many vertices and faces that are not particularly important to be guarded. 
The idea of the critical witness set is to identify a subset $W^*\subseteq W$ of critical witnesses which are relevant.
At the beginning, 
we initialize  the \emph{critical witness} set~$W^*$ by randomly picking~$10$ percent of vertices and faces, for each weak visibility polygon separately. 
{The precise starting size of the critical witness set is not very important, as long as it is roughly equally distributed over the whole polygon.}
Later, through out the iterations, we add to the critical witnesses set as necessary, by using the guards~$G$ given by the integer program.
See Section~\ref{sub:IntegerProgram} for a detailed description of 
the integer program.
We compute all the vertices and faces that $G$ sees. This also gives us the sets of unseen face-witnesses and vertex-witnesses $U$ which were not marked as critical before. 
We then randomly choose a small constant size subset of vertices and faces from $U$ that we add to $W^*$. 
In the practical implementation of the program, the size of this subset of $U$ that we add to $W^*$ depends on the number of cores of the processor of the system that it is run on. 
Using the number of cores is plausible, because our practical implementation runs the visibility queries in parallel.
Note that if we were to mark all unseen witnesses as critical this would lead to very large numbers of visibility queries, thus defeating the purpose of using the critical witness set. 
It would also increase the size of all subsequent integer programs that we need to solve.
It is important to find a good balance between adding too few critical witnesses and adding too many. 

Every time we update the critical witness set, we re-run the IP to check if we can find a better solution given the critical witnesses. 
We keep adding to the critical witness set as long as there are unseen witnesses left that are not marked as critical. 
This means that the witness set will keep {growing}, which makes sure that the algorithm is deterministic.
We then check if we need to update the new critical witness set again using this new-found guard set. We keep doing these \emph{critical cycles}, 
until we find a guard set that can see the entire polygon. 
Only then we split the faces and continue with the next iteration. 

A face can only be removed from the critical witness set, if it is split.
For every critical witness face, we also add a critical vertex to the interior of that face.
Critical witness vertices are only removed from the critical witness set,
if they are interior to a face that is removed.

To summarize, we do not need to compute all the visibilities 
between all the candidates and all the witnesses. 
Instead, we compute all visibilities between the critical witnesses $W^*$ and candidates $C$ and then all visibilities between the guards $G$ and the witnesses $W$.
As there are much fewer guards than candidates and also much fewer critical witnesses $W^*$ than witnesses overall, this saves a lot of running time in practice.

\paragraph*{Protocol.}
Although the use of critical witnesses gives a big improvement in practice, we cannot show theoretical performance guarantees using critical witnesses.
For later reference, we say that we use the \emph{critical witness} protocol in case that we use critical witnesses.
Otherwise, we do not use critical witnesses.
We also define the \emph{delayed critical witness protocol}.
In this protocol, we use critical witnesses and we add all faces, which have a{n} \power larger than the granularity $\sqrt{\lambda}$.
In this way, we balance theoretical and practical performance.

\subsection{Building the integer program.}
\label{sub:IntegerProgram}
In order to ensure that we keep getting closer to an optimal solution, we need to use in total at most two integer programs per iteration. 
The \emph{normal IP} and the \emph{big IP}.
We want one of the following three scenarios to happen.
\begin{enumerate}
	\item We find an optimal solution and we can confirm that it is optimal.
	\item We find at least one face that we can split and thus make progress.
	\item We find that our current granularity $\lambda$ is too high and we can update the granularity. 
\end{enumerate}
Recall that the granularity is an estimate for the \visionstability.
In case that we use critical witnesses, we are also content to find unseen witnesses that we can add to the critical witness set.

\paragraph*{Normal IP.}
The normal IP is the same
as the one-shot IP.
In case that we use critical witnesses the normal IP only adds constraints and variables for the critical witnesses.

\paragraph*{Big IP.}
Before we decrease the granularity, we want to make sure that
we have not missed a face that we may want to split.
We construct the big IP with the purpose of finding a solution involving at least one splittable face.
Let us denote by $s\in \Z$ the value of the normal IP rounded down, which is the number of used guards.
Given a set $S$ of faces and vertices, we denote the  subset of splittable faces by~$\splittable(S)$.
We define the objective function as 
\[f = \sum_{x\in \splittable(W\cup C)}  \const{x}.\]
We aim to maximize the number of splittable faces that are used.
They appear either as guards or unseen witness faces.
We add several constraints.
Every variable is either~$0$ or~$1$.
We require that the total number of candidates equals~$s$.
\[ \sum_{c \in C} \const{c}  = s\]
We also require all vertex-witnesses to be seen.
This leads to the following constraint, 
for every $w\in \vertex(W)$.
\[\sum_{c\in \VIS(w)} \const{c}\geq 1 \]
Recall that $\VIS(w) \subseteq C$ denotes the set of candidates that see $w$ completely.
As we want to identify exactly those witness faces , that we can split, we add the following constraint, for every splittable witness-face~$w\in \splittable(W)$. The $\epsilon$ here is the same as in the Normal IP.
\[ 1 - \left( \eps \sum_{c\in \VIS(w) } \const{c} \right)  \geq \const{w}\]
This constraint ensures that
\[   \sum_{c\in \VIS(w) } \const{c} \geq 1 \quad  \Rightarrow  \quad \const{w} = 0 \]
In other words, if there is a guard that sees $w$, then the corresponding variable $\const{w}$ must be set to~$0$.

The solution to the normal IP and to the big IP yields
naturally a set of guards ($G = \{ c\in C : \const{c} = 1 $)
and a set of unseen witnesses ($U = \{ w\in W : \const{w} = 1 $).
Reversely, if we have given a set of guards and unseen witnesses, this gives a feasible solution to the normal IP and to the big IP.

\vspace{0.1cm}

\paragraph*{IP Protocol.}
Let us first describe the \emph{normal IP protocol}.
We will always use the normal IP. 
Let us denote by $G$ the set of guards returned by the normal IP.
(If there are faces or vertices that are not seen 
by the guards~$G$,
we may add some of them to the critical-witness set, see Section~\ref{sub:CriticalWitnesses}.)
If~$G$ contains only vertex-guards 
and we checked that $G$ sees the entire polygon, then the algorithm reports~$G$ as the optimal solution.
If either~$G$ or the unseen witnesses contain at least one splittable face, we split that face and continue with the next iteration.
If none of the above happens, we run the big IP.
Again, we check for the same set of events.
In case that again none of the events above happen,
the algorithm 
updates our estimate granularity~$\lambda$.

In practice, it seems that using the big IP rather slows down the running time. Although we may subdivide the ``wrong faces'' too many times, it seems not too often.
Even worse, running the big IP seems to only prevent us from progressing faster.
We define the \emph{simple IP protocol} as follows.
Only run the normal IP and split all the faces that are selected as unseen face-witnesses or face-guards, 
according to Section~\ref{sub:splits}.
Repeat with the new arrangement.

\paragraph*{Granularity Update Protocol.}
We {only consider one protocol } for the granularity update.
In the \emph{normal granularity update protocol}, we replace $\lambda$ by $\lambda/2$. 
{We do this whenever there is an unsplittable face.}


\begin{figure}[H]
	\centering
	\includegraphics[page=5]{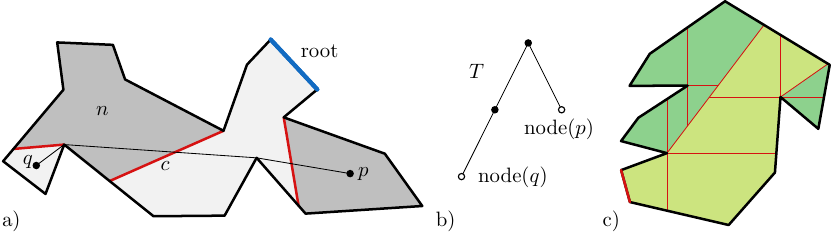}
	\caption{a) \& b) The points $p$ and $q$ cannot see each other. Otherwise, $q$ would belong to the node $n = \vis(c)$.
		c) The arrangement $\A$ is initialized by shooting horizontal rays from the reflex vertices.}
	\label{fig:Saved-Vis}
\end{figure}

\subsection{\WeakVisPolyTree}
\label{sub:WeakVisTree}
Consider the polygon in Figure~\ref{fig:WeakVisPolTree}.
Despite the fact that the polygon has many vertices, it seems
to have locally low complexity.
The idea of the weak visibility tree is to exploit this 
low local complexity in order to reduce the number of visibility queries to be answered.

To build the \weakVisPolyTree $T$ of a simple polygon $P$, we start with an arbitrary edge~$e$ on the boundary of $P$.
We compute the weak visibility polygon $\vis(e)$ of $e$, which is the root of~$T$.
This forms the \weakVisPolyTree $T_1$.
We use the algorithm by Hengeveld, Miltzow and Staals~\cite{ConvexExpansion} to compute the weak visibility polygon.
In a first iteration of the paper we used the algorithm by Abrahamsen~\cite{abrahamsen2013constant}, but we soon found out that this was too slow for practical use.
We construct $T_{i+1}$ from $T_i$ as follows.
Take any edge~$e'$ of some $\vis(e)$, which is not part of the boundary of $P$, we compute the weak visibility polygon~$\vis(e')$ with respect to the polygon $P\setminus T_i$. Those weak visibility polygons are the children of $\vis(e)$.
We continue inductively to compute the children of every weak visibility polygon. Note that every node of $T$ is a weak visibility polygon~$W$ of some defining chord~$c$. Obviously, $c$ splits $P$ into two polygons $Q$ and $Q'$, as $c$ is a chord. The weak visibility polygon~$W$ is completely contained in either~$Q$ or~$Q'$. To be precise, we consider each node of~$T$ to be closed and contain its boundary.

\begin{lemma}[Saved Visibilities]
	\label{lem:saved-visibilities}
	Let $p,q$ be in the interior of two different nodes $\node(p),\node(q)$ of $T$. 
	If $\node(p)$ and $\node(q)$ are neither siblings nor in a parent-child relationship then $p,q$ cannot see each other.
\end{lemma}
\begin{proof}
	For an illustration of this proof consider Figure~\ref{fig:Saved-Vis}~a)~\&~b).
	Consider the shortest path $s = \short(p,q)$ from $p$ to $q$.
	We will show that $s$ is not a line-segment and thus
	conclude that $p$ and $q$ cannot see each other.
	Note that $s$ will traverse several nodes of $T$.
	In particular, it will cross a defining chord $c$ and its 
	corresponding node $n = \vis(c)$.
	Without loss of generality, {$n$ is an ancestor of $q$}.
	Now, if $s$ would be a line-segment, than $q$ would be in $n = \vis(c)$.
	But this is impossible as $q$ is in the interior of $\node(q)$,
	as stated in the assumption of the lemma.
\end{proof}

Thus, whenever, we want to compute the set $\vis(p)$, we only need to consider the parent children and siblings of $\node(p)$.
This is true whether $p$ is a face or a point.
We build a shortest path map for every weak visibility polygon~\cite{guibas1987linear}.
{
These shortest paths can be later used in order to answer face visibility queries.
See Guibas et al.~\cite{guibas1987linear} for details on the shortest path map and how to use it for visibility queries between edges and faces.}
Note that, while Guibas et al.~\cite{guibas1987linear} gives better theoretical performance guarantees, we use a new algorithm
by Hengeveld, Miltzow and Staals~\cite{ConvexExpansion}, 
as it works much better in practice.
{Note that the \weakVisPolyTree approach only works for simple polygons. When~$P$ has holes, the assumptions about children and siblings may not hold.}

\subsection{Correctness}
\label{sub:iterative-correct}
In this paragraph, we prove Theorem~\ref{thm:iterative}.
To be precise, we use the normal split protocol, the normal IP protocol, the normal granularity update protocol, and no critical witnesses.

The proof is divided into three parts.
First, we show that the algorithm is reliable (even if the input polygon is not \visionstable, the reported result is correct).
Second, we show that the algorithm makes progress in each step.
Third, we will show an upper bound on the total number of iterations.

\paragraph*{Reliability.}
We prove reliability in the same way that we proved it for the one-shot algorithm. For the benefit of the reader, we repeat the argument.
We show that if all guards $G$ found by the integer program are point-guards and all face-witnesses are seen, then the algorithm indeed returned the optimal solution.
Here, we are not assuming that $P$ is \visionstable.
As $G$ guards the entire polygon it holds that
$|G|\geq \opt$.
Let us now consider a set $F$ of point-guards of optimal size.
For every $p\in F$, we denote by $\face(p)$ the face of the arrangement~\A that contains~$p$.
For some finite set $S \subseteq P$, we denote by $\face(S) = \{\face(p): p \in S\}$.
Note that $\face(F)$ sees all vertex-witnesses.
It may be that there is one or several face-witnesses~$w$ that require several faces of $\face(F)$ to be completely guarded.
Such a face-witness~$w$, would be counted as unseen face-witness.
Still the normal IP could use exactly those face-guards $\face(F)$.
As the normal IP computed a set of point guards~$G$,
it implies $|G| < \opt + 1$.
Note that $\eps ( |G| + |W|) < 1$, by definition of~$\eps$.  
Thus, if the normal IP returns a solution using point-guards 
only and all face-witnesses are seen, we have~$|G| = \opt$.
This shows that the algorithm is reliable.

\paragraph*{Progress per iteration.}
In this paragraph, we will show that the iterative algorithm makes progress
in every iteration.
To be specific, we will show that after every iteration one of the following events will happen.
\begin{enumerate}
	\item \label{itm:opt}
	We find an optimal solution and we can confirm that it is optimal.
	\item \label{itm:split}
	At least one face is splittable.
	\item \label{itm:update-delta}
	We find that $\delta < 8\pi \lambda$.
\end{enumerate}
In case~\ref{itm:opt}, we are done, as the algorithm is reliable and the algorithm terminates.
In case~\ref{itm:split}, we split those faces according to the description given in Section~\ref{sub:splits}.
In the last case, we have the granularity $\lambda$.
It remains to show $\delta < 8\pi \lambda$.
We show the contraposition.
\begin{claim}\label{clm:HalfLambda}
	We denote by $\delta$ the \visionstability of the underlying polygon~$P$.
	Let $\delta \geq 8\pi \lambda$ and neither the big IP, nor the normal IP return a splittable face.
	Then the normal IP will return the optimal solution.
\end{claim}
\begin{proof}[of Claim~\ref{clm:HalfLambda}]
	We denote by $s$ the value of the normal IP. We will show that $s = \opt$.
	Before we go into details of the proof, note that we defined unsplittable faces exactly so that we  can apply Lemma~\ref{lem:Face-Point-Replacement} and Corollary~\ref{cor:Vision-Enclosure}.
	
	The first step is to show that $s\geq \opt$.
	Let $G_0$ be the guards returned by the normal IP.
	Let us denote by $U_0\subseteq W$ the set of faces that are not seen by $G_0$. 
	By assumption of the lemma, all faces in $U_0$ are unsplittable.
	We construct a new set of point-guards $G_1$ of size at most $s$ that is $\delta$-guarding $P$. 
	As $s \geq |G_0| \geq |G_1| \geq \opt(P,\delta) = \opt$, it holds that 
	$s \geq \opt$.
	The guard set~$G_0$ contains potentially face-guards.
	Note that all face-guards of~$G_0$ are unsplittable and 
	thus their \power is at most $\delta/2$.
	We define $G_1 = \representative(G_0)$.
	That is, for a vertex $v \in G_0$, we define $\representative(v) = v$.
	For a face $f\in G_0$ we define $\representative(f)$ as in Section~\ref{sec:Visstability}.
	Let us denote by $U_1\subseteq W$ the set of unsplittable faces that are not $\delta$-guarded by $G_1$. 
	As all face-guards of $G_0$ are unsplittable, we can apply Corollary~\ref{cor:Vision-Enclosure}. This implies that
	$U_1\subseteq U_0$ and thus contains only unsplittable faces.
	To this end let~$w\in U_1$ be any face-witness. 
	We will show that $w$ is actually $\delta$-guarded by $G_1$.
	By construction, there is an interior vertex $v \in \interior(w)\cap W$ {that} is guarded by some guard~$g_0 \in G_0$.
	Let $g_1 = \representative(g_0)$.
	Then, the \power of $w$ and $g_0$ are at most $\delta/2$.
	Furthermore, neither $w$ nor $g_0$ is properly intersected by a reflex chord. 
	Both faces are incident to at most one reflex vertex.
	Thus, we can apply Lemma~\ref{lem:Face-Point-Replacement}.
	This implies that $w \subseteq \vis_\delta(g_1)$.
	
	To summarize,~$G_1$ is finite vertex set, which is $\delta$-guarding~$P$. Thus, $s \geq |G_1|\geq \opt(P,\delta) = \opt$.
	
	\vspace{0.2cm}
	
	As a second step, we show that $s\leq \opt$.
	As we know that $\opt = \opt(P,-\delta)$ there exists a set $G_0$
	of point-guards with $|G_0| = \opt$ such that $G_0$ is $(-\delta)$-guarding $P$.
	The set $G_0$ is not necessarily among the vertex-candidates.
	Note that the set of faces $G_1 = \face(G_0)$ containing $G_0$ is also $(-\delta)$-guarding~$P$.
	(Recall that the visibility of a face~$f$ is defined as
	$\vis_{\gamma}(f) = \bigcup_{p\in f} \vis_{\gamma}(p)$.)
	We denote by $U_1\subseteq W$ the set of witnesses that are not seen by any guard in $G_1$.
	Note that, as $G_0$ is $(-\delta)$-guarding $P$ it holds that $U_1$ contains only face-witnesses.
	Thus, $(G_1,U_1)$ defines a feasible solution to the normal IP and the big IP. 
	This implies that all faces of $G_1$ and $U_1$ must be unsplittable. 
	Otherwise, the big IP would have identified those splittable faces.
	Now, let $G_2 = \representative(G_1)$ be the set of representatives of $G_1$.
	We will show that $G_2$ is guarding $P$.
	Let $U_2\subset W$ be the set of unseen face-witnesses of $G_2$.
	By Corollary~\ref{cor:Vision-Enclosure}, it holds that $U_2 \subseteq U_1$. This in turn implies that $U_2$ contains only unsplittable faces.
	Let $w \in U_2$. We show that there is a guard in $G_2$ that actually sees $U_2$.
	By assumption, there is an interior vertex $v\in\interior(w)$.
	And there is a guard $g_1 \in G_1$ that sees $v$.
	Let $g_2 = \representative(g_1) \in G_2$ be the representative of $g_1$.
	By the above, it can be checked that all conditions of Lemma~\ref{lem:Face-Point-Replacement} are met. 
	This implies that $w \subseteq \vis(g_2)$.
	Thus, $G_2$ guards the entire polygon~$P$.
	In particular, $|G_2| \leq \opt$ and $G_2$ defines a feasible solution for the normal IP. Thus, $s\leq \opt$.
\end{proof}

\paragraph*{Total number of iterations.}
The idea is to upper bound the number of iterations with the number of faces in the final arrangement~$\A^*$ of the iterative algorithm.
By the previous paragraph, we split at least one face in every step of the algorithm, except if we lower the  granularity. 
Every edge in $\A^*$ comes from one of the following sources:
\begin{enumerate}
	\item \label{itm:ExtensionSegemnts}
	It is part of an extension segment, see Figure~\ref{fig:splits}~d).
	\item \label{itm:WeakSegemnts}
	It is an edge that comes from the boundary of a weak visibility polygon.
	\item \label{itm:SquareSegemnts}
	It comes from a square split in order to prevent that a
	face contains more than one reflex vertex, see Figure~\ref{fig:square-split}~b).
	\item \label{itm:AngularSegemnts}
	It comes from an angular split, see Figure~\ref{fig:splits}~b).
	\item It comes from a reflex chord split, see Figure~\ref{fig:splits}~c).
\end{enumerate}
By definition there are $O(r)$ segments of type~\ref{itm:ExtensionSegemnts},~\ref{itm:WeakSegemnts} and~\ref{itm:SquareSegemnts}, here $r$ denotes the total number of reflex vertices of~$P$.
Let $\lambda^*$ be the final value of the granularity.
Then there are at most $r/\lambda^*$ segments of type~\ref{itm:AngularSegemnts}.
Here we consider the maximal segment starting at the reflex vertex and not just the edge that split the specific face.
By the the paragraph above, we know that this implies that
the \visionstability~$\delta$ is at most $8\pi \lambda$.
In other words, there are at most~$O
(\frac{r}{\delta})$ segments of type~\ref{itm:AngularSegemnts}.
By definition there are at most $O(r^2)$ reflex chords.
As any two segments intersect at most once, we have at most
\[O\left(\left( r + r/\delta + r^2 \right)^2\right) = O(r^4/\delta^2 ) = \left(\frac{n}{\delta}\right)^{O(1)}\]
vertices. As the arrangement is planar, 
this also describes the total number of edges and faces asymptotically.

Each iteration solves at most two IPs and takes polynomial time. This finishes the proof of Theorem~\ref{thm:iterative}. \hfill \qedhere

\section{Chord-Visibility Width}
\label{sec:ChordWidth}
In this section, we define the notion of \emph{\chordwidth}, 
and give a justification 
why we believe it to be an interesting parameter in practice. 
Then, we describe a fixed parameter time algorithm for the art gallery problem with respect to the \chordwidth.
We denote with \emph{vertex-guarding} a variant of the art gallery problem where the guards are restricted to lie on the vertices of the input polygon.

We will show the following two theorems.
\begin{theorem}
	\label{thm:vertex-guardingFPT}
	Let $P$ be a simple polygon.
	Then there is an FPT algorithm for vertex-guarding $P$
	with respect to the \chordwidth.
\end{theorem}

The running time of this algorithm is $2^{k^{3}}n^{O(1)}$, where $k$ is the \chordwidth.

\ChordFPT*

The running time of this algorithm is $2^{k^{7}/\delta^2} n^{O(1)}$, where $k$ is the \chordwidth.
We want to point out that the theoretical running times of both algorithms make it prohibitive to use those algorithms in practice. 
Thus, both algorithms are a theoretical contribution to the art gallery problem.

\subsection{Definition and Justification}
Given a chord $c$ of~$P$, we denote by $n(c)$ the number of vertices visible from $c$. 
The {\em \chordwidth} ($\cw(P)$) of a polygon is the maximum $n(c)$ over all possible chords $c$.
We refer to a recent paper by Klute, M Reddy and Miltzow, illustrating chord visibility width~\cite{klute2021local}.

The \chordwidth is a way of capturing 
the local complexity of a polygon, 
with respect to the notion of visibility.
Clearly, not all polygons in practice have small \chordwidth.
We like to think about \chordwidth as {a} measure on local complexity.
It is noteworthy that many synthetic polygons that are created in a random process have much smaller \chordwidth than they have reflex vertices.
On the other hand, polygons~$P$ constructed in many hardness reductions have typically~$\cw(P)$ roughly proportional to the total number of vertices.

\subsection{FPT algorithms}
For the rest of this section, we fix $k$ to denote the \chordwidth of $P$.
As a warm-up, we prove Theorem~\ref{thm:vertex-guardingFPT}.
Let $T$ be, for the rest of this section, the \weakVisPolyTree as defined in Section~\ref{sub:WeakVisTree}.
We aim to describe a dynamic programming algorithm 
on the tree~$T$.
We note the following lemmas as a preparation.
\begin{lemma}
	\label{lem:WeakPolComplexity}
	Let $u$ be a node of $T$, then $u$ has at most $k = \cw(P)$ vertices.
\end{lemma}
\begin{proof}
	Let $c$ be the defining chord of $u$.
	Every vertex of $u$ is visible from $c$.
	The lemma follows from the definition of \chordwidth.
\end{proof}

Given a node $u$ of~$T$, it consists geometrically of several edges and vertices. 
Some of those edges are part of the boundary of~$P$. Other edges are in the interior of~$P$.
We call those second type of edges \emph{windows} of $u$.

\begin{lemma}
	\label{lem:WindowComplexity}
	Let $u$ be a node of $T$, then $u$ has at most $k = \cw(P)$~windows.
	In other words, every node in $u$ has at most $k$ children
	and thus also at most $k$ siblings.
\end{lemma}
\begin{proof}
	Any window of $u$ contains at least one reflex vertex,
	which is visible from the defining chord of $u$.
	The lemma follows from the definition of \chordwidth.
\end{proof}

We are now ready to prove Theorem~\ref{thm:vertex-guardingFPT}.

\begin{proof}[of Theorem~\ref{thm:vertex-guardingFPT}]
	Let us denote by $V$ the set of~$n$ vertices of~$P$
	and by~$R$ the set of reflex vertices of~$P$.
	First we construct an arrangement~\A defined by $\chord(V,R)$
	and the boundary of~$P$.
	Note that $\chord(V,R)$ consists of $O(n^2)$ line-segments.
	Thus, \A consists of at most $O(n^4)$ edges and line-segments.
	(Note that we could give even better bounds using \chordwidth, but they are not needed.)
	Let $\face(A)$ denote the set of faces of \A.
	Let $u$ be some node of the weak visibility polygon~$T$, as described in Section~\ref{sub:WeakVisTree}.
	Let $\vertex(u)$ be the set of vertices of $P$ inside $u$.
	We create for each node~$u$ of~$T$ a table denoted by $\tafel_1(u)$.
	We have an entry for every set $S\subseteq \vertex(u)$ and 
	we list all the faces $F\subseteq \face(A)$ that are visible from~$S$.
	Note that by definition of~\A, it holds that every face is either completely seen or not at all from a vertex of~$P$.
	Constructing $\tafel_1(u)$ takes at most $2^k n^{O(1)}$ time and space, by Lemma~\ref{lem:WeakPolComplexity}.
	
	As a next step, we create for every node $u$ 
	a second table~$\tafel_2(u)$.
	Let $W$ be the set of vertices of the node $u$ and all its children.
	Clearly $W$ has at most $k^2 + k$ vertices, see Lemma~\ref{lem:WeakPolComplexity} and Lemma~\ref{lem:WindowComplexity}.
	For every subset $F\subseteq W$, we store the faces of $\face(A)$
	it sees and whether it is possible to completely see all the faces
	in all descendants of $u$, using $F$.
	And if it is possible, then we also compute how many guards are needed at least, including the guards in~$S$.
	This is trivial for the leaves of~$T$ as the leaves have no children.
	Now consider a node~$u$ with at most~$k$ children $u_1,\ldots,u_k$.
	We consider the case that the tables of $\tafel_2(u_1),\ldots,\tafel_2(u_k)$ are already created. 
	Then we can create the table $\tafel_2(u)$ in $O(2^{k^3}n^{O(1)})$, as follows.
	For every set~$F$, check for all the children, using the previous entries of $\tafel_2(u_i)$, if the subtree below~$u_i$ can be completely seen, and if yes, how many vertices are needed.
	As $T$ has only $O(n)$ nodes, we can create all tables in 
	$2^{k^3}n^{O(1)}$ as well.
	In particular, this also creates $\tafel_2$ also for the root of~$T$.
	We go through all the entries of the root.
	We check which of them sees all the faces of the root as well.
	Among those entries, we choose one with the minimum number of guards used.
\end{proof}

Note that, by definition, in vertex-guarding, all the vertices of~$P$ form a candidate set that contains the optimal solution and we used this candidate set in a crucial way in the previous proof.
Now, as a corollary of Theorem~\ref{thm:OneShotAlgo}, all the vertices of the arrangement~\A (as defined in Section~\ref{sec:OneShotAlgo}) form a candidate set, which contains the optimal number of guards for the art gallery problem.
We will do dynamic programming along this new candidate set to prove Theorem~\ref{thm:FPT-chordwidth}.
The proof of Theorem~\ref{thm:FPT-chordwidth} follows along the same lines as the proof of Theorem~\ref{thm:vertex-guardingFPT}. 
We need the following lemma as preparation.

\begin{lemma}
	\label{lem:Vertices-A-in-Weak}
	Let $\delta$ be the \visionstability of $P$ and $u$ a node of the \weakVisPolyTree of~$P$.
	Then $u$ contains at most $O(k^6 / \delta^2)$ vertices of~\A.
\end{lemma}
\begin{proof}
	We first count the number of maximal segments of \A that are
	at least partially inside $u$.
	Note that every reflex vertex of~$P$ emits at most $1/\delta$ such chords, by the angular ray shooting.
	Furthermore, every reflex vertex~$r$ can see also at most~$k$
	other reflex vertices. 
	This bounds the segments from $\chord(R,R)$ associated with~$r$ by~$k$ as well.
	Furthermore, the number of reflex vertices in~$u$, its parent or one of its children is upper bounded by $(k+2)k$.
	Thus, there are at most $2k^2(k+1/\delta)$ many segments of \A intersecting~$u$.
	As any two segments intersect at most once, we get that there are at most~$O(k^6 / \delta^2)$ vertices of \A inside~$u$
\end{proof}

We note that the bound in the lemma can clearly be improved.
However, the algorithm would still be prohibitively slow. 
We are focusing here on a simpler exposition rather on getting the best theoretical bounds.

We are now ready to prove Theorem~\ref{thm:FPT-chordwidth}.

\begin{proof}[of Theorem~\ref{thm:FPT-chordwidth}]
	The algorithm is almost identical to the algorithm in Theorem~\ref{thm:vertex-guardingFPT}.
	The difference is in the candidate set that we use in each node. Instead of size $k$ it has now size $O(k^6/\delta^2)$.
	Thus, each table has size $O(k^7/\delta^2)$.
	This changes the running time to $2^{O(k^7/\delta^2)}n^{O(1)}$.
\end{proof}

\section{Test Results}
\label{sec:Tests}

We tested the practical implementation of the iterative algorithm, described in Section~\ref{sec:IterativeAlgo}, in several ways. 
The goal of the first experiment, described in the Section~\ref{sub:PractRunning}, was to find out the practical running time of the implementation and how it relates to input factors such as size, \chordwidth and \visionstability. 
These tests were performed on random, simple input polygons of different sizes up to 500 vertices. 
The input polygons were obtained from the AGPLIB library~\cite{art-gallery-instances-page}, a library used for other papers on the art gallery problem as well~\cite{engineering}. 
We will describe the polygons in more detail in Section~\ref{sub:PractRunning}. 
Furthermore, in Section~\ref{sub:speedupeffects}, we discuss how the two speedup methods, the \weakVisPolyTree and the critical witnesses, improve the running time.
An additional experiment was conducted to show that the practical implementation gives iteratively improving guard solutions, even for the irrational guard polygon that requires irrational guards, presented in~\cite{abrahamsen2017irrational}. More information about this experiment can be found in  Section~\ref{sub:convergence}.

Section~\ref{sub:Distrib} shows the CPU time distribution, showing what percentage of CPU-time is spent on doing which tasks.

The experiments {concerning the newly introduced algorithms} were ran on a computer with a 64-bit Windows 10 operating system, a 8-core Intel(R) Core i7-7700HQ CPU at 2800 Mhz 
and 16 GB of main memory. 
{The experiments with the algorithm from Tozoni et~al. were done on the same computer described above, but on a Linux Mint operating system (using a dual-boot set-up).}

The practical implementation heavily makes uses of version 4.13.1 of CGAL~\cite{cgal:eb-20a}. 
The IP solver used was IBM ILOG CPLEX version 12.10~\cite{cplex}. 

\subsection{Practical running times and correctness}
\label{sub:PractRunning}
To find out how the implementation performs in practice, we tested the algorithm on several input polygons. As mentioned previously, the input polygons were taken from the AGPLIB library~\cite{art-gallery-instances-page}. We had access to random simple polygons of four different size classes: 60, 100, 200 and 500 vertices. 
We tested on 30 different instances per size class.
An example of one of the 200-vertex polygons and its solution is shown in Figure~\ref{fig:200v-example}.

\begin{figure}[H]
	\centering
	\includegraphics[width=\textwidth]{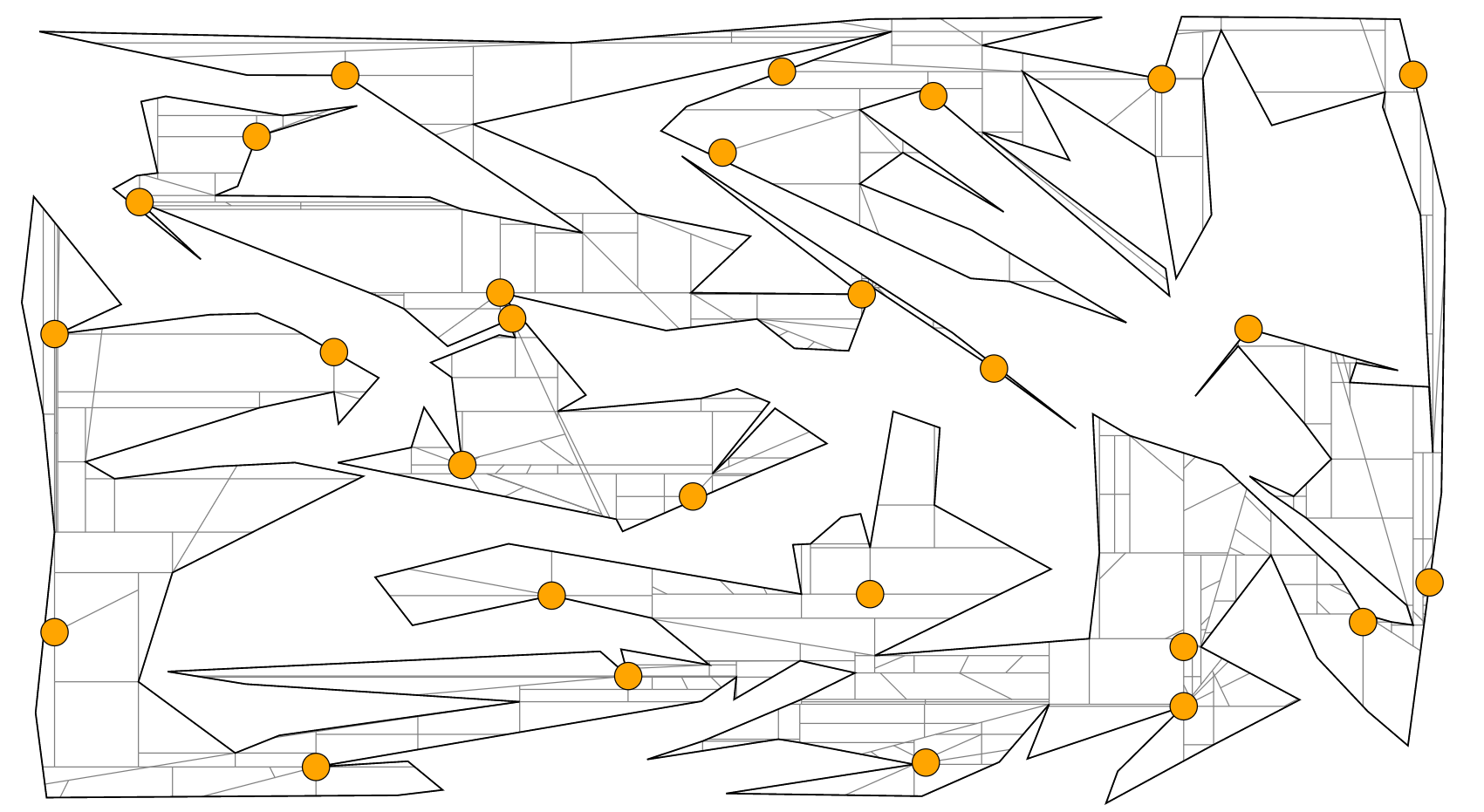}
	\caption{An example of a 200-vertex input polygon and its solution, as found by the iterative Algorithm without safe guards. We have also included the final arrangement at the end of the algorithm. 
	}
	\label{fig:200v-example}
\end{figure}

For these tests, we used two different versions of the iterative algorithm.
{
The first version is the most efficient algorithm, including several optimizations that make it perform much better in practice, at the cost of losing the performance guarantees. This version is referred to as the iterative algorithm without safeguards. The optimizations used are discussed in the paragraph below.
 The second version is the algorithm with performance guarantees from Theorem~\ref{thm:iterative}. 
}
The correctness  of the practical implementation was verified using the implementation provided by~Tozoni et al.~\cite{tozoni}. 
This implementation was also tested on the same input data-set, which makes it possible for us to cross-check our results to see whether we have found a solution of the correct size. 
\def\arraystretch{1.15}
\begin{table}[tph]
	\centering
	\begin{tabular}{|c|c|c|c|} \hline
		\multirow{2}{*}{\begin{tabular}{c}\textbf{Sizes}\end{tabular}} & \multicolumn{3}{c|}{\textbf{Average time (s)}} \\ \cline{2-4}
		& \textbf{Tozoni et al.} & \textbf{Tozoni et al. (Our hardware)} & \textbf{The Iterative Algorithm}  \\ \hline
		
		60 & 0.26 & 0.18 & 0.39 \\
		100 & 0.94 & 0.68 & 0.52  \\
		200 & 3.77 & 2.54 & 2.02 \\
		500 & 35.04 & 22.34 & 18.2  \\ \hline
	\end{tabular}
	\caption{A comparison of the iterative algorithm without safe guards with the results from Tozoni et al.~\cite{tozoni}, both the results reported by Tozoni et al. themselves~\cite{tozoni} and results found using their implementation on our hardware. Note that all reported times include pre-processing times as well. Tests were ran on 150 polygons, but our algorithm could not find the optimal solution for one polygon within the time limit, so the times of 149/150 polygons are displayed in this table. 
	}
	\label{tab:running-times}
\end{table}

\begin{figure}[tbph]
	\centering
	{
    	\includegraphics[width=\textwidth]{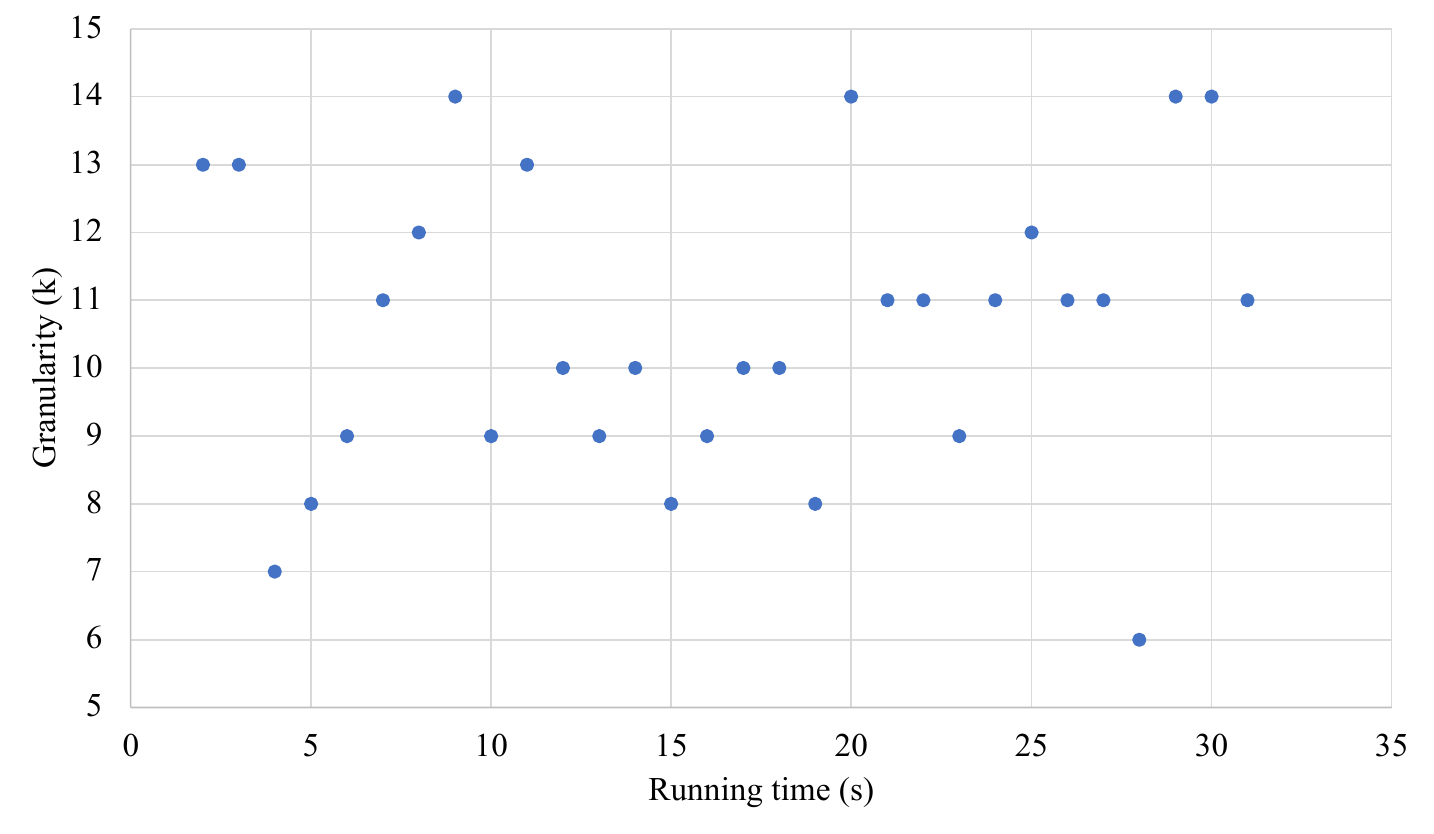}
    			\vspace{-24pt}
    }
	\caption{
	A scatter plot containing on the $x$-axis the running time and on the $y$-axis the granularity obtained from our experiments with the iterative algorithm without safe guards.
 The polygons used are the $30$~polygons of size~$500$. The granularity~$\lambda$ on the $y$-axis is displayed as $k$ in {$\lambda = 2^{-k}$}. 
	}
	\label{fig:scatter500}
\end{figure}

\paragraph*{Iterative Algorithm without Safe Guards.}
For this version, we opted for the most efficient protocols in practice. This means that we forsake some of the theoretical guarantees that we do not need in practice. 
See Section~\ref{app:Intermediate-Arrangement} for an illustration.
For these tests, we used the critical witness protocol. When splitting, we used the normal split protocol. Finally and most importantly, we applied the simple IP protocol, to make sure to make as few unnecessary splits as possible. 
Using this variant, we found some reasonable practical results. We were able to find the optimal solutions
for all tested instances, except for one polygon of size ({Polygon \#25 in the set of polygons of size 100}).
We put a thirty minute timeout on our experiments, and this is the only polygon that took longer than this time limit. Furthermore, when looking at this polygon, we noticed that it has an exceptionally low \visionstability. 
When running the algorithm with an increased probability of performing visibility line splits, we found that we could find the optimal solution in the same time-frame as the other polygons of size~$100$. 
However, increasing this probability had a strong negative effect on the average running times for the other 149 polygons.
Because of this, we chose to leave this polygon out of the testbed, rather than to adapt the implementation specifically for one polygon. 
If we left the polygon in the testbed, the average runtime would be infinite if we use our original parameters. 
In case that we were to use parameters tailored to that polygon the runtime would be misleading.

The averages running times for each of the size classes are shown in Table~\ref{tab:running-times}. Note that the times we report include the pre-processing time of computing the weak visibility decomposition before the first iteration, which are very short, even for the larger polygon sizes. 
This table also shows the results by Tozoni et al.~\cite{tozoni} pertaining to polygons of these sizes. Because their implementation was made publicly available, we were able to test it on the same system that we used to run our own experiments. 
To make the comparison more fair, we decided to display these times in the table, alongside the results found in 2016. 
However, note that we were limited to use a worse, free IP solver (GLPK instead of XPRESS) because we did not have a license for the best solver available that can be used for the implementation. 

We acknowledge the fact that the results from Tozoni et al.~\cite{tozoni} were improved on by de~Rezende et al.~\cite{engineering}. 
{However, we were not able to find the individual running times from this improved version corresponding to these sizes, and this updated version was not made publicly available. We could only find average running times.}
This table shows that running times of our algorithm are faster, except for the smallest size of 60.
The algorithm, amongst other factors, is sensitive to the vision stability of a polygon. This means that few polygons are very hard to solve and thus may overly influence the average.
Additional details on the running times can be found in Figure~\ref{fig:scatter500}, where we show each of the running times of the polygons with size 500 plotted against the largest value of the {inverse granularity necessary to find the optimal solution, where granularity $\lambda = 2^{-k}$}.

\paragraph*{Iterative Algorithm with Safe Guards.}
This is the version of the algorithm from Theorem~\ref{thm:iterative}. We did not make use of critical witnesses. 
For splits, we used the normal split protocol, in which we mostly focus on doing splits that guarantee a decrease in the \power of the face. Furthermore, we  used the normal IP protocol.
This means that we use the normal IP at every iteration, and the big IP only if the normal IP did not find a splittable face. We also use the normal granularity update. 

After running the tests on polygons of size~$60$, we found that the algorithm finds the solution in an efficient manner in $25$ out of $30$ instances. For the other~$5$ instances, no solution was found after~$60$ minutes. 
Looking at the results from the iterative algorithm without safe guards, we found that these~$5$ polygons seem to have rather low vision stability. 
To estimate the \visionstability, we use the minimum granularity $\lambda$ of angular-splits of a face that we split, as a heuristic. The polygons that we were able to solve using this variant had granularity within the $\frac{1}{16}$ - $\frac{1}{128}$ range. From testing the other variant, described in the previous paragraph, we found that the~$5$ polygons that this algorithm could not solve had granularity values between~$\frac{1}{256}$ and~$\frac{1}{2048}$.

When we use the big IP, we split 
until there is no optimal solution with splittable faces. 
In these instances with low vision stability, this means that we will make many seemingly unnecessary splits. 
This leads to infeasible running times.

\begin{table}[tbhp]
	\centering
	\begin{tabular}{|c|cccc|}
		\hline
		\textbf{Size} &  $60$ & $100$ & $200$ & $500$  \\
		\hline
		\textbf{Correlation} & $0.07$ & $0.3$ & $0.1$ & $0.6$ \\
		\hline
	\end{tabular}
	\caption{The correlation coefficients between the {measured granularity ($k$ in $\lambda = 2^{-k}$)} and the running time, computed per size. }
	\label{tab:Correlations}
\end{table}

\paragraph*{Correlation of Granularity and Running Time.}
As mentioned before, the iterative algorithm is sensitive to the \visionstability of the input polygon. 
To test this, we saved the smallest granularity $\lambda$ necessary to find the solution. 
We did this for the iterative algorithm without safe guards. 
We then computed the correlation coefficients between the running times and the inverse of the granularity $\frac{1}{\lambda}$. 
These coefficients are shown in Table~\ref{tab:Correlations}.  
These are fairly strong correlations. Note that the minimum granularity~$\lambda$ might not be the best indication of the \visionstability of a polygon. 
We actually have no efficient way to compute the \visionstability efficiently.
A larger polygon might need a very fine subdivision in one part of the polygon, but can be relatively coarse in other parts. 

Besides this, there are several other random factors which influence the running time. 
For example, the IP chooses an arbitrary optimal solution out of several possible options and the splits we do at each iteration are also chosen randomly. Additionally, the \chordwidth of the \weakVisPolyTree has some effect on the running time. 
We believe that these factors account for the fluctuation in the correlations.

\paragraph*{Standard deviation in running times for individual polygons.}
In the previous paragraph, we explained that the differences between the running times for polygons can be quite large, and can be explained by several different factors, including the \visionstability of  the input polygon.
However, it would also be interesting to know the standard deviation of the running times if we test the same polygon multiple times. 
For this experiment, we ran the most efficient version of the algorithm, without Safe Guards, $15$ times on each of the~$30$ input polygons of size~$100$ and measured the running times and computed the averages standard deviations of each input polygon.
These statistics are visualized in a graph in Figure~\ref{fig:stdev}. 
We see that for some polygons the standard deviations can be quite large.
Several random factors in the iterative algorithm could account for these {large} standard deviations.
Firstly, the \weakVisPolyTree has a small degree of randomness in its construction. This is because we choose a random first edge. 
However, the results showed that for the $15$ different instances of the same polygon, the complexity of the \weakVisPolyTree did not change much. 
In fact, the number of weak visibility polygons and the largest weak visibility polygon were often very similar.
Next, we randomly choose the type of split that we perform.
For some polygons, this may be largely responsible for the deviation from the average. 
{If} a solution relies on collinearity, an extension or chord split can often create a candidate vertex at the necessary position.
Because of the randomness of the splits, we see that, for the same polygon, the minimal granularity that the algorithm used varies.
When we look at the instances with lower granularity, it seems that extension splits and chord splits were used more often than the instances with higher granularity for the same polygon. Figure~\ref{fig:lowgranexample}~(a) and (b) show a simple example of how this could happen. 
This suggests that we can lower the standard deviation and {increase the} minimal granularity by doing these types of splits more often. Finally, when several solutions are available, the IP solver will choose a random one. This means that the solver sometimes chooses a face that is useful to split and other times it may choose a face that is not so interesting.

\begin{figure}[h]
	\centering
	\def\svgwidth{\textwidth}
	\input{figures/New_Variance.pdf_tex}
	\caption{We tested~$30$ polygons of size~$200$ from the AGLIP \cite{art-gallery-instances-page}. Each polygon was tested $15$ times. This graph shows the standard deviation of the running times for each of the~$30$ polygons. Note that the time here excludes the pre-processing time, which is why the average is lower than the 2.02 seconds presented in Table~\ref{tab:running-times}.
	}
	
	\label{fig:stdev}
\end{figure}

\begin{figure}[h]
	\centering
	\includegraphics[page=18]{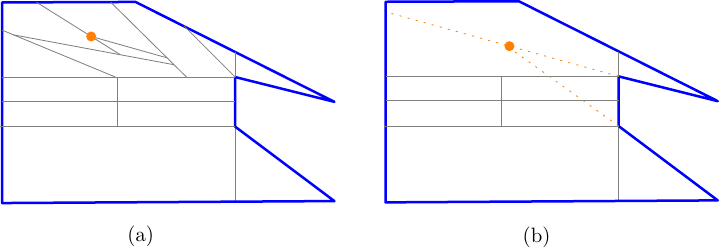}
	\caption{In (a), we perform several angular splits in a row. We are unlucky, and never choose to do an extension split. The minimum granularity gets quite low before we find a solution. In (b) we are lucky and do two extension splits in a row. The minimum granularity is never updated.}
	\label{fig:lowgranexample}
\end{figure} 

\subsection{Effect of the speed-up methods}
\label{sub:speedupeffects}
To find out the added value of the two speedup methods, we will look at different statistics that were measured in the experiments. We first look at the effect of the \weakVisPolyTree and then at the effect of the use of the critical witnesses. We show that both methods offer large speedups.

\paragraph*{The value of the \weakVisPolyTree.} In the experiments from the previous section, we also measured several things about the weak visibility polygon trees computed for the different input polygons. 
In particular, we tracked several characteristics of the \weakVisPolyTree: the number of weak visibility polygons in the tree, the largest number of {non-reflex} vertices, the largest number of reflex vertices of the largest weak visibility polygon in the tree and finally the percentage of visibility queries that we skip while using the \weakVisPolyTree.
Table~\ref{tab:weakvisstats} summarizes these characteristics for each size class. 

We see that the tree size seems to almost grow linearly with the size of the polygon.
However, the other two statistics grow much more slowly.
This attests to the effectiveness of the \weakVisPolyTree because it means that even the larger polygons are divided into relatively small weak visibility polygons.
This prevents the computation of many unnecessary visibilities throughout the course of the iterative algorithm. We can see this clearly when we look at the percentage of queries we save when we use the \weakVisPolyTree. This percentage goes up very quickly the larger the size of the input polygons.
\begin{table}[ht]
	\centering
	\begin{tabular}{|c|cccc|}
		\hline
		\textbf{Size} &  $60$ & $100$ & $200$ & $500$  \\
		\hline
		\textbf{Tree size} &
		14.2 & 23.0 & 46.3 & 115.0  \\
		\textbf{Largest polygon} & 20.5 & 23.3 & 26.2 & 28.4 \\
		\textbf{Largest number of reflex vertices} & 5.9 & 6.2 & 7.0 & 9.2 \\
		\textbf{Percentage of queries saved} & 16.7\% & 35.4\% & 63.5\% & 87.3\% \\
		\hline
	\end{tabular}
	\caption{We tested 30 input polygons from the AGPLIB library~\cite{art-gallery-instances-page} of four sizes. For each size class we see the averages of characteristics of the \weakVisPolyTree{s}.  
	}
	\label{tab:weakvisstats}
\end{table}

\paragraph*{Critical witnesses speedup.} To find if our second speedup method was also as effective, we conducted an additional experiment.
We tested the same polygons of sizes $60$, $100$, $200$ and $500$ without critical witnesses.
Table~\ref{tab:critspeedup} compares the running times of the versions with and without critical witnesses. 
Additionally, we see the number of witness points and faces used on average. For the method without critical witnesses, we list the number of total witnesses while the method with critical witnesses shows the number of critical witnesses only.
The table shows that the critical witnesses cause speedups factors of $2.2$ and larger for the different size classes.
When we compare the size of the critical witness sets to the total witness sets, we see that the ratios are about $8 : 1$ and $10 : 1$ for points and faces respectively.
This means we have to compute fewer visibilities, which gives a speed-up. However, in order to keep the critical witness set up-to-date we may need to solve several extra integer programs. This is clearly a trade-off, especially considering solving integer programs dominates the running time of the algorithm (see Section~\ref{sub:Distrib}). However, using critical witnesses still gives us some speed-up, which means that for the current state of the implementation, they are still an improvement.

\def\arraystretch{1.45}
\begin{table}[ht]
	\centering
	\begin{tabular}{|c|c|c|c|c|c|c|c|} \hline
		\multicolumn{1}{|c|}{\multirow{2}{*}{\textbf{Sizes}}} & \multicolumn{3}{c|}{\textbf{Average time (s)}} & \multicolumn{2}{c|}{\textbf{Witness points}} & \multicolumn{2}{c|}{\textbf{Witness faces}} \\ \cline{2-8} 
		\multicolumn{1}{|c|}{} & \multicolumn{1}{l|}{\textbf{With}} & \multicolumn{1}{c|}{\textbf{Without}}
		& \multicolumn{1}{c|}{\textbf{Speedup}}
		& \multicolumn{1}{c|}{\textbf{With}} & \multicolumn{1}{c|}{\textbf{Without}} & \multicolumn{1}{c|}{\textbf{With}} & \multicolumn{1}{c|}{\textbf{Without}} \\ \hline
		
		60	&	0.39	&	0.83	&	2.2	&	36.7	&	272.3	&	18.37	&	166.3	\\
		100	&	0.52	&	2.34	&	2.9	&	62.97	&	497.90	&	33.06	&	301.48	\\
		200	&	2.02	&	9.32	&	2.5	&	133.5	&	1125.53	&	63.8	&	683.37	\\
		500	&	18.2	&	69.92	&	2.5	&		378.77 &	3080.90	&	173.4	&	1851.45	\\
		
		\hline
		
	\end{tabular}
	\caption{A comparison of the iterative algorithm with and without the use of critical witnesses, tested 30 polygons of sizes $60$, $100$, $200$ and $500$.
		We compare the running times and the number of witnesses used in the IP. Depending on the version of the algorithm, the witnesses shown are either the number of critical witnesses or the total number of witnesses.  }
	\label{tab:critspeedup}
\end{table}

\subsection{Convergence to the optimal solution}
\label{sub:convergence}
In this experiment, we ran the iterative algorithm for~$30$ minutes for the irrational-guard-example from~\cite{abrahamsen2017irrational}. 
See Section~\ref{app:Intermediate-Arrangement} for an illustration of the first~$20$ iterations.
We used the square split protocol, critical witnesses and the simple IP protocol. The results with angular splits and reflex chord splits are similar in spirit, but the convergence is slower. Because we only used square splits, there was no need for the big IP as we did not update the granularity. 
This polygon is illustrated on the right of  Figure~\ref{fig:NotVisionStable}. 
Even though the algorithm will not be able to find the optimal solution,
we get a set of guards $G_i$ at the end of iteration~$i$. 
Note that $G_i$ is often a combination of point-guards and face-guards.
As we know the optimal solution~$F$, 
we can compute the Hausdorff distance $d_i = \disth(G_i,F)$ at the end of each iteration.
The Hausdorff distance is a metric that measures the distance between compact sets.

See Figure~\ref{fig:irrational-convergence}~a) for an illustration of the convergence speed per iteration.
Interestingly, the distance approaches the optimum very quickly.
Additionally, Figure~\ref{fig:irrational-convergence}~b) shows the Hausdorff distance plotted against the cumulative running time. 
Here, we see that, when plotted against the running time, the Hausdorff distance does not decrease as dramatically. 
This is because the later iterations take more time. 
That later iterations take more time is plausible.
The face splits at every iteration introduces new candidates and witnesses. 
This means that the IP has to deal with more candidate variables at each iteration and that we have to compute a larger number of visibilities. 
Section~\ref{app:Intermediate-Arrangement} includes images produced by the program showing several iterations of the algorithm obtained during this experiment.

\begin{figure}[p]
	\centering
	\def\svgwidth{\textwidth}
	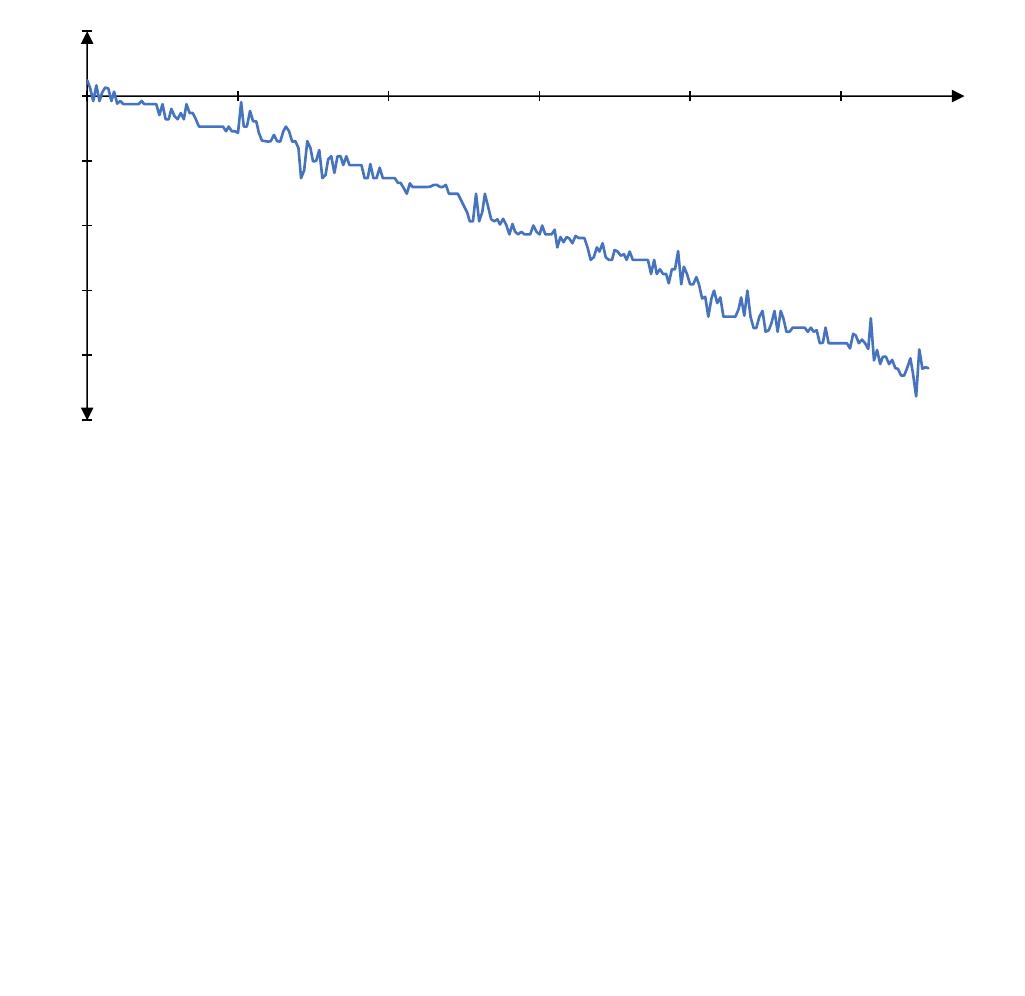
	\caption{The iterative algorithm based on the notion of \visionstability reports a sequence of solutions. Graph a) shows on the $x$-axis the iterations from~$1$ to about~$300$ and on the $y$-axis, the $\log_{2}$ of the Hausdorff distance to the optimal solution. Graph b) shows on the $x$-axis the cumulative running time and on the $y$-axis, the $\log_{2}$ of the Hausdorff distance to the optimal solution. All times are in seconds. }
	\label{fig:irrational-convergence}
\end{figure}

\begin{figure}[H]
	\centering
	\def\svgwidth{\textwidth}
	\input{figures/CPU_Distribution.pdf_tex}
	\caption{The chart shows the CPU distribution of the iterative algorithm implementation for solving $30$ polygons of size $200$.}
	\label{fig:distrib-chart}
\end{figure}

\subsection{Distribution of CPU usage}
\label{sub:Distrib}
This section describes how the workload of the algorithm is distributed on the CPU. To achieve this, {we} analyzed the CPU usage when computing solutions for all $30$ polygons of size $200$.
The CPU usage was comparable on polygons of different sizes.
The version of the algorithm analyzed here is the best performing variant: the Iterative Algorithm without safe guards. 

We performed this analysis using the Visual Studio profiler~\cite{profiler}.
Figure~\ref{fig:distrib-chart} shows the results found by the profiler. 
The CPU usage distributions for other polygons that we tested are very similar.
We divided the algorithm into several parts. 
These different parts are shown in a pie-chart in Figure~\ref{fig:distrib-chart}. 
The first part, in blue is the pre-processing time, which as described in Section~\ref{sub:Initialization}, consists of setting up the \weakVisPolyTree. 
The part shown in orange are point visibility queries, computing whether a point sees a witness face or witness point. Practically, this is achieved by using the existing CGAL methods, which is based on the triangular expansion method presented in~\cite{hemmer}. 
The chart shows in purple the percentage of the CPU time spent on face visibility queries, the question of whether a candidate face guard can see a face or point witness. 
We also compute these visibilities when checking whether we should add new critical witnesses, as we must make sure that a candidate solution sees the complete polygon. These queries are solved using a new weak visibility computation method that we developed in a follow-up project~\cite{ConvexExpansion}.
The yellow part shows the CPU time spent on solving the IP using the CPLEX solver. 
Finally, in gray, other, smaller tasks are combined into one section. 
These tasks typically consist of updating the intermediate arrangement, splitting faces using our different split techniques and keeping track of the results.
We see that most of the CPU time is spent on solving the IP. 
This means that, to improve the running time of the iterative algorithm, it is desirable to improve this part of the implementation. Perhaps this can be done by  reducing the number of IPs that the program needs to solve. For these experiments, we used the default parameters of CPLEX. {A first option} would be to look into different solvers. 
Another way of improving this would be to experiment with different parameters of the IP solver.
This could be achieved in many different ways.
For example, we could be more generous in adding critical witnesses.
{Also, more splits could be} made per iteration.
This would probably lead to less overall iterations.
It seems also plausible to use techniques that avoid using the IP solver
and use for instance linear programs as a proxy for the IP.

\subsection{GitHub code and results}
The practical implementation is available at the following GitHub repository~\href{https://github.com/simonheng/AGPIterative}{\url{https://github.com/simonheng/AGPIterative}}. 
The folder \emph{results} in the repository contains an excel file which reports the test results discussed in Section~\ref{sub:PractRunning}. 
The code provided in the repository are C++ source and header files, in addition to a Visual Studio solution file. 
Bear in mind that the code is dependent on CGAL version 4.13.1~\cite{cgal:eb-20a}, IBM ILOG CPLEX version 12.10~\cite{cplex}, the boost library and Libxl version 3.9.0.0 (used to read and write the excel result files). 
In order to be able to compile and run the project, the above dependencies must be installed.

\paragraph*{Acknowledgments.}
We thank Christopher Bouma for letting us use his computer to run our tests. 
We thank Sofia Rosero Abad for her graphic designs.
Furthermore, we would like to thank Marjan van den Akker and Rogier Wuijts  for interesting discussions on using IP solvers.
Special thanks goes to Matthew Drescher and Emile Palmieri-Adant for attempting to implement previous versions of the algorithm.
We thank Pedro de~Rezende for helping us to access previous implementations and links to relevant literature.
Finally, we thank \'{E}douard Bonnet 
and Mikkel Abrahamsen for helpful feedback on the presentation.

\section{Intermediate Arrangements}
\label{app:Intermediate-Arrangement}
\paragraph*{Irrational-Guard Polygon.}
Below we show the first $20$ iterations of the Irrational-Guard polygon. Section~\ref{sub:convergence} shows that this iteratively updating set of guards converges to the optimal solution. The orange points and faces represent point- and face-guards in the intermediate solution. The green faces represent faces not fully seen by the current candidate solution. Both orange and green faces are split in the next iteration. Note that for each of the orange and green faces, we draw a random vertex in the same colour. This is so that we can still see these faces when they become very small.
\\
\begin{minipage}[b]{0.48\textwidth}
	\centering
	\includegraphics[width=\textwidth]{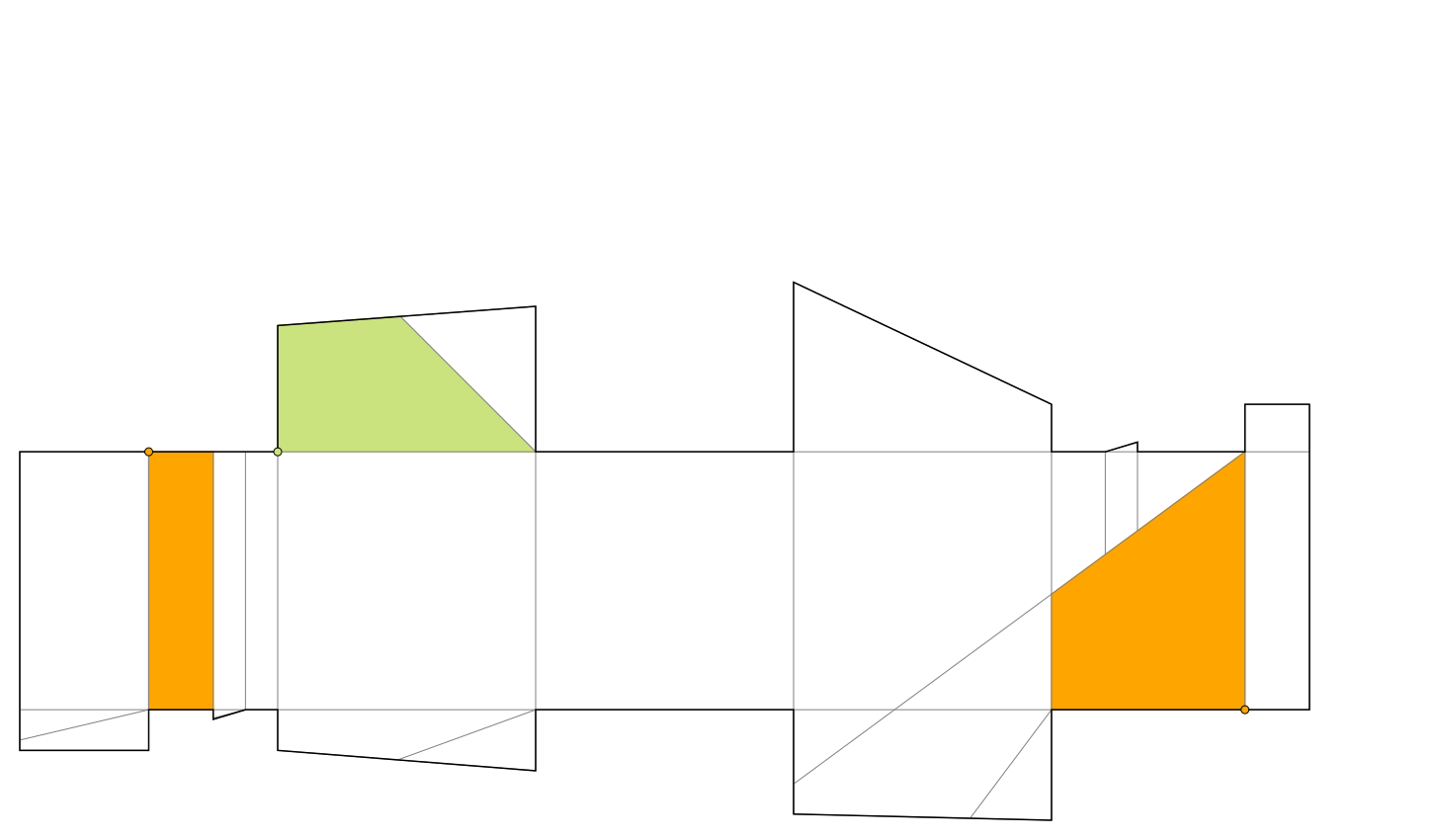}
	1
\end{minipage}
\hfill
\begin{minipage}[b]{0.48\textwidth}
	\centering
	\includegraphics[width=\textwidth]{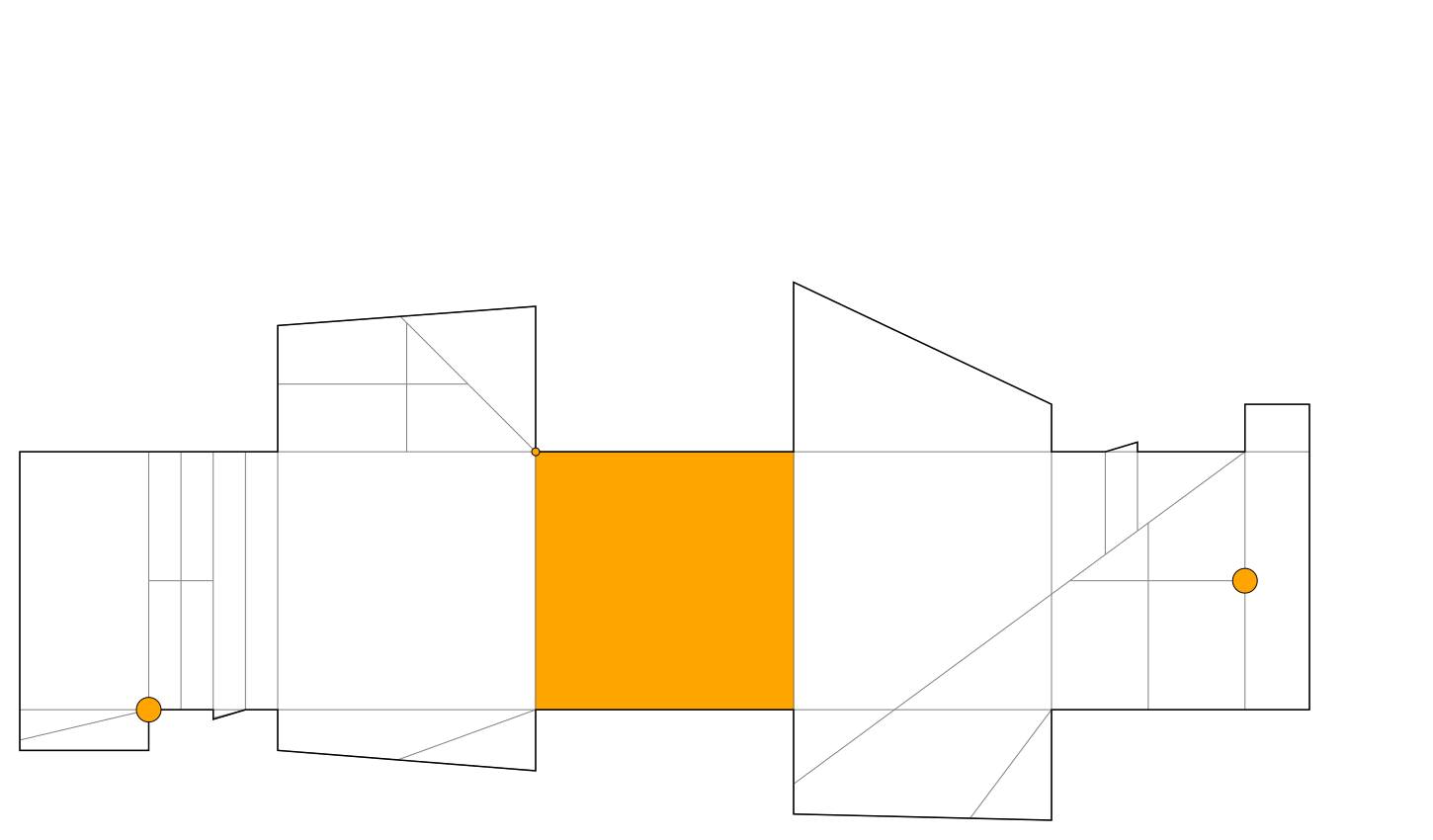}
	2
\end{minipage}
\begin{minipage}[b]{0.48\textwidth}
	\centering
	\includegraphics[width=\textwidth]{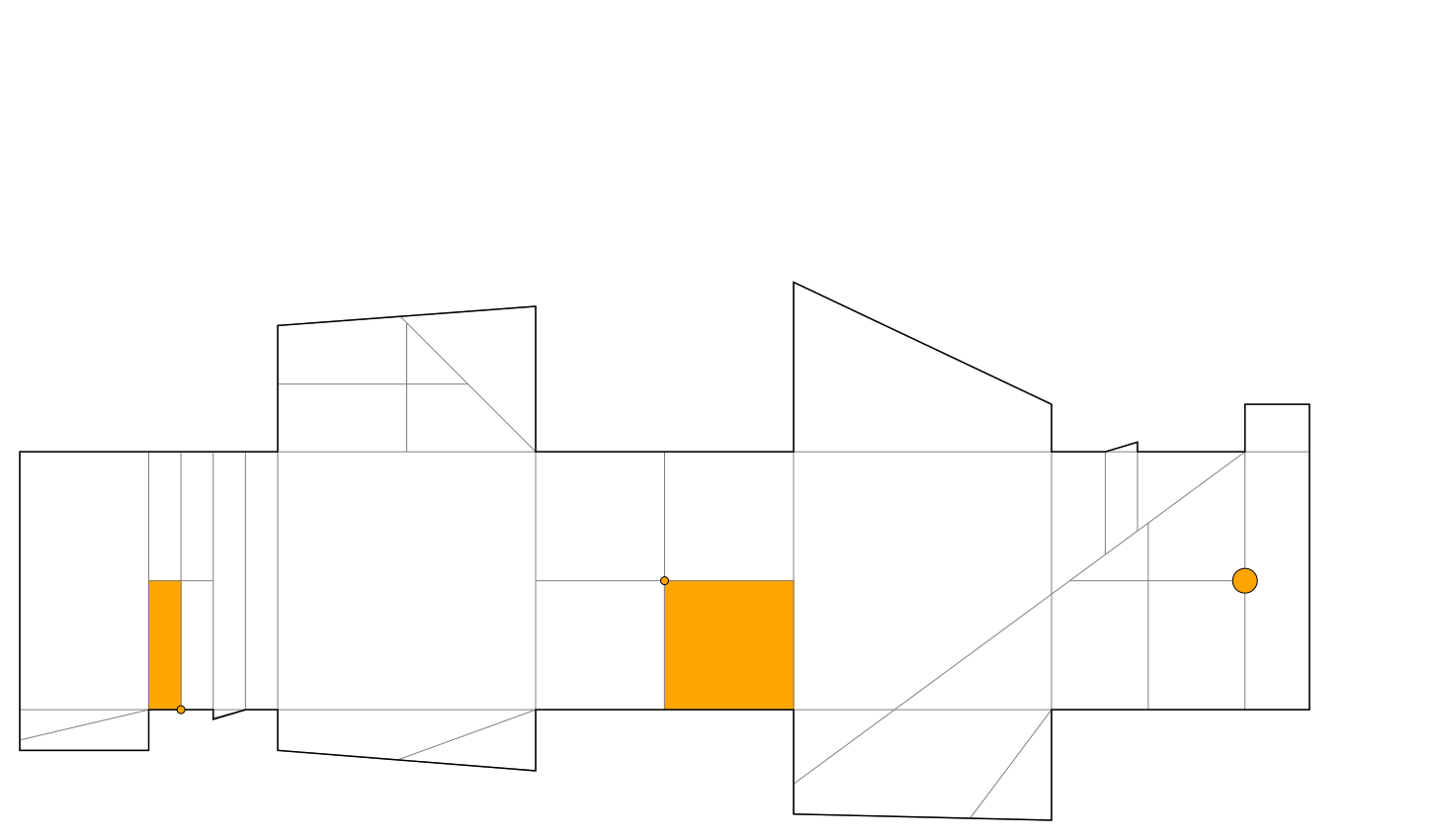}
	3
\end{minipage}
\hfill
\begin{minipage}[b]{0.48\textwidth}
	\centering
	\includegraphics[width=\textwidth]{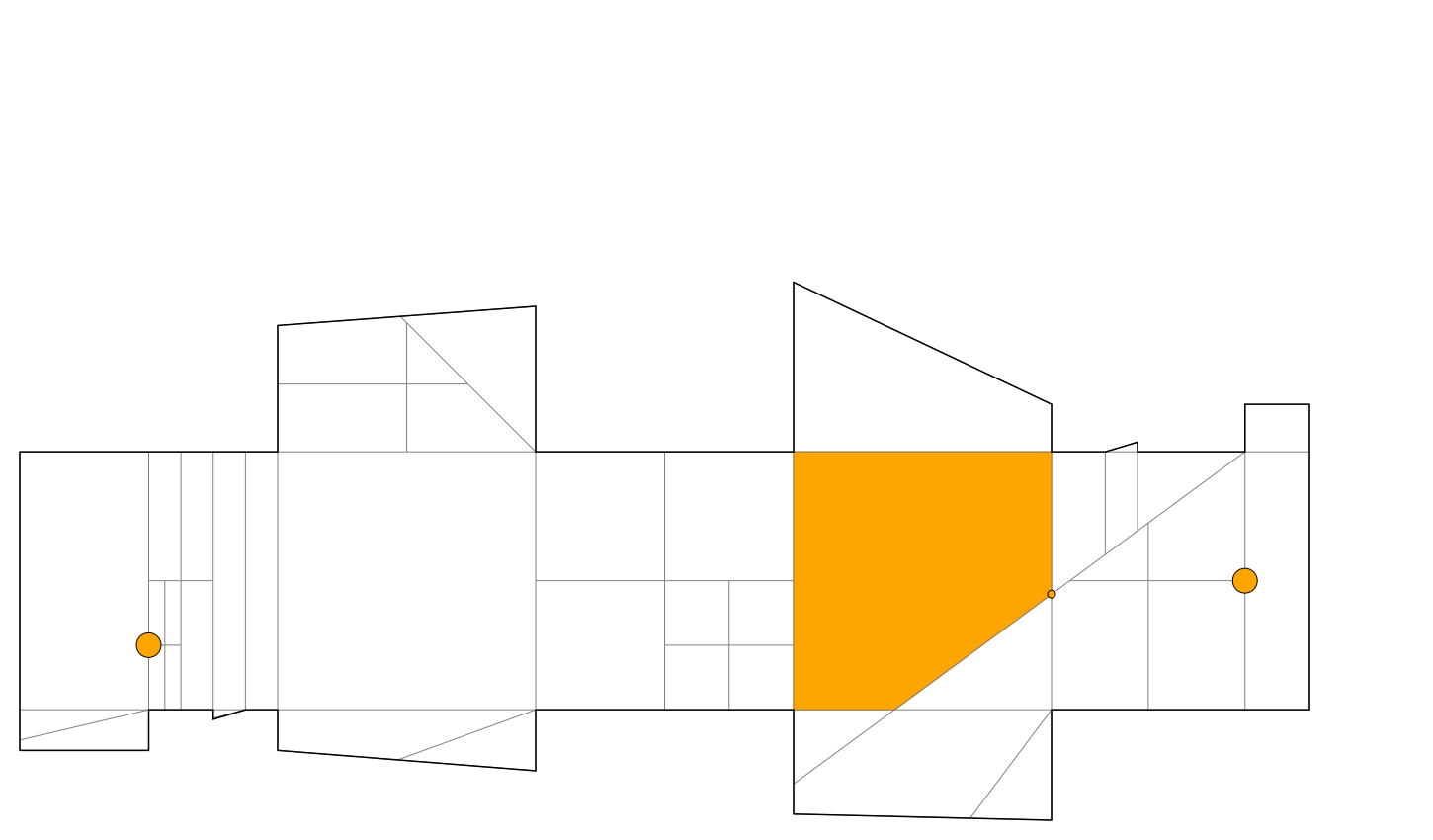}
	4
\end{minipage}
\begin{minipage}[b]{0.48\textwidth}
	\centering
	\includegraphics[width=\textwidth]{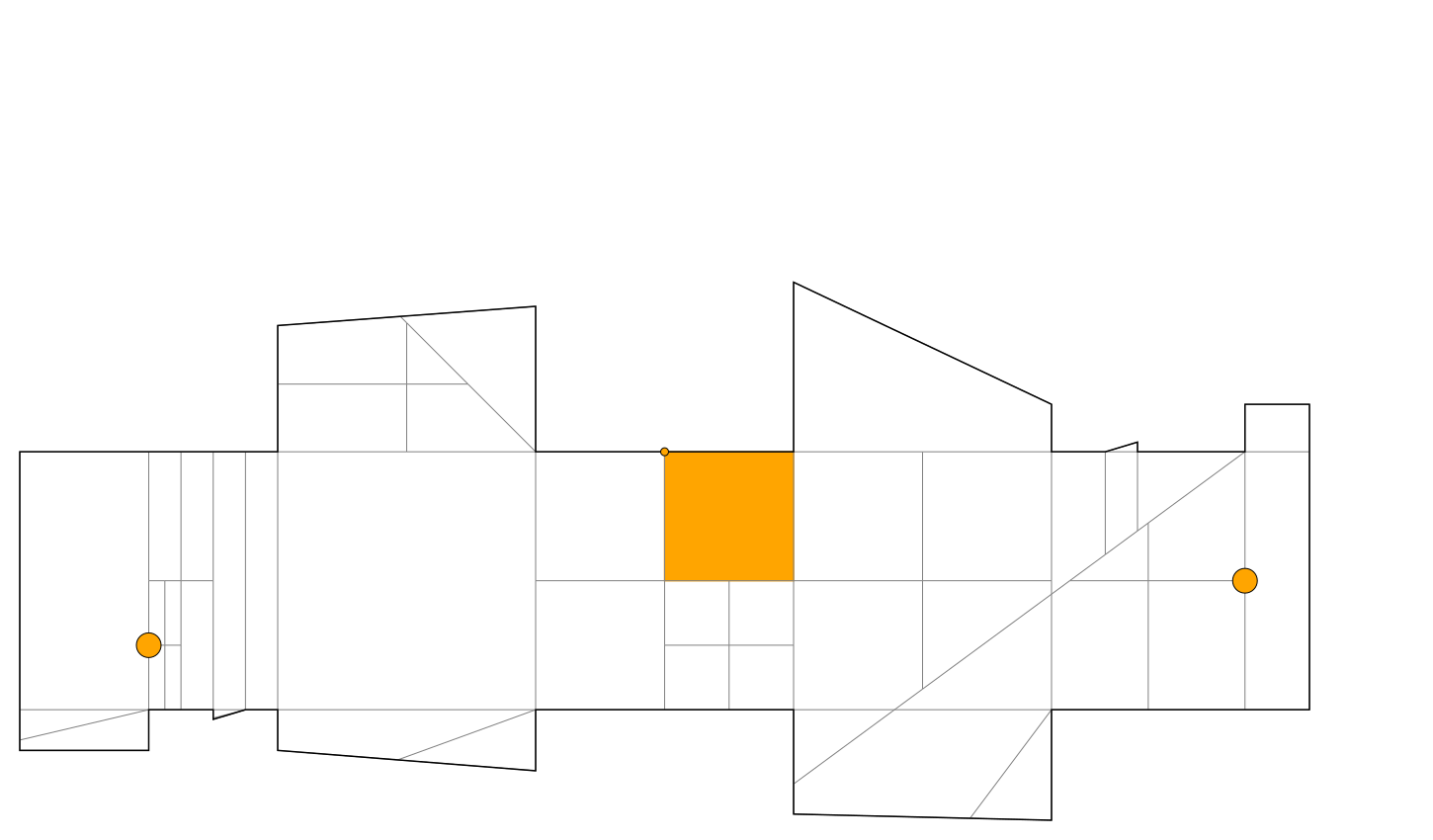}
	5
\end{minipage}
\hfill
\begin{minipage}[b]{0.48\textwidth}
	\centering
	\includegraphics[width=\textwidth]{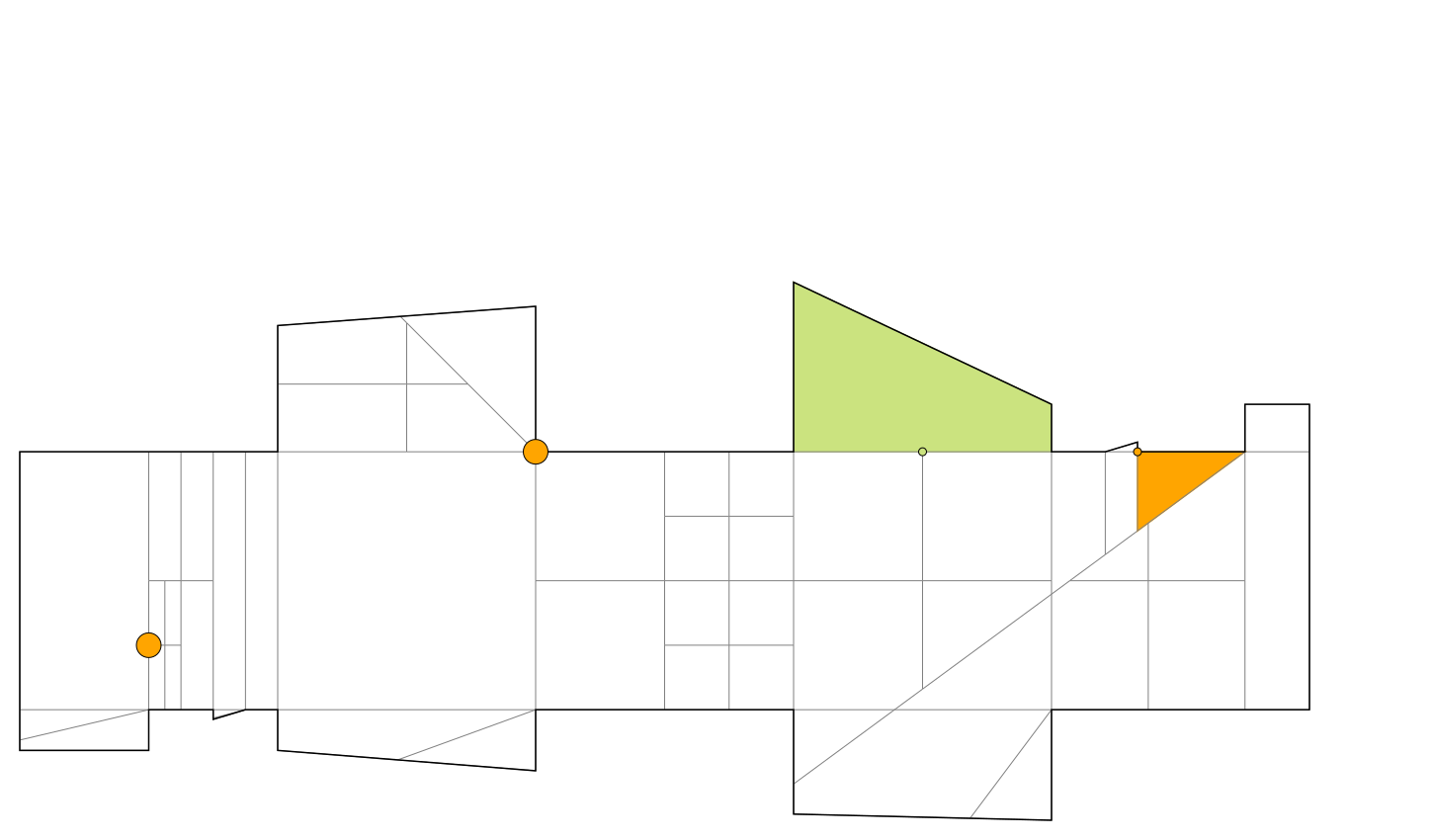}
	6
\end{minipage}
\begin{minipage}[b]{0.48\textwidth}
	\centering
	\includegraphics[width=\textwidth]{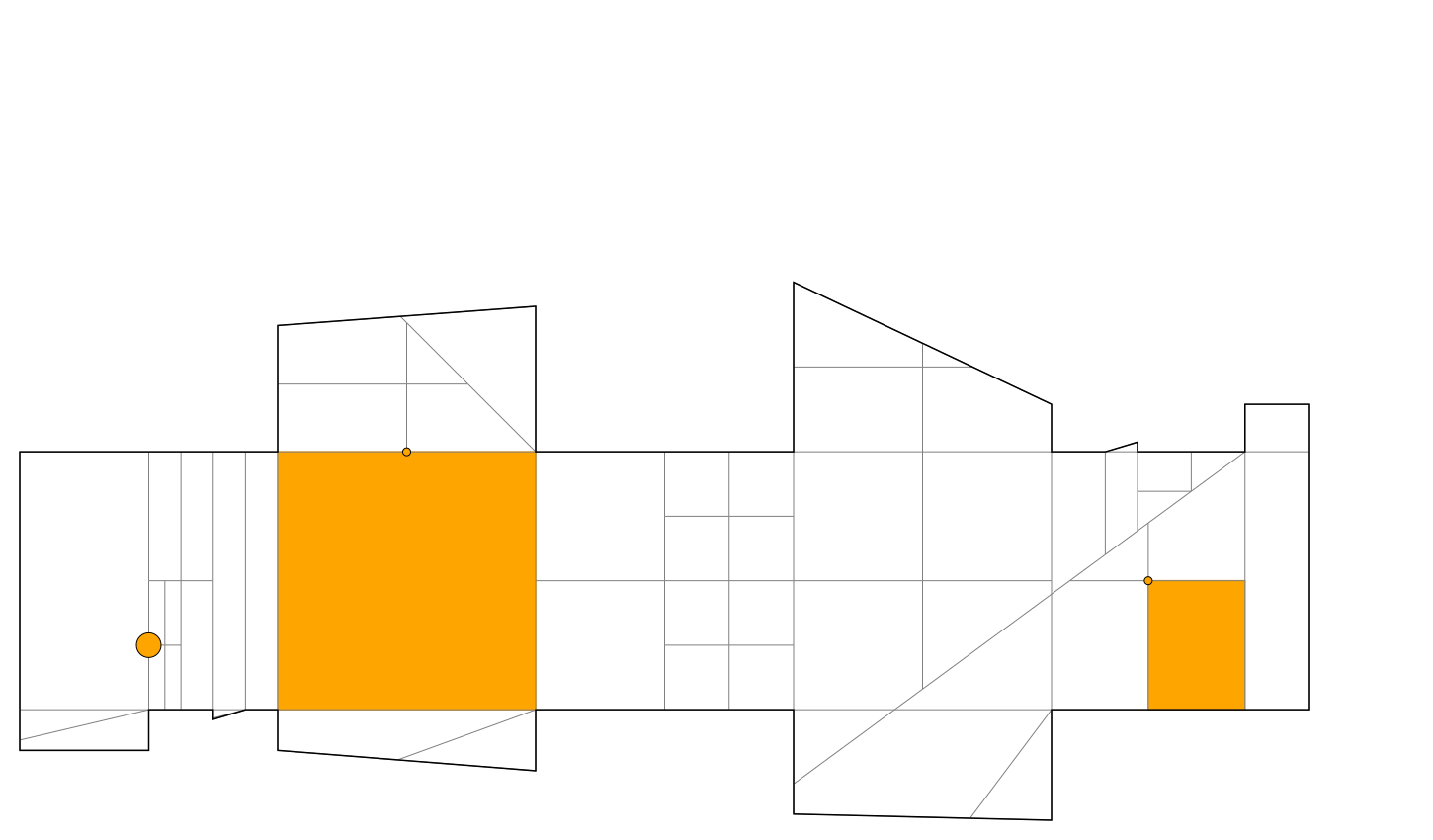}
	7
\end{minipage}
\hfill
\begin{minipage}[b]{0.48\textwidth}
	\centering
	\includegraphics[width=\textwidth]{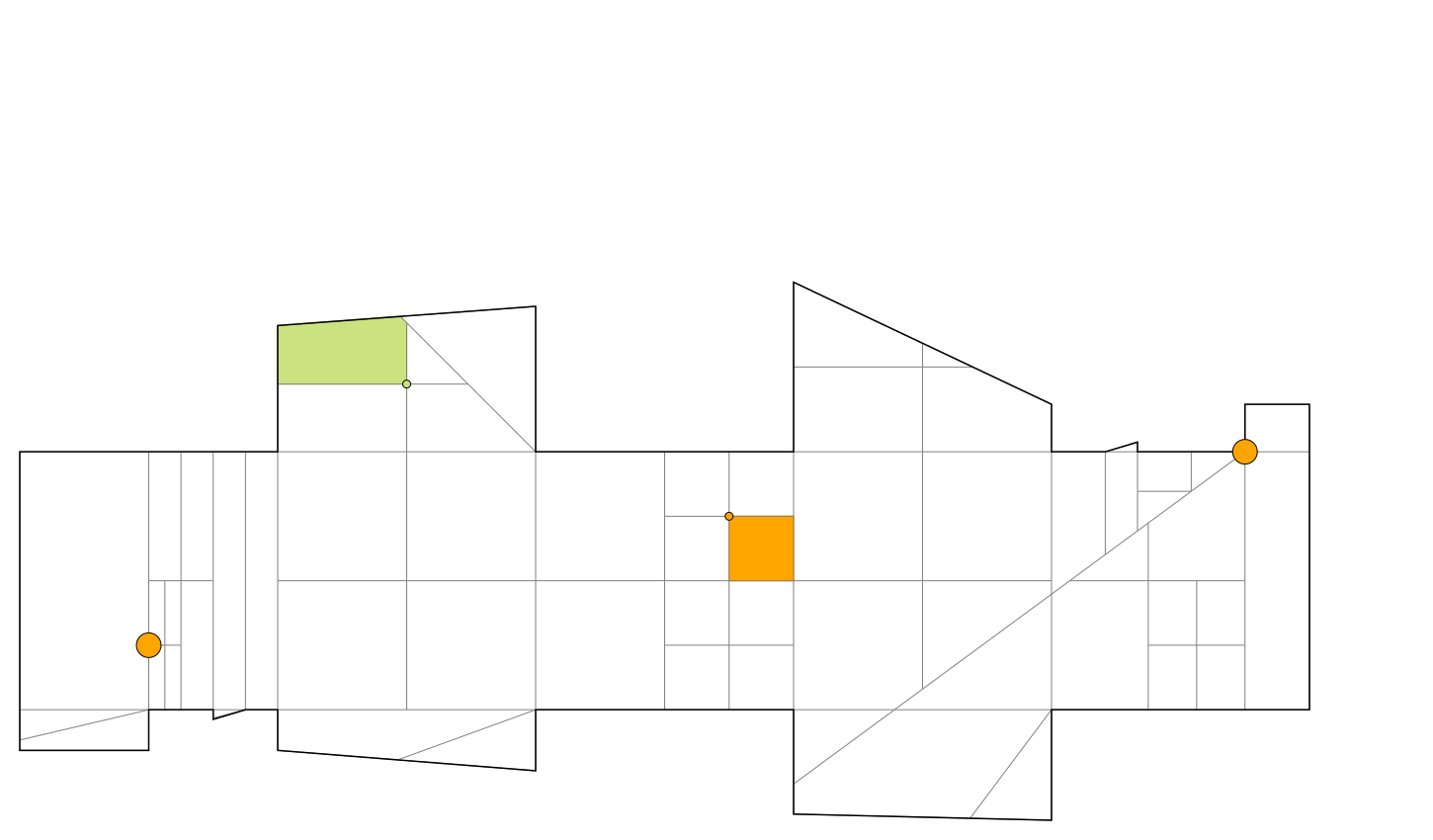}
	8
\end{minipage}
\begin{minipage}[b]{0.48\textwidth}
	\centering
	\includegraphics[width=\textwidth]{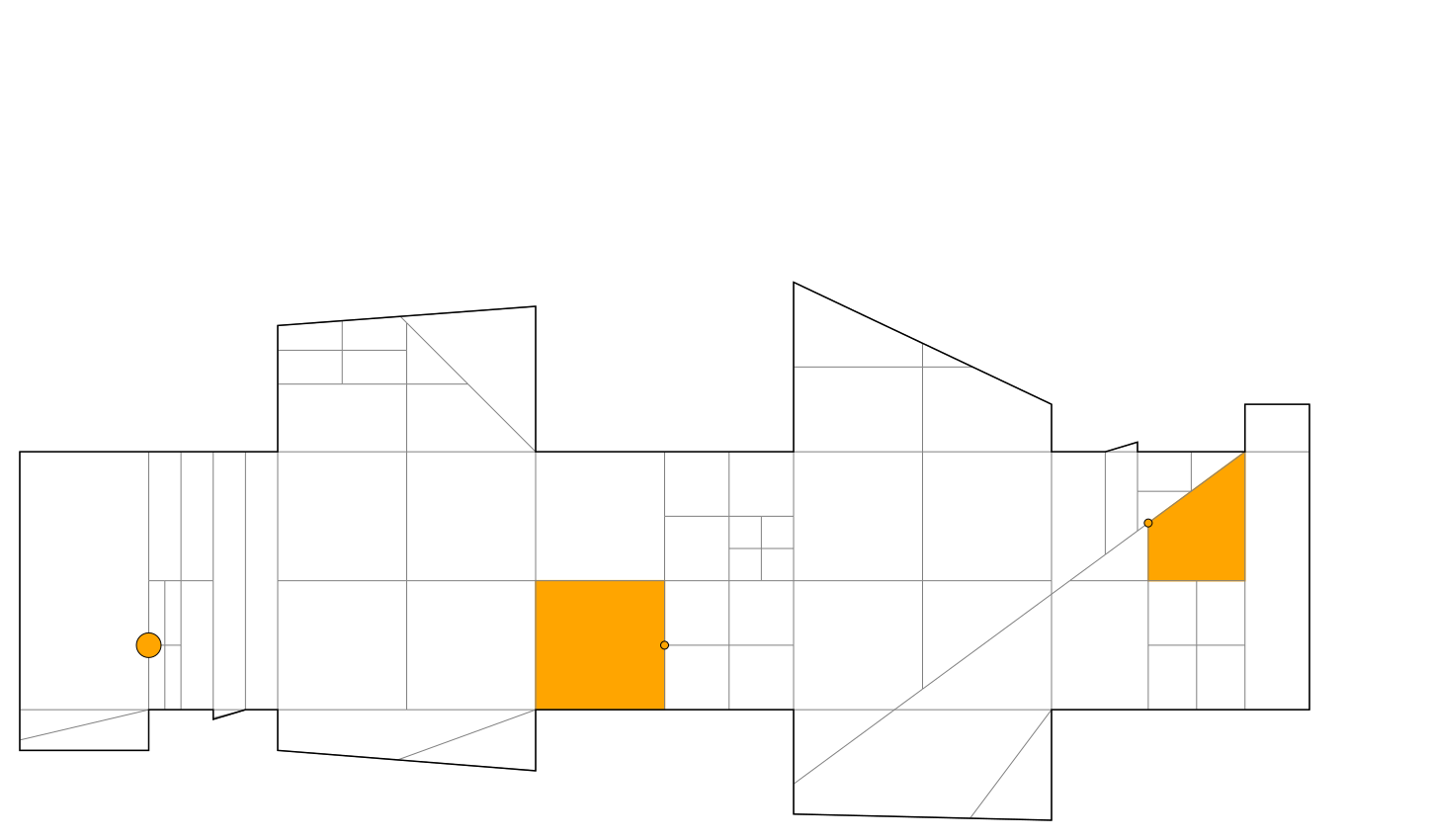}
	9
\end{minipage}
\hfill
\begin{minipage}[b]{0.48\textwidth}
	\centering
	\includegraphics[width=\textwidth]{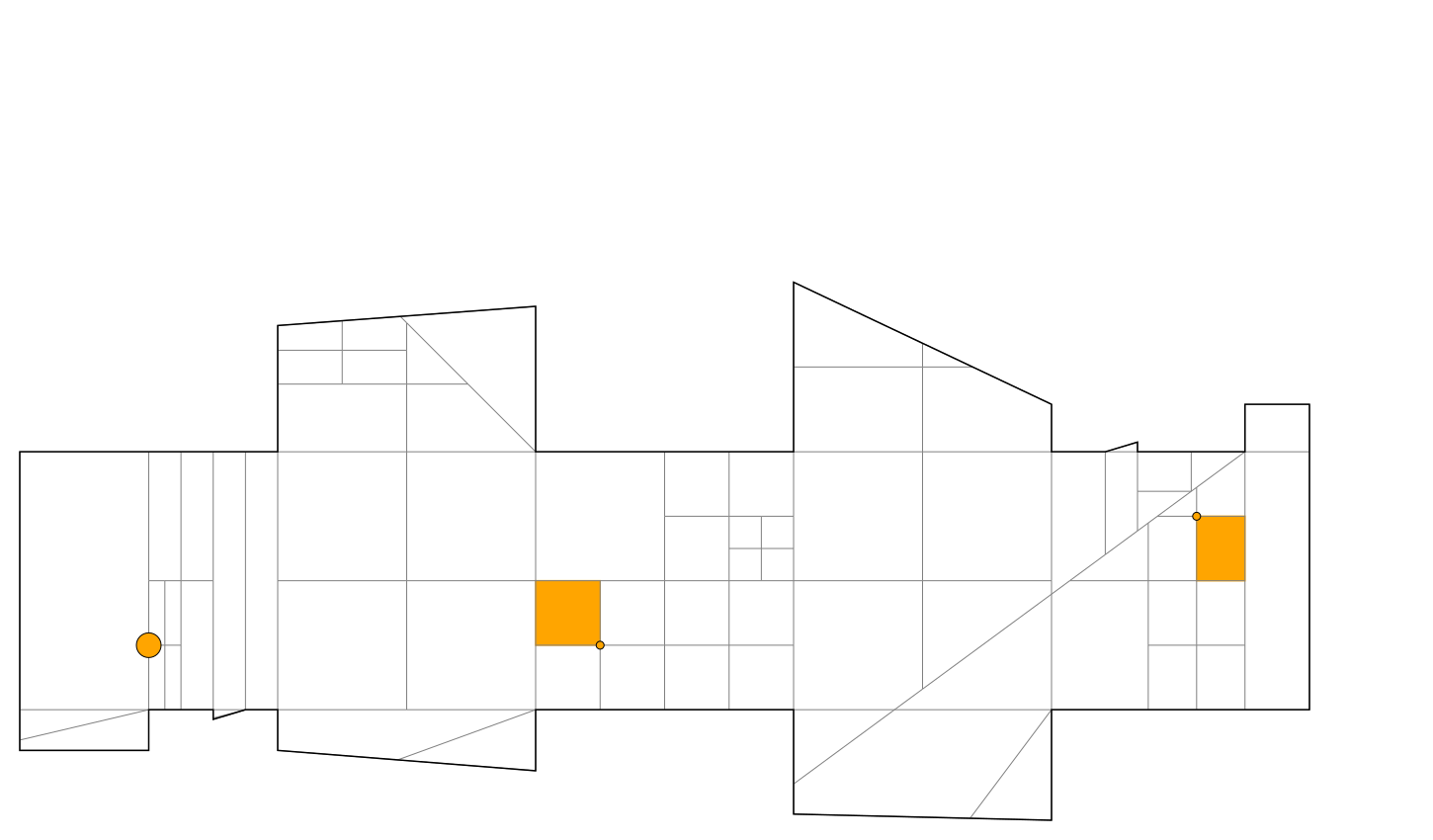}
	10
\end{minipage}
\begin{minipage}[b]{0.48\textwidth}
	\centering
	\includegraphics[width=\textwidth]{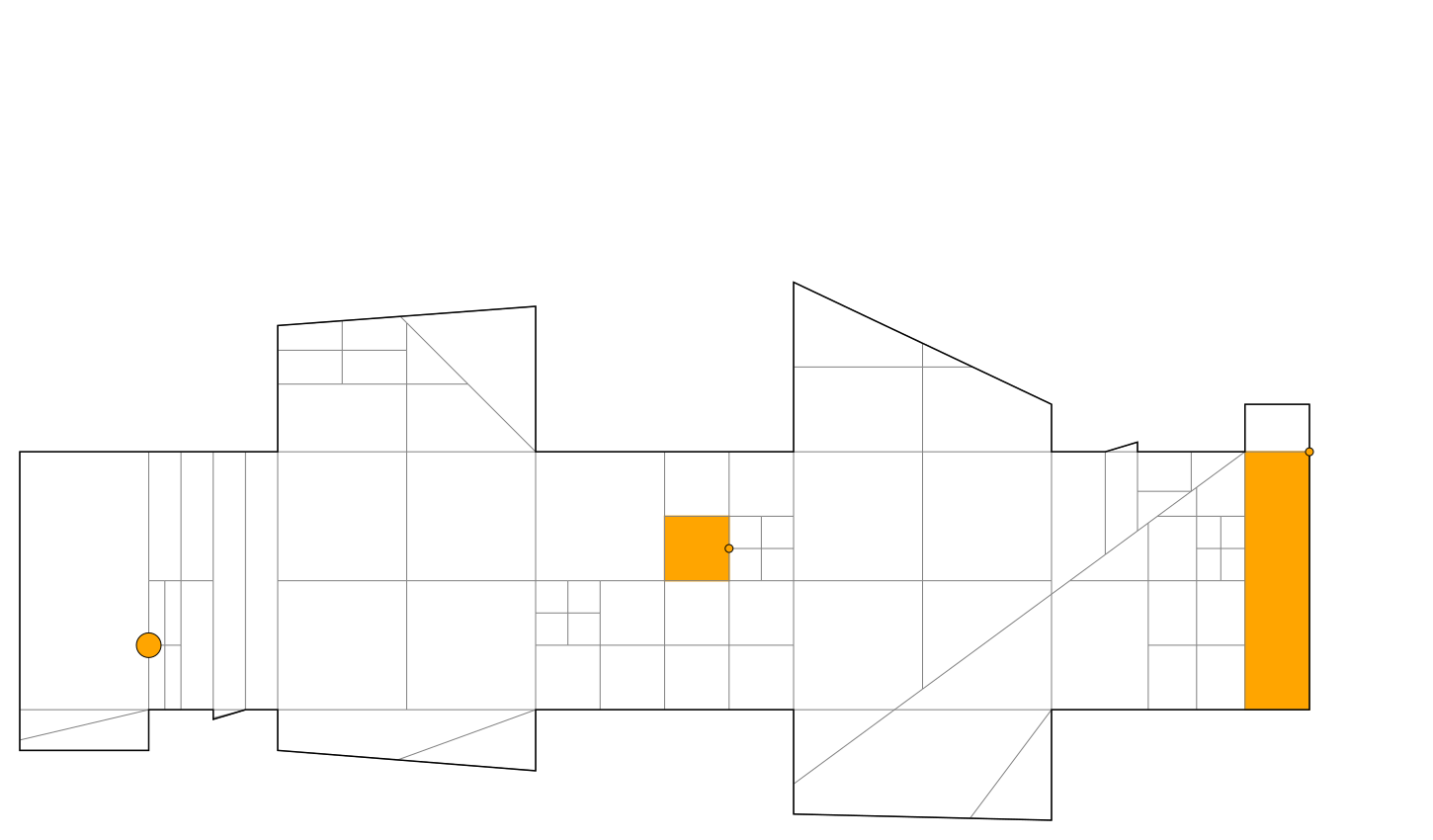}
	11
\end{minipage}
\hfill
\begin{minipage}[b]{0.48\textwidth}
	\centering
	\includegraphics[width=\textwidth]{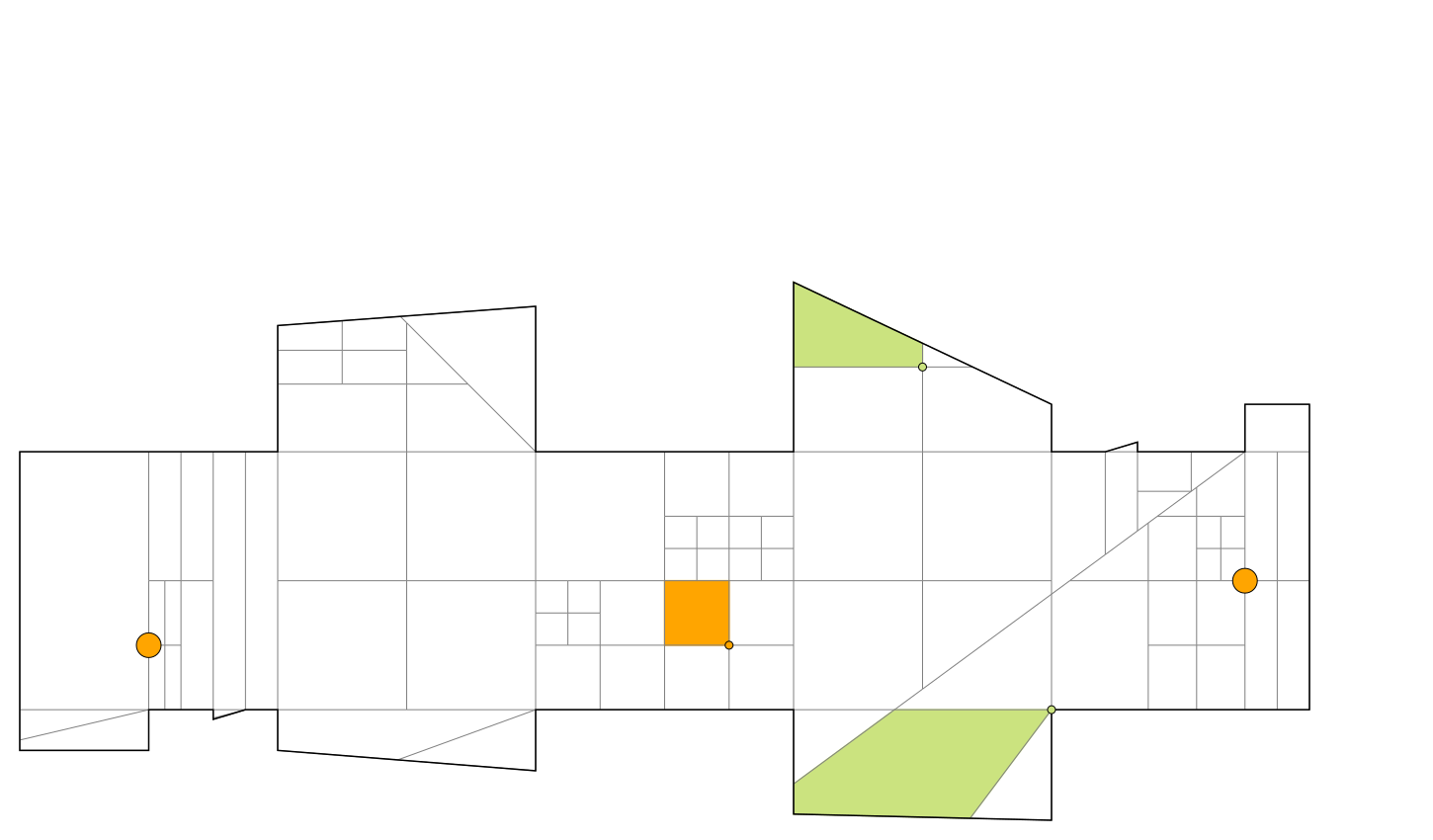}
	12
\end{minipage}
\begin{minipage}[b]{0.48\textwidth}
	\centering
	\includegraphics[width=\textwidth]{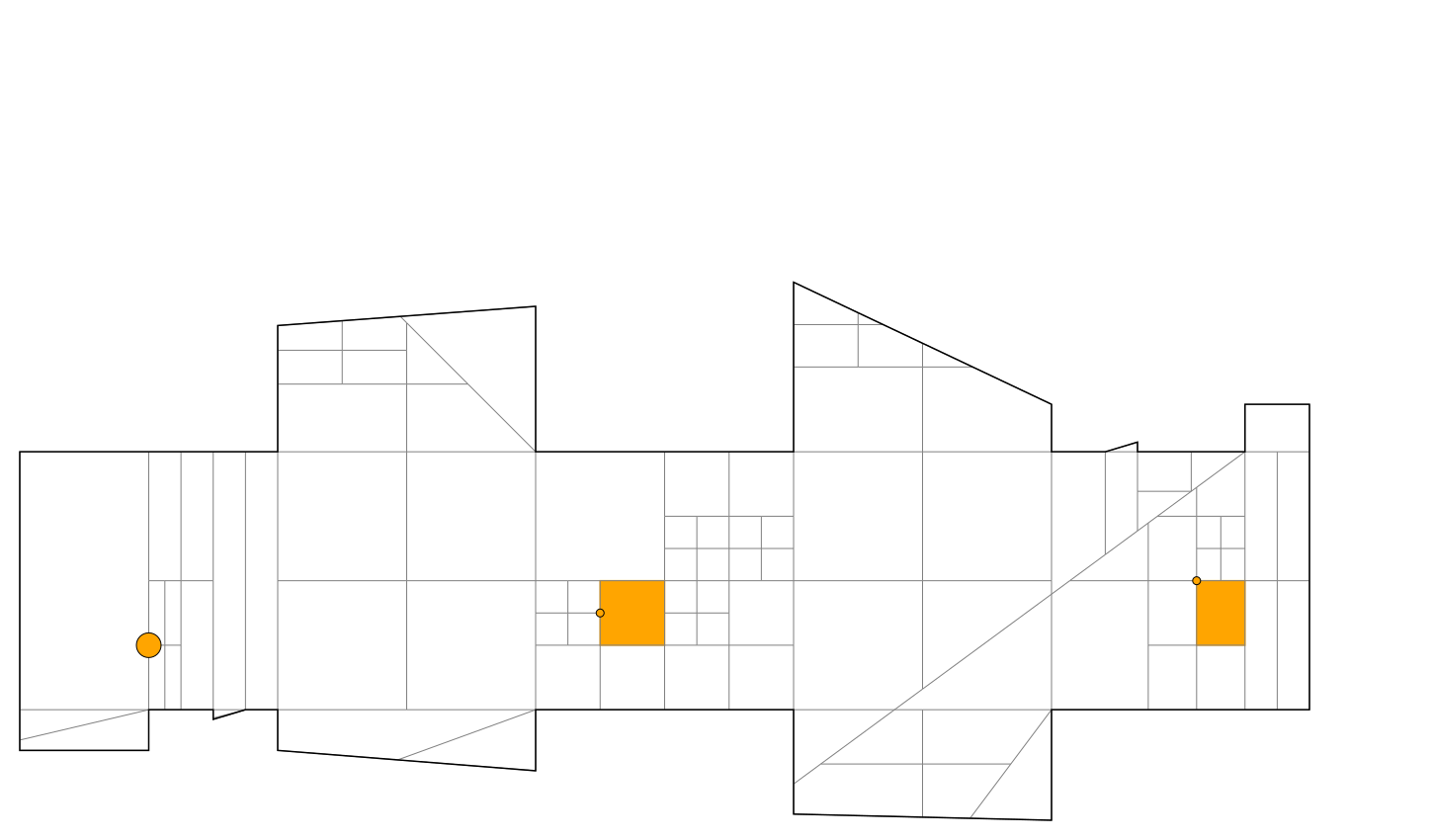}
	13
\end{minipage}
\hfill
\begin{minipage}[b]{0.48\textwidth}
	\centering
	\includegraphics[width=\textwidth]{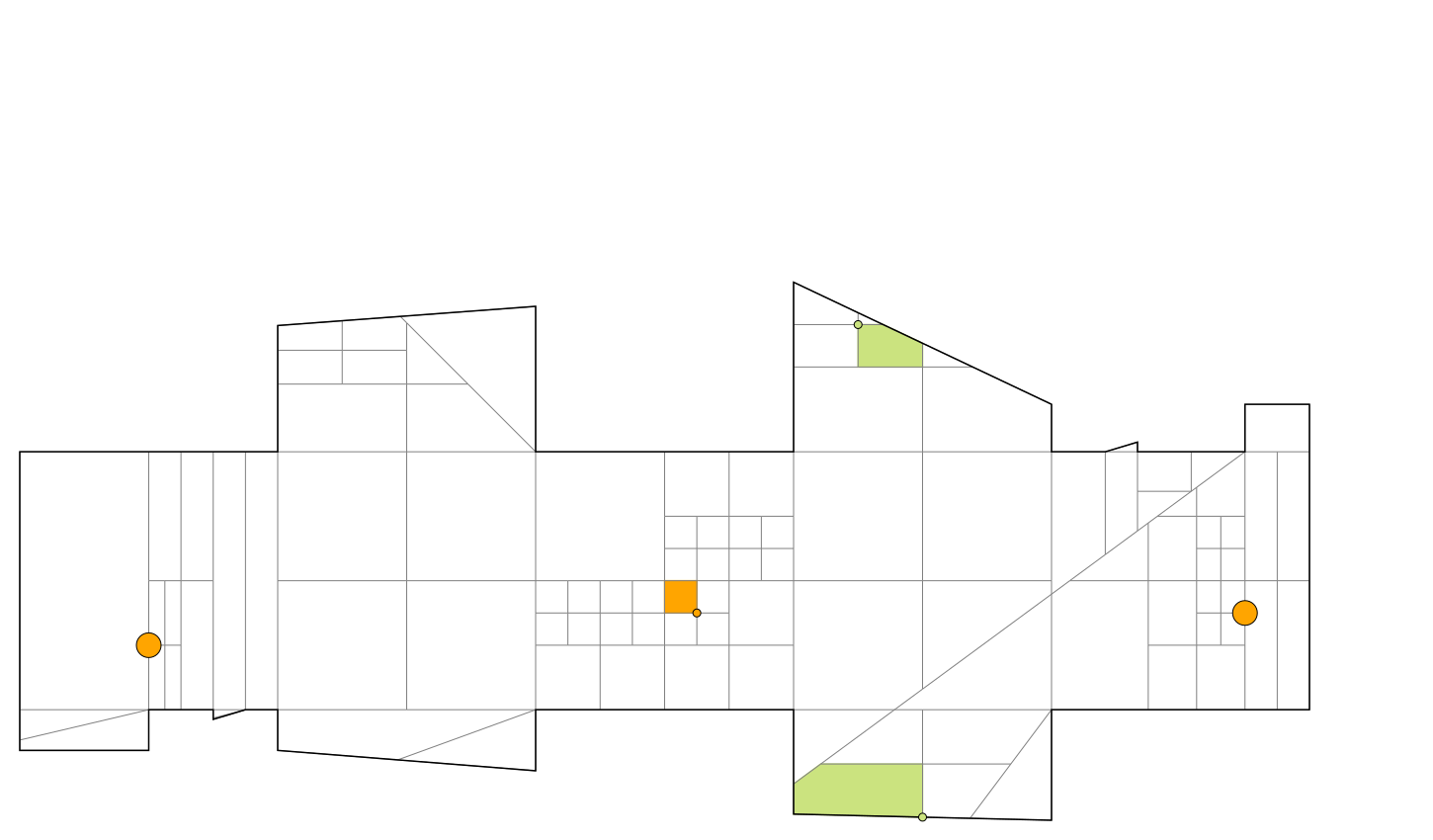}
	14
\end{minipage}
\begin{minipage}[b]{0.48\textwidth}
	\centering
	\includegraphics[width=\textwidth]{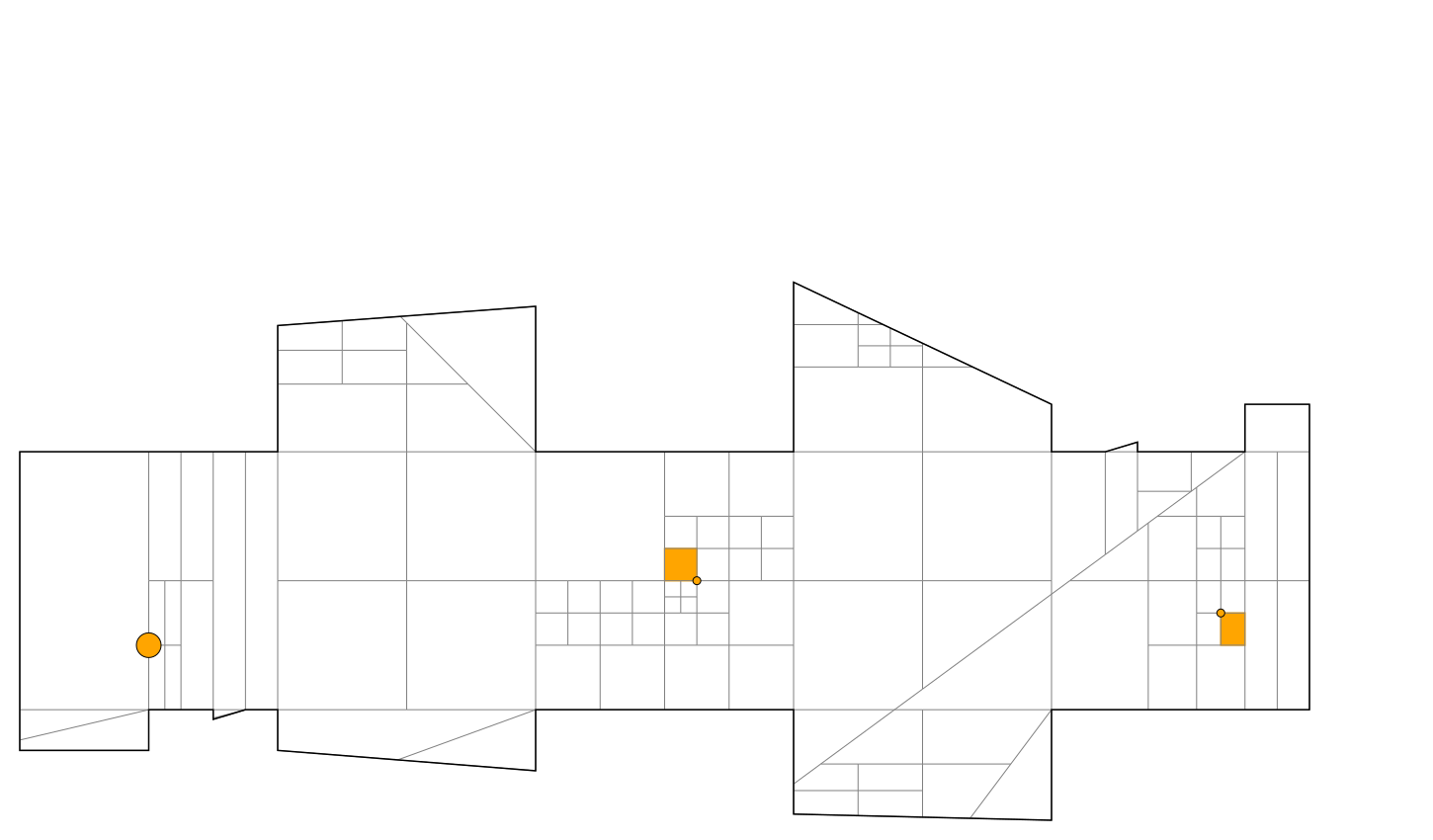}
	15
\end{minipage}
\hfill
\begin{minipage}[b]{0.48\textwidth}
	\centering
	\includegraphics[width=\textwidth]{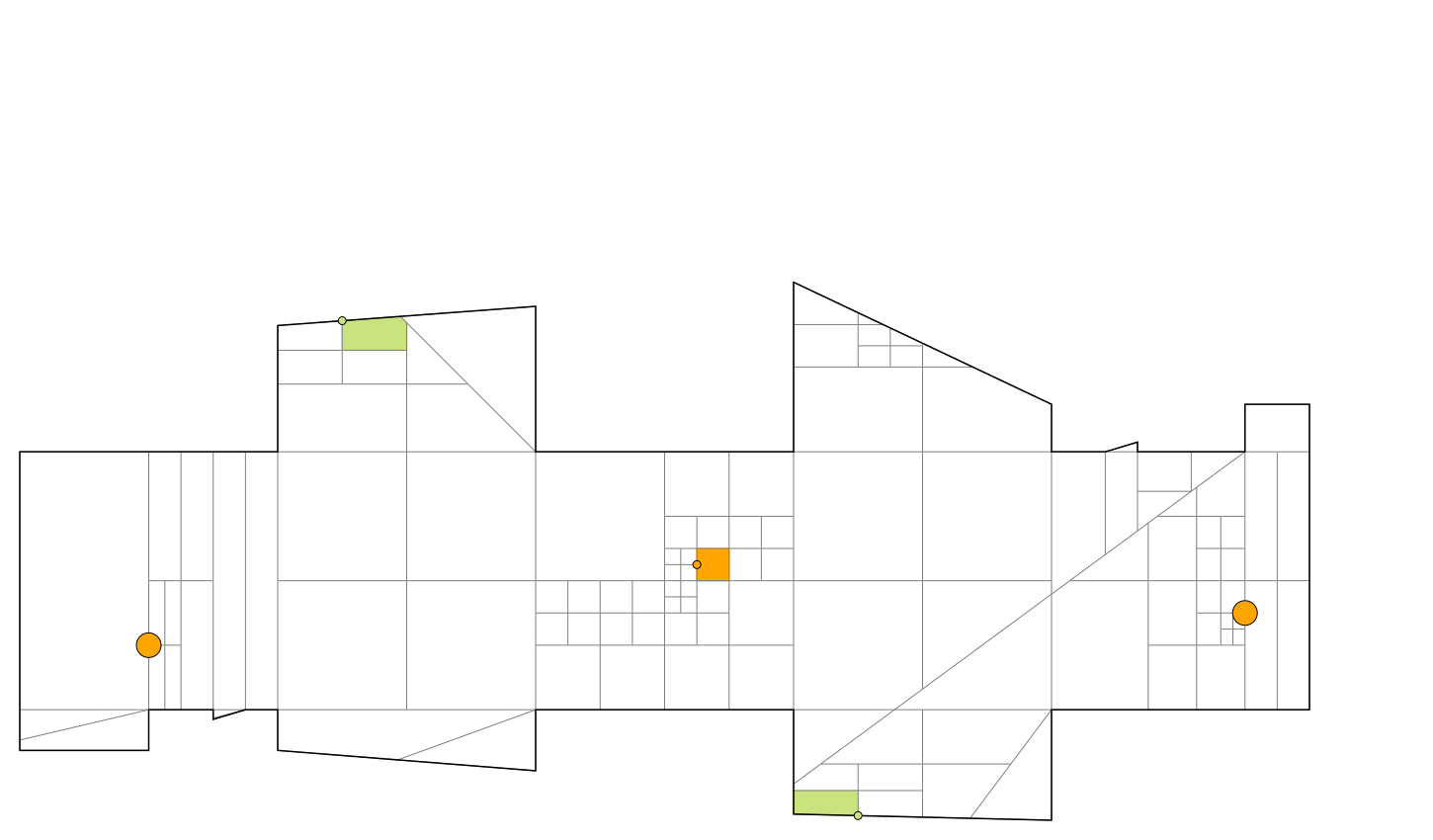}
	16
\end{minipage}
\begin{minipage}[b]{0.48\textwidth}
	\centering
	\includegraphics[width=\textwidth]{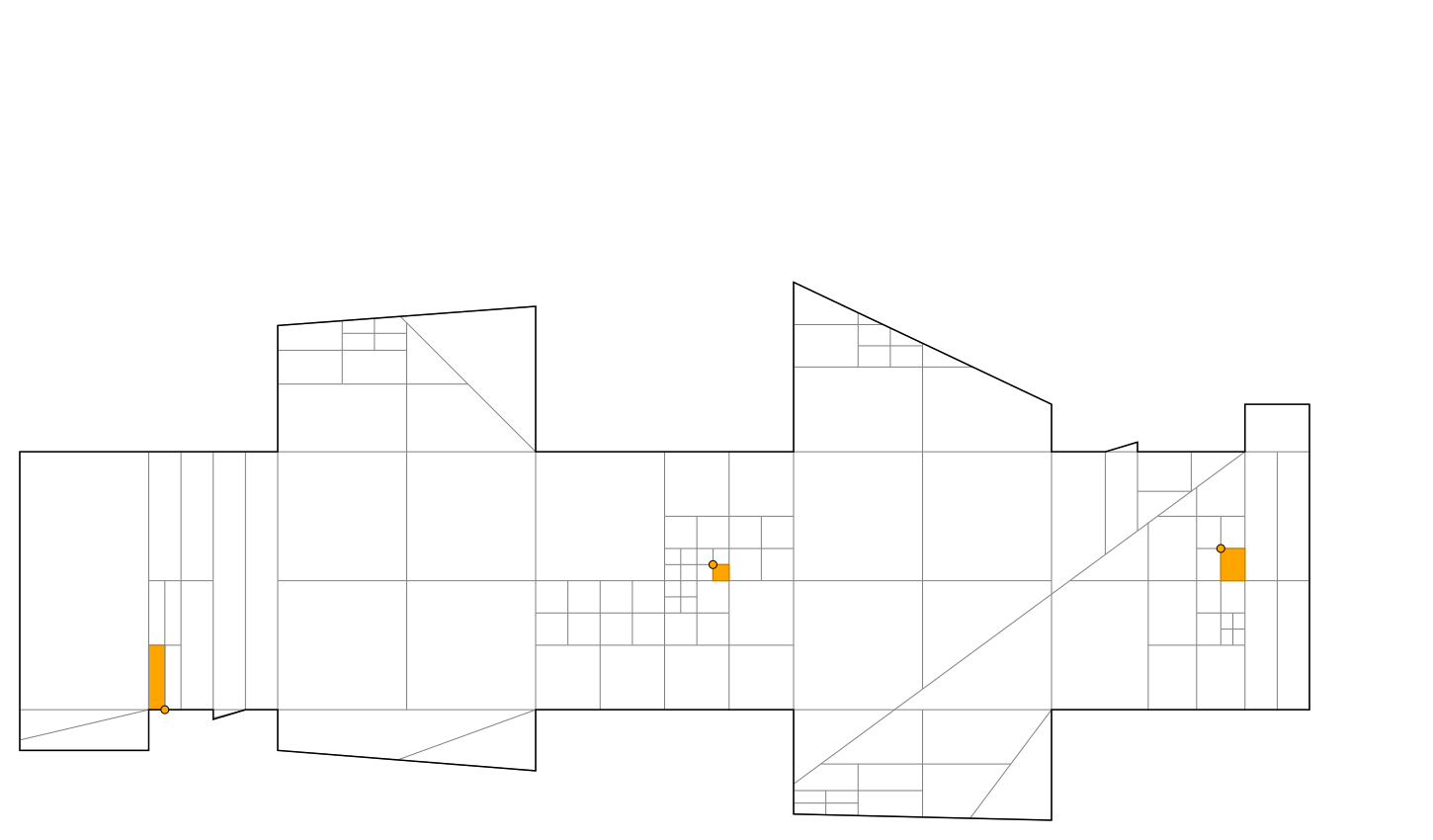}
	17
\end{minipage}
\hfill
\begin{minipage}[b]{0.48\textwidth}
	\centering
	\includegraphics[width=\textwidth]{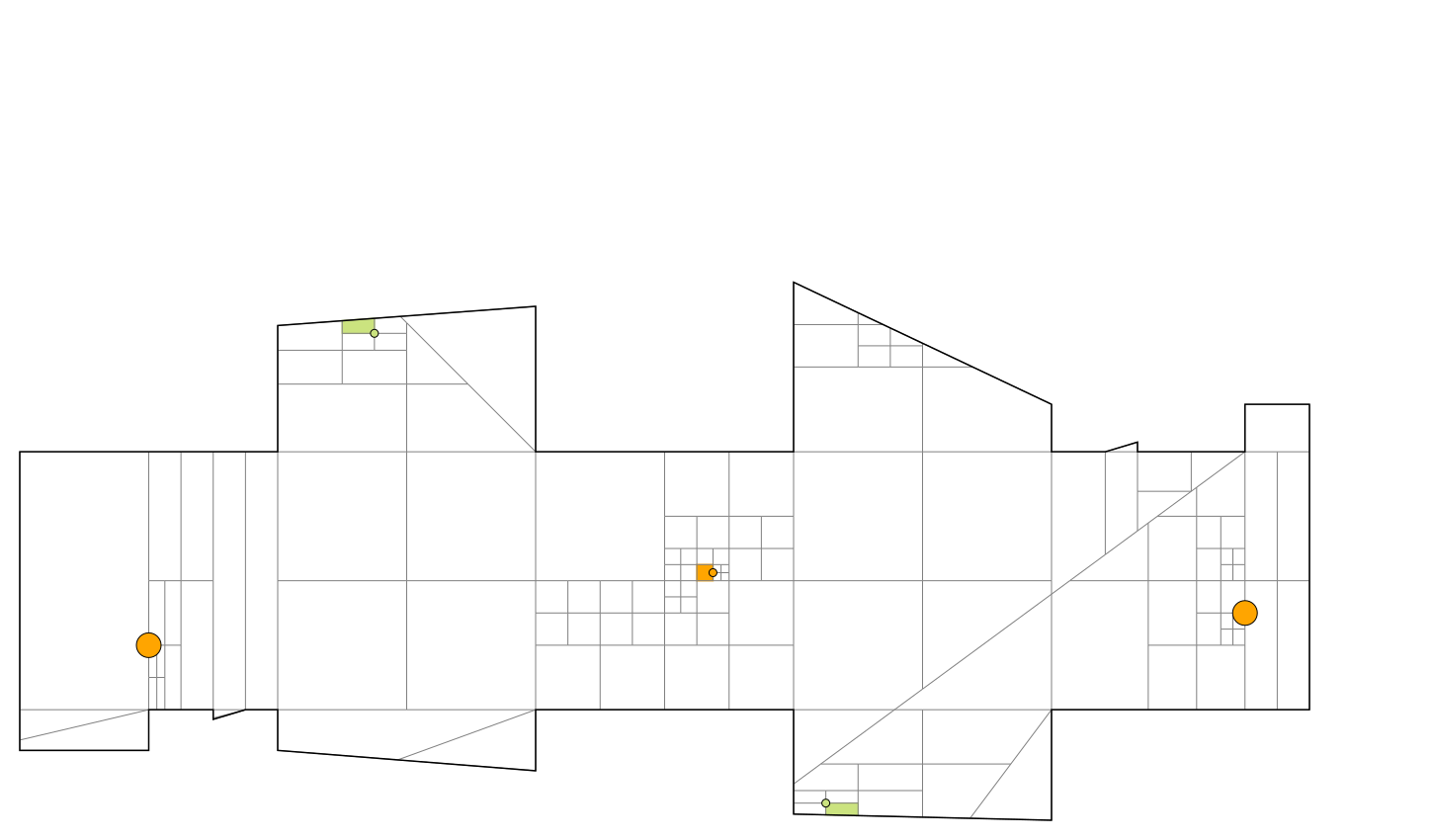}
	18
\end{minipage}
\begin{minipage}[b]{0.48\textwidth}
	\centering
	\includegraphics[width=\textwidth]{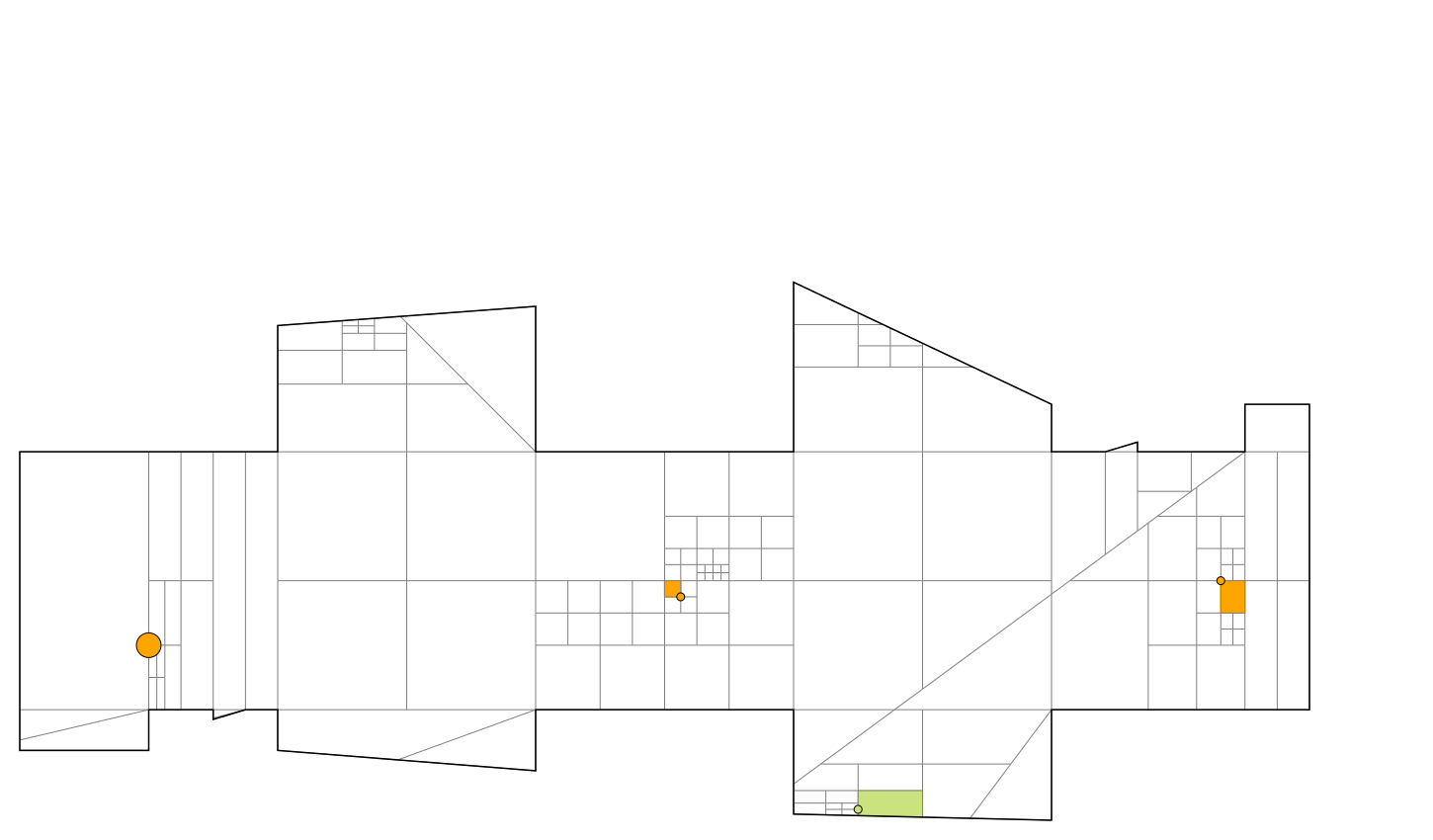}
	19
\end{minipage}
\hfill
\begin{minipage}[b]{0.48\textwidth}
	\centering
	\includegraphics[width=\textwidth]{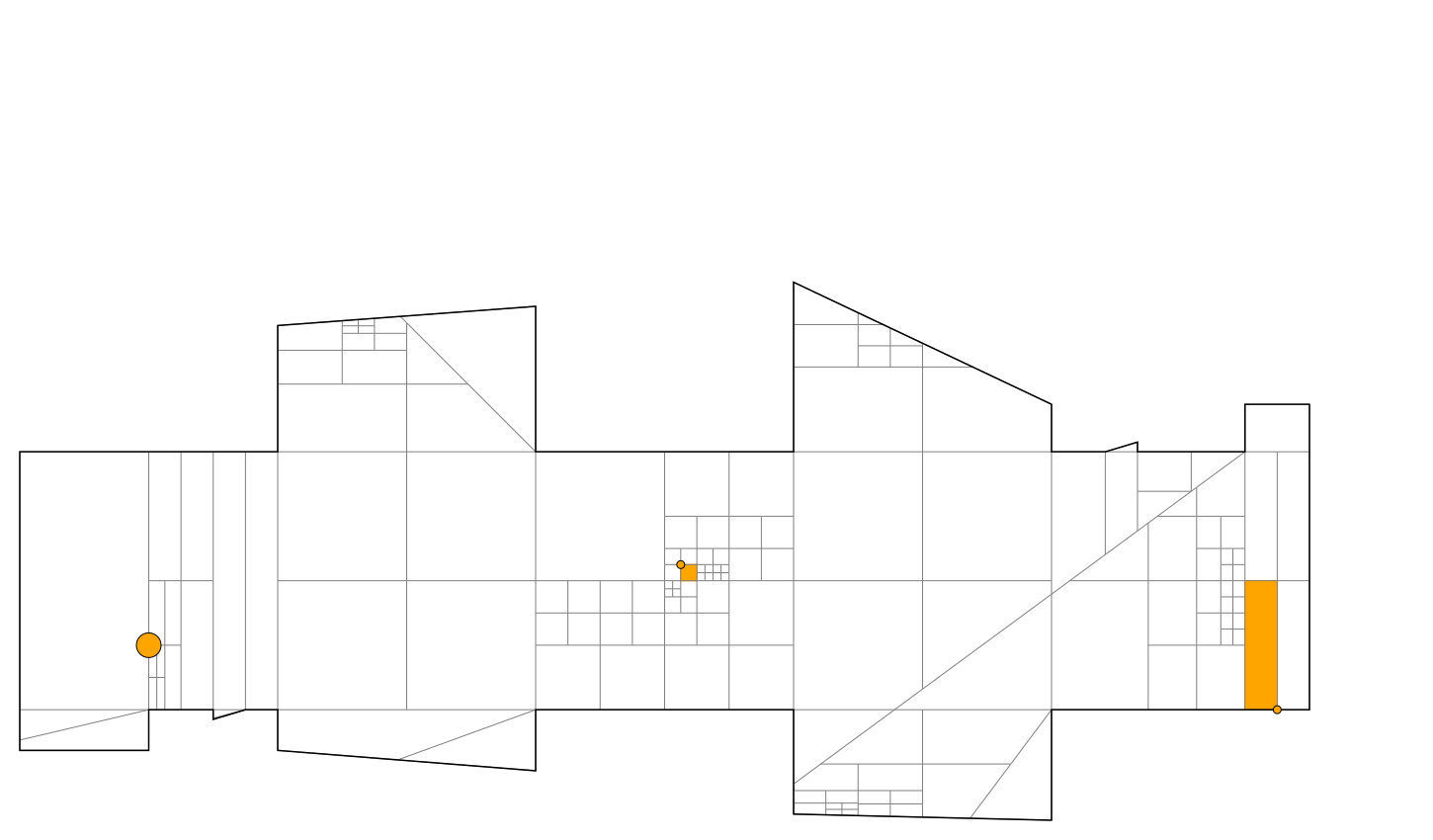}
	20
\end{minipage}

\paragraph*{Polygon with 60 vertices.}
Here, we show how the Iterative Algorithm without Safe Guards finds a solution in~9 iterations for a polygon with~60 vertices. 
This is one of the input polygons from Couto et al.~\cite{art-gallery-instances-page}. 
The orange coloured vertices and faces represent guards, and a green face is a witness faces.
\\
\begin{minipage}[b]{0.48\textwidth}
	\centering
	
	\includegraphics[width=\textwidth]{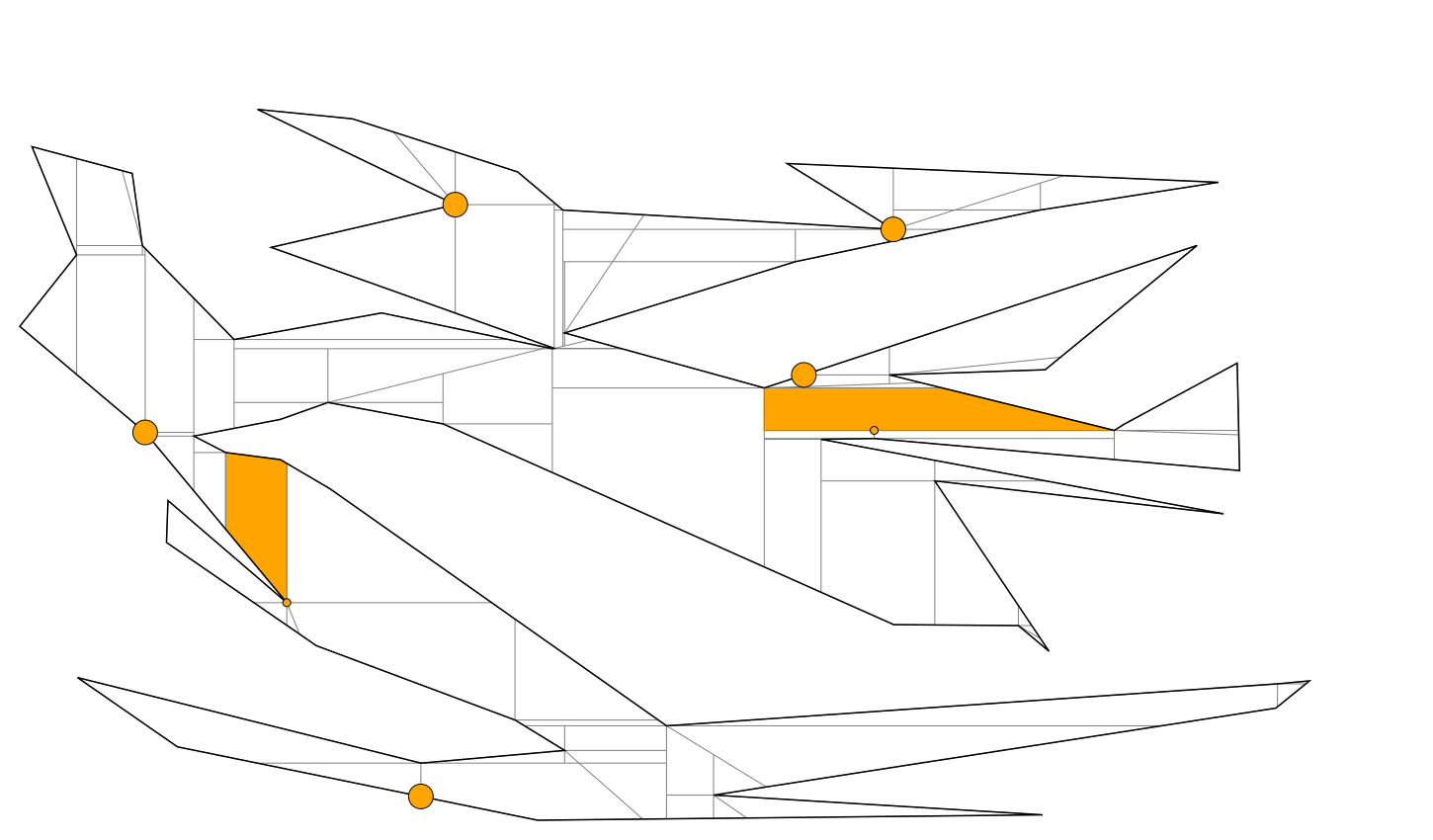}
	1
\end{minipage}
\hfill
\begin{minipage}[b]{0.48\textwidth}
	\centering
	
	\includegraphics[width=\textwidth]{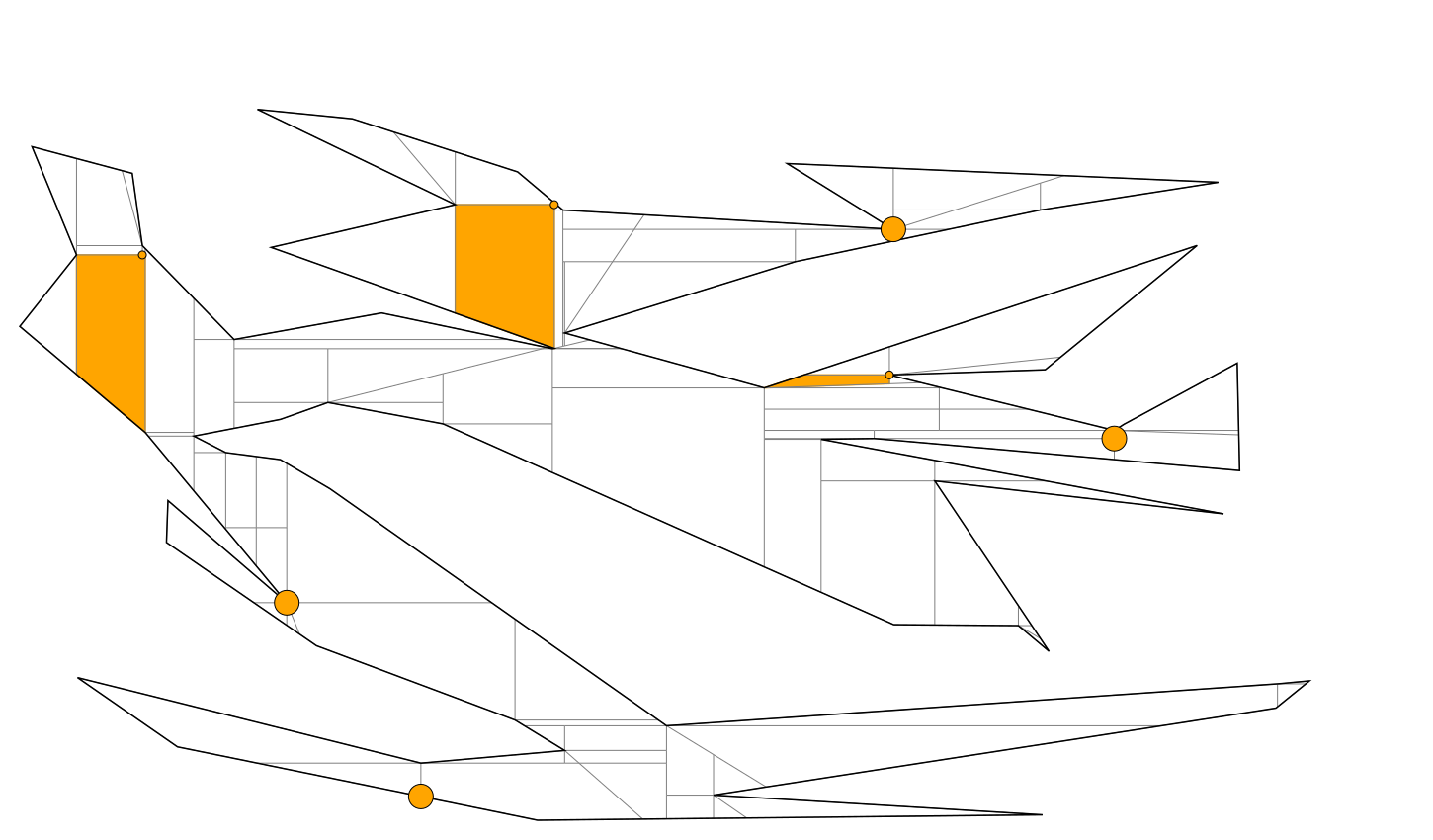}
	2
\end{minipage}
\begin{minipage}[b]{0.48\textwidth}
	\centering
	
	\includegraphics[width=\textwidth]{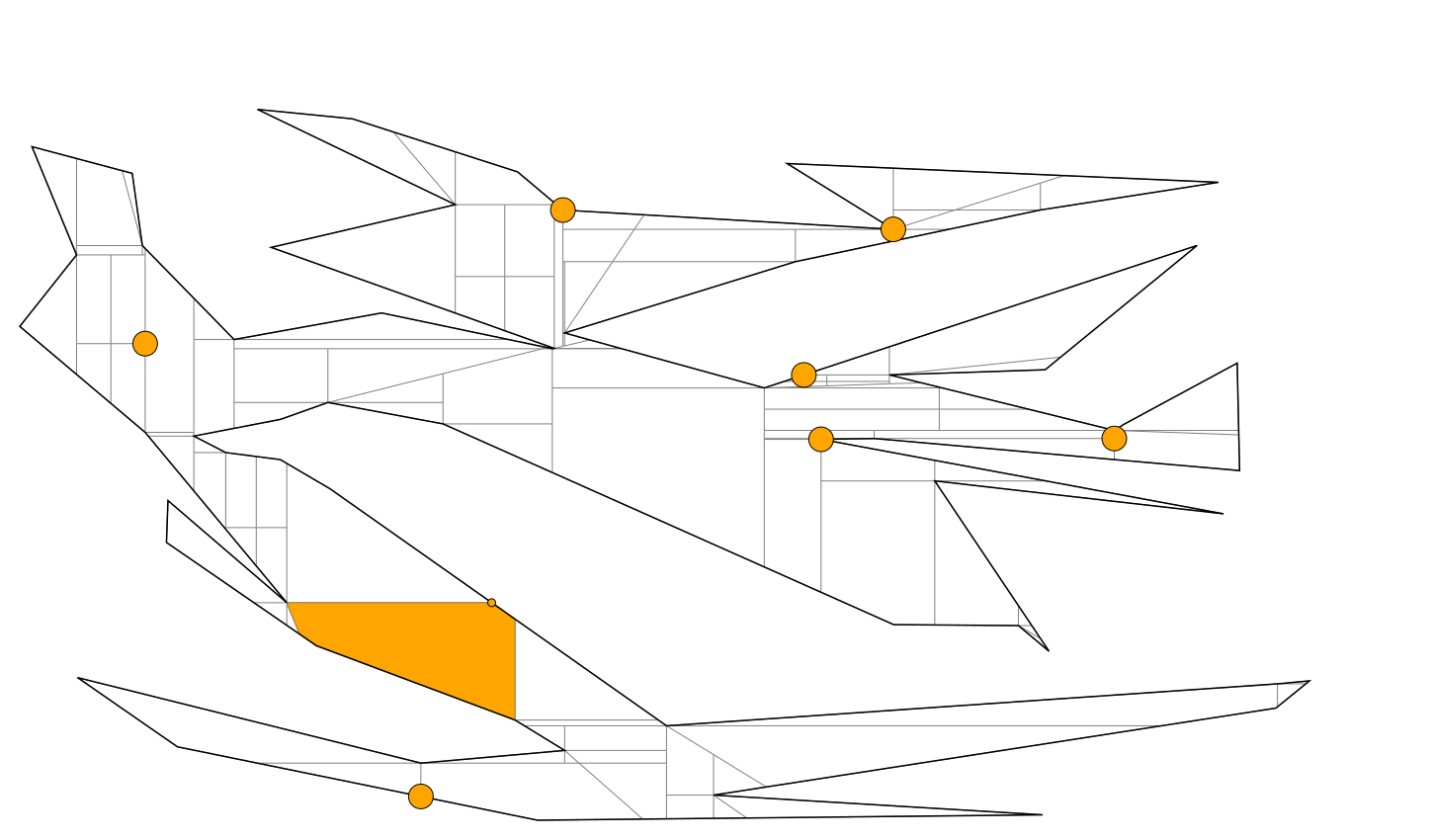}
	3
\end{minipage}
\hfill
\begin{minipage}[b]{0.48\textwidth}
	\centering
	
	\includegraphics[width=\textwidth]{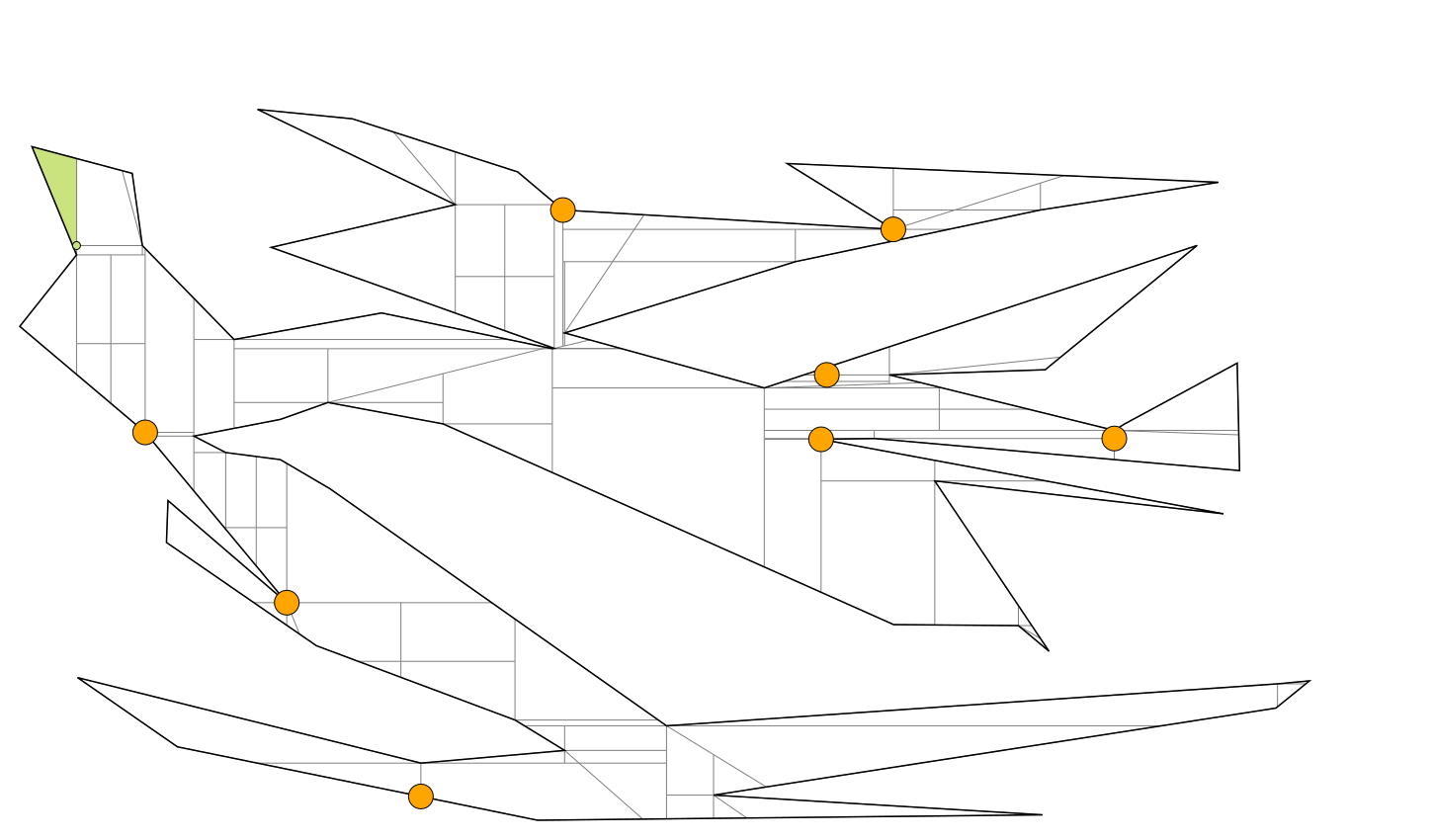}
	4
\end{minipage}
\begin{minipage}[b]{0.48\textwidth}
	\centering
	
	\includegraphics[width=\textwidth]{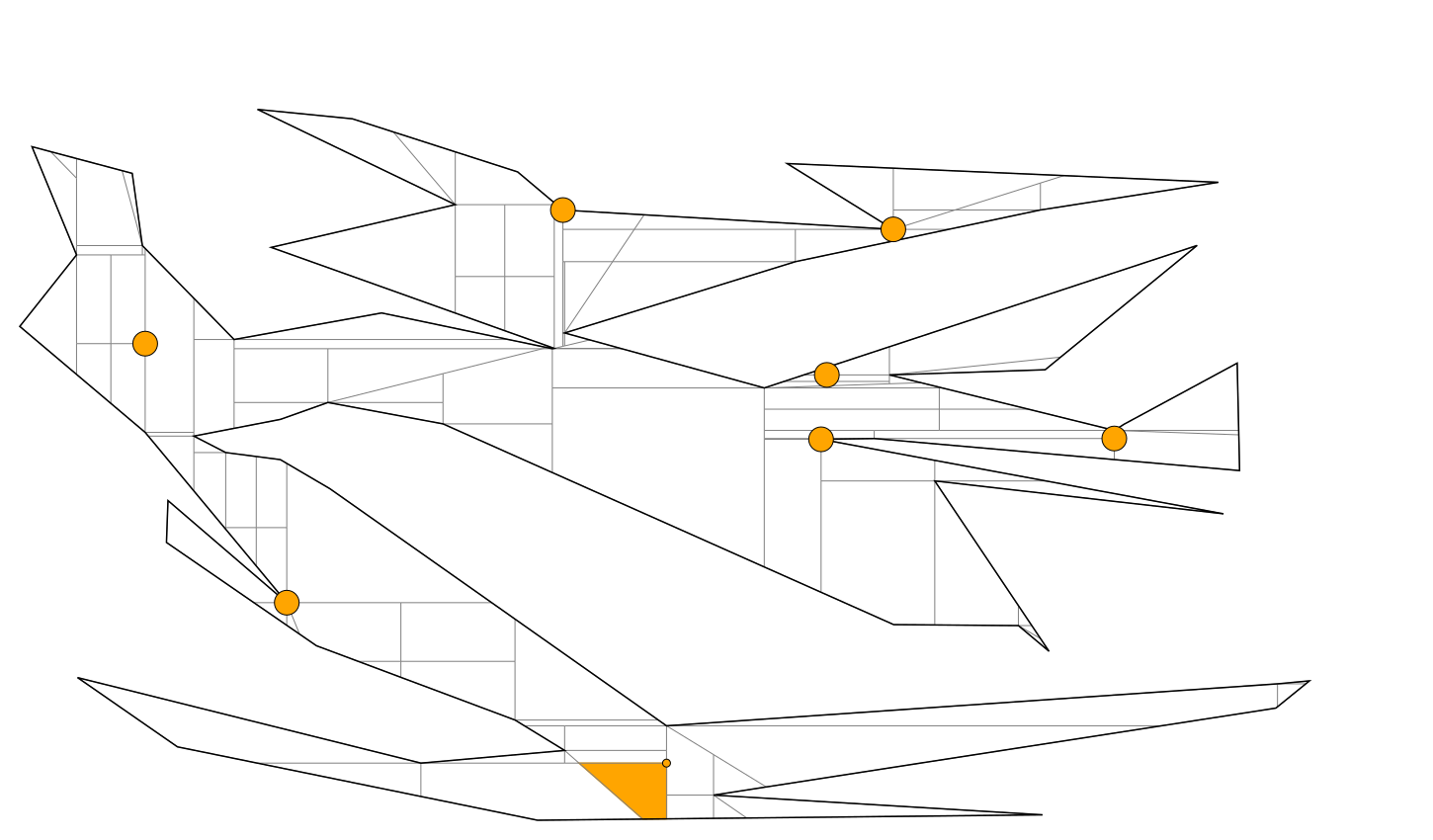}
	5
\end{minipage}
\hfill
\begin{minipage}[b]{0.48\textwidth}
	\centering
	
	\includegraphics[width=\textwidth]{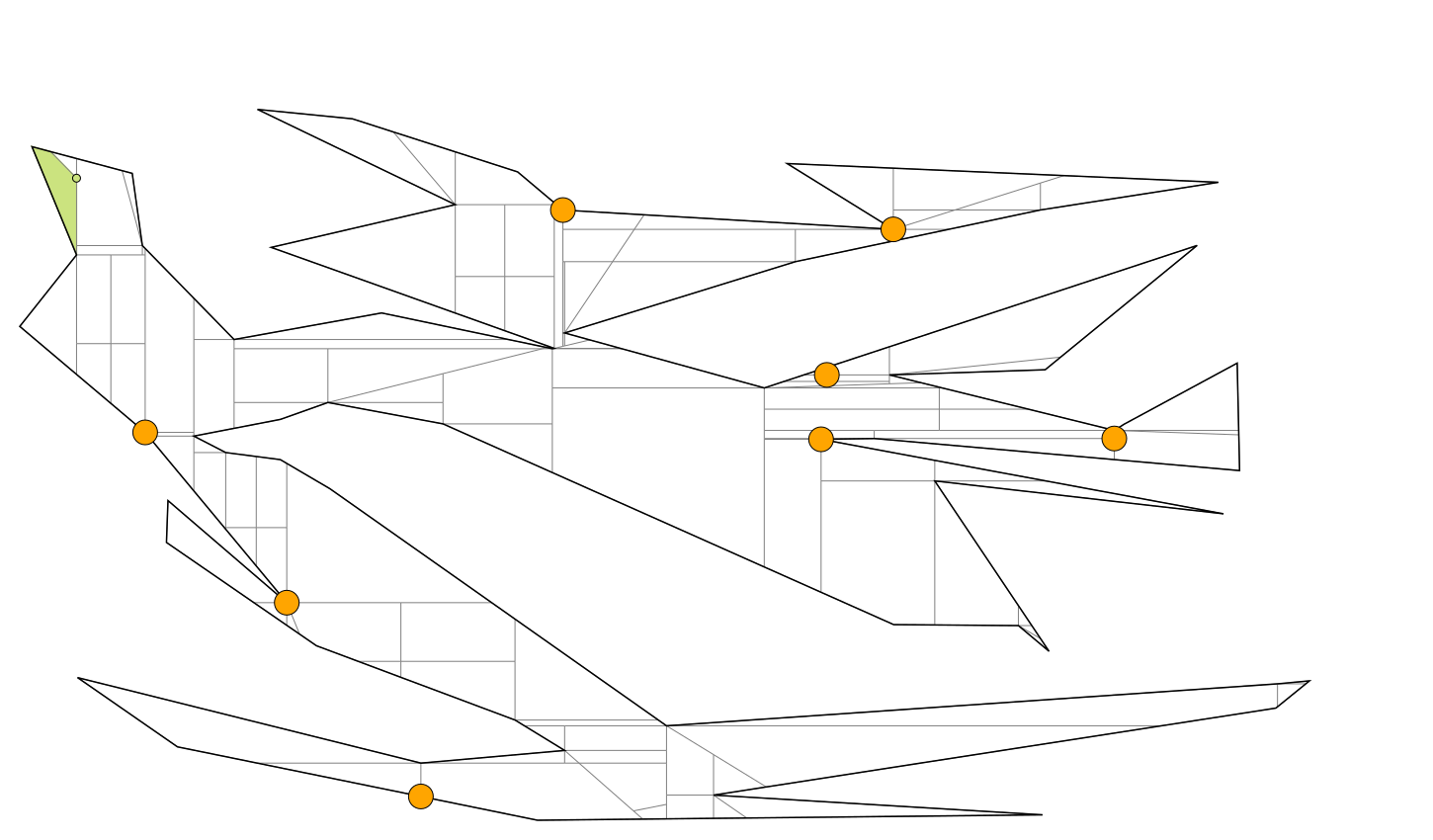}
	6
\end{minipage}
\begin{minipage}[b]{0.48\textwidth}
	\centering
	
	\includegraphics[width=\textwidth]{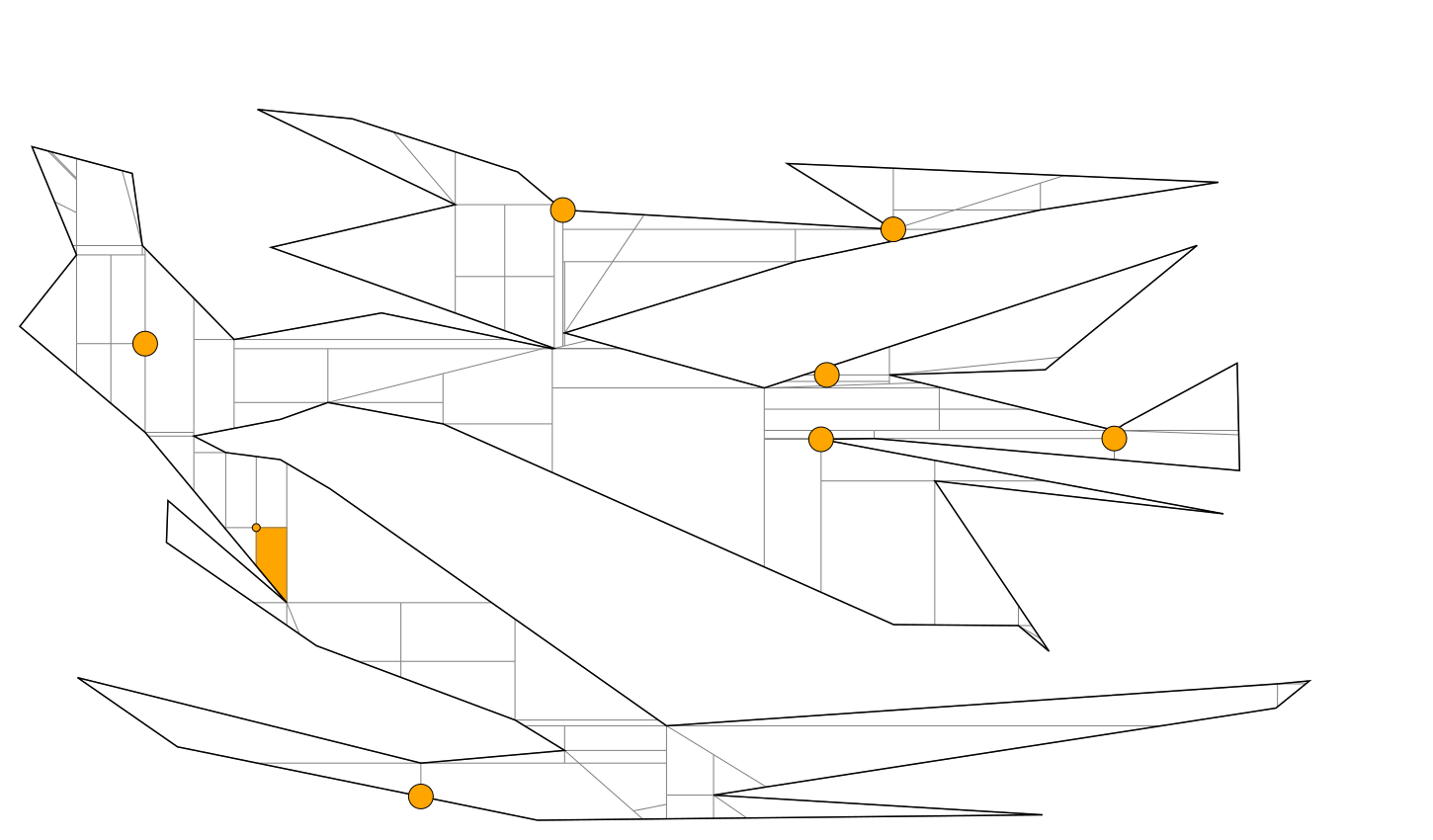}
	7
\end{minipage}
\hfill
\begin{minipage}[b]{0.48\textwidth}
	\centering
	
	\includegraphics[width=\textwidth]{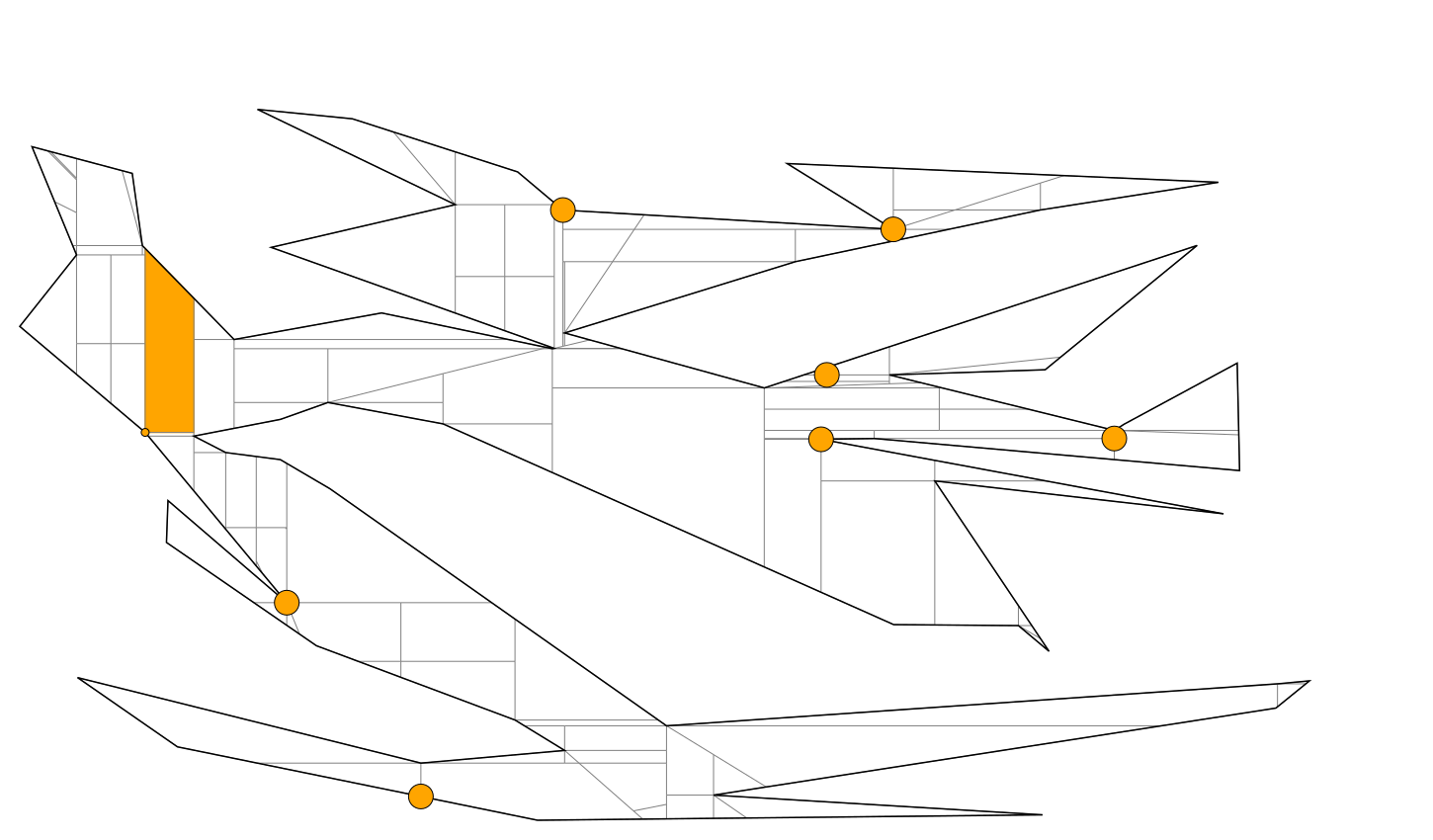}
	8
\end{minipage}
\begin{minipage}[b]{0.48\textwidth}
	\centering
	
	\includegraphics[width=\textwidth]{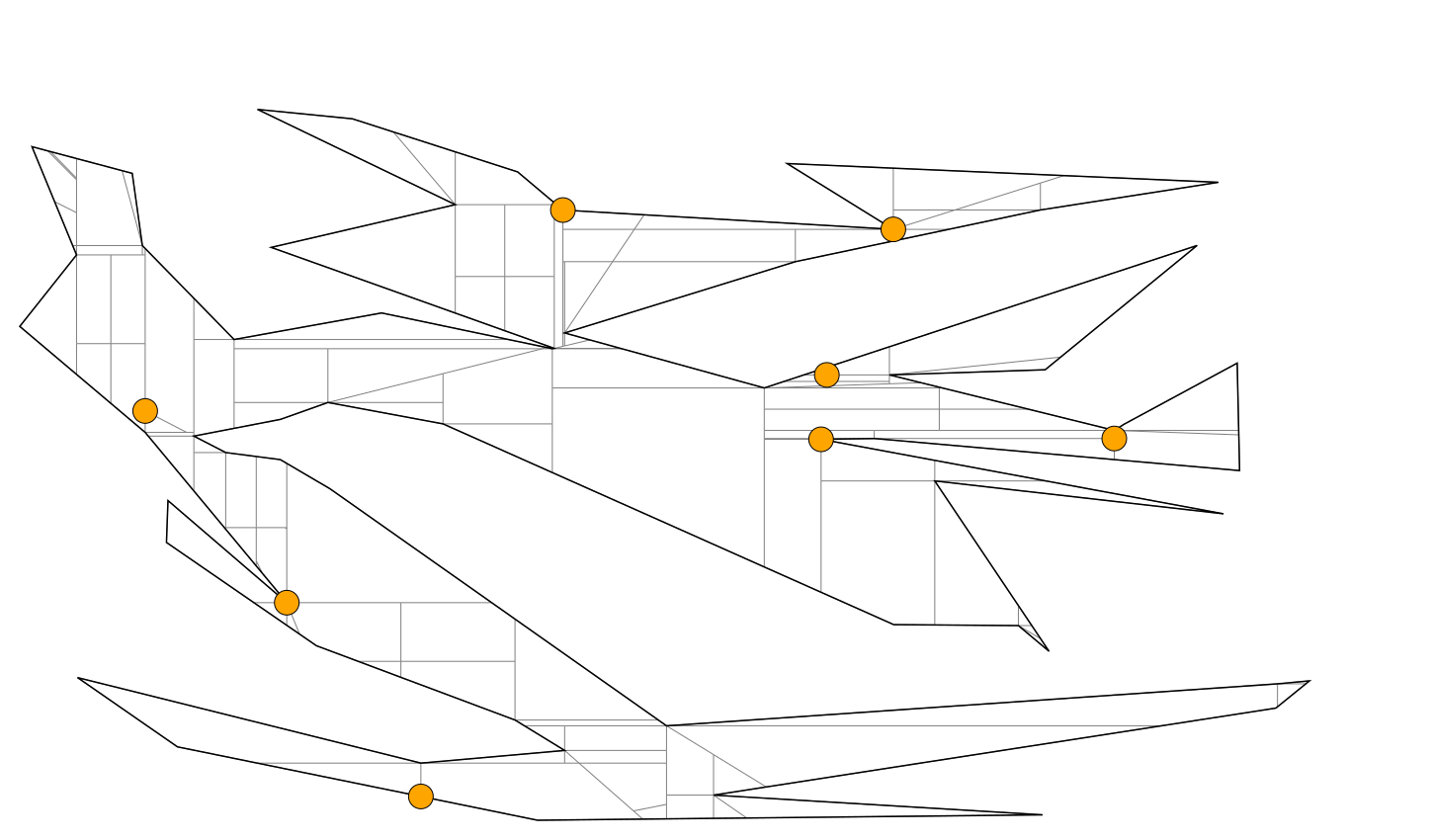}
	9
\end{minipage}

\newpage
\bibliographystyle{plain}
\bibliography{library}

\begin{thebibliography}{10}

\bibitem{abel}
Zachary Abel, Erik~D. Demaine, Martin~L. Demaine, Sarah Eisenstat, Jayson
  Lynch, and Tao~B. Schardl.
\newblock Who needs crossings? {H}ardness of plane graph rigidity.
\newblock In {\em 32nd International Symposium on Computational Geometry (SoCG
  2016)}, pages 3:1--3:15, 2016.

\bibitem{abrahamsen2013constant}
Mikkel Abrahamsen.
\newblock Constant-workspace algorithms for visibility problems in the plane.
\newblock {\em Master's thesis, University of Copenhagen}, 2013.

\bibitem{abrahamsen2017irrational}
Mikkel Abrahamsen, Anna Adamaszek, and Tillmann Miltzow.
\newblock Irrational guards are sometimes needed.
\newblock In {\em SoCG 2017}, pages 3:1--3:15, 2017.
\newblock Arxiv 1701.05475.

\bibitem{ARTETR}
Mikkel Abrahamsen, Anna Adamaszek, and Tillmann Miltzow.
\newblock The art gallery problem is {\ER}-complete.
\newblock {\em JACM 2022}, 2022.
\newblock Arxiv 1704.06969.

\bibitem{etrNeurons}
Mikkel Abrahamsen, Linda Kleist, and Tillmann Miltzow.
\newblock Training neural networks is er-complete.
\newblock {\em Neurips, see also \url{https://arxiv.org/abs/2102.09798}}, 2021.

\bibitem{SimplexER}
Mikkel Abrahamsen, Linda Kleist, and Tillmann Miltzow.
\newblock Geometric embeddability of complexes is
  {\(\exists\)}{\(\mathbb{r}\)}-complete.
\newblock In Erin~W. Chambers and Joachim Gudmundsson, editors, {\em 39th
  International Symposium on Computational Geometry, SoCG 2023, June 12-15,
  2023, Dallas, Texas, {USA}}, volume 258 of {\em LIPIcs}, pages 1:1--1:19.
  Schloss Dagstuhl - Leibniz-Zentrum f{\"{u}}r Informatik, 2023.

\bibitem{etrPacking}
Mikkel Abrahamsen, Tillmann Miltzow, and Nadja Seiferth.
\newblock A framework for $\exists\mathbb{R}$-completeness of
  two-dimensional~packing~problems.
\newblock {\em FOCS}, 2020.

\bibitem{TerrainConflictFreeFPT}
Akanksha Agrawal, Pradeesha Ashok, Meghana~M. Reddy, Saket Saurabh, and Dolly
  Yadav.
\newblock {FPT} algorithms for conflict-free coloring of graphs and chromatic
  terrain guarding.
\newblock {\em Arxiv}, 1905.01822, 2019.

\bibitem{AlmostConvex}
Akanksha Agrawal, Kristine V.~K. Knudsen, Daniel Lokshtanov, Saket Saurabh, and
  Meirav Zehavi.
\newblock {The Parameterized Complexity of Guarding Almost Convex Polygons}.
\newblock In {\em SoCG 2020}, LIPIcs, pages 3:1--3:16, 2020.

\bibitem{agrawal2020parameter}
Akanksha Agrawal, Sudeshna Kolay, and Meirav Zehavi.
\newblock Parameter analysis for guarding terrains.
\newblock In {\em SWAT}, 2020.

\bibitem{agrawal2020parameterized}
Akanksha Agrawal and Meirav Zehavi.
\newblock Parameterized analysis of art gallery and terrain guarding.
\newblock In {\em International Computer Science Symposium in Russia}, pages
  16--29. Springer, 2020.

\bibitem{PracticalARTamit2010}
Yoav Amit, Joseph~S.B. Mitchell, and Eli Packer.
\newblock Locating guards for visibility coverage of polygons.
\newblock {\em International Journal of Computational Geometry \&
  Applications}, 20(05):601--630, 2010.

\bibitem{GermanTSP}
David Applegate, Robert Bixby, Va\v{s}ek Chv\'{a}tal, and William Cook.
\newblock {\em TSP in practice}, 2021.

\bibitem{ashok2019efficient}
Pradeesha Ashok and Meghana Reddy.
\newblock Efficient guarding of polygons and terrains.
\newblock In {\em International Workshop on Frontiers in Algorithmics}, pages
  26--37. Springer, 2019.

\bibitem{belleville1991computing}
Patrice Belleville.
\newblock Computing two-covers of simple polygons.
\newblock Master's thesis, McGill University, 1991.

\bibitem{train-fully-neural-networks}
Daniel Bertschinger, Christoph Hertrich, Paul Jungeblut, Tillmann Miltzow, and
  Simon Weber.
\newblock Training fully connected neural networks is
  $\exists\mathbb{R}$-complete.
\newblock {\em Arxive}, abs/2204.01368, 2022.

\bibitem{PritamConstantFactor}
Pritam Bhattacharya, Subir~Kumar Ghosh, and Sudebkumar~Prasant Pal.
\newblock Constant approximation algorithms for guarding simple polygons using
  vertex guards.
\newblock {\em Arxiv 1712.05492}, 2017.

\bibitem{bienstock1991some}
Daniel Bienstock.
\newblock Some provably hard crossing number problems.
\newblock {\em Discrete \& Computational Geometry}, 6(3):443--459, 1991.

\bibitem{cplex}
Robert Bixby.
\newblock {IBM CPLEX}.
\newblock \url{https://www.ibm.com/analytics/cplex-optimizer}.

\bibitem{BonnetM17Approx}
{\'{E}}douard Bonnet and Tillmann Miltzow.
\newblock An approximation algorithm for the art gallery problem.
\newblock In {\em SoCG 2017}, pages 20:1--20:15, 2017.
\newblock arXiv 1607.05527.

\bibitem{BonnetW1HARD}
\'{E}douard Bonnet and Tillmann Miltzow.
\newblock Parameterized hardness of art gallery problems.
\newblock {\em ACM Transactions on Algorithms}, 16(4), 2020.

\bibitem{SoCGVideo}
Dorit Borrmann, Pedro~J. de~Rezende, Cid~C. de~Souza, S{\'{a}}ndor~P. Fekete,
  Stephan Friedrichs, Alexander Kr{\"{o}}ller, Andreas N{\"{u}}chter,
  Christiane Schmidt, and Davi~C. Tozoni.
\newblock Point guards and point clouds: solving general art gallery problems.
\newblock In {\em SoCG}, pages 347--348. {ACM}, 2013.

\bibitem{PracticalARTbottino2008}
Andrea Bottino and Aldo Laurentini.
\newblock A nearly optimal sensor placement algorithm for boundary coverage.
\newblock {\em Pattern Recognition}, 41(11):3343--3355, 2008.

\bibitem{PracticalARTbottino2011}
Andrea Bottino and Aldo Laurentini.
\newblock A nearly optimal algorithm for covering the interior of an art
  gallery.
\newblock {\em Pattern Recognition}, 44(5):1048--1056, 2011.

\bibitem{hemmer}
Francisc Bungiu, Michael Hemmer, John Hershberger, Kan Huang, and Alexander
  Kr{\"{o}}ller.
\newblock Efficient computation of visibility polygons.
\newblock {\em arXiv 1403.3905}, 2014.

\bibitem{cardinal2017intersection}
Jean Cardinal, Stefan Felsner, Tillmann Miltzow, Casey Tompkins, and Birgit
  Vogtenhuber.
\newblock Intersection graphs of rays and grounded segments.
\newblock {\em Journal of Graph Algorithms and Applications}, 22:273--295,
  2018.

\bibitem{cardinal2017recognition}
Jean Cardinal and Udo Hoffmann.
\newblock Recognition and complexity of point visibility graphs.
\newblock {\em Discrete \& Computational Geometry}, 57(1):164--178, 2017.

\bibitem{KnauerWitness}
Kyung{-}Yong Chwa, Byung{-}Cheol Jo, Christian Knauer, Esther Moet, Ren{\'{e}}
  van Oostrum, and Chan{-}Su Shin.
\newblock Guarding art galleries by guarding witnesses.
\newblock {\em Int. J. Comput. Geometry Appl.}, 16(2-3):205--226, 2006.

\bibitem{art-gallery-instances-page}
Marcelo~C. Couto, Pedro~J. de~Rezende, and Cid~C. de~Souza.
\newblock Instances for the {Art Gallery Problem}, 2009.
\newblock \url{www.ic.unicamp.br/~cid/Problem-instances/Art-Gallery}.

\bibitem{PracticalARTcouto2011}
Marcelo~C. Couto, Pedro~J. de~Rezende, and Cid~C. de~Souza.
\newblock An exact algorithm for minimizing vertex guards on art galleries.
\newblock {\em International Transactions in Operational Research},
  18(4):425--448, 2011.

\bibitem{PracticalARTcouto2008}
Marcelo~C. Couto, Cid~C. de~Souza, and Pedro~J. de~Rezende.
\newblock Experimental evaluation of an exact algorithm for the orthogonal art
  gallery problem.
\newblock In {\em International Workshop on Experimental and Efficient
  Algorithms}, pages 101--113. Springer, 2008.

\bibitem{engineering}
Pedro~J. de~Rezende, Cid~C. de~Souza, Stephan Friedrichs, Michael Hemmer,
  Alexander Kr{\"{o}}ller, and Davi~C. Tozoni.
\newblock Engineering art galleries.
\newblock {\em Algorithm Engineering}, pages 379--417, 2016.

\bibitem{deligkas2020square}
Argyrios Deligkas, John Fearnley, and Themistoklis Melissourgos.
\newblock Square-cut pizza sharing is ppa-complete.
\newblock {\em arXiv preprint arXiv:2012.14236}, 2020.

\bibitem{AreasKleist}
Michael~G. Dobbins, Linda Kleist, Tillmann Miltzow, and Pawe\l{} Rz{\c
  a}{\.z}ewski.
\newblock {$\forall \exists \mathbb{R}$-completeness and area-universality}.
\newblock {\em WG 2018}, 2018.
\newblock Arxiv 1712.05142.

\bibitem{ArxivSmoothedART}
Michael~Gene Dobbins, Andreas Holmsen, and Tillmann Miltzow.
\newblock Smoothed analysis of the art gallery problem.
\newblock {\em arXiv}, 1811.01177, 2018.

\bibitem{NestedPolytopesER}
Michael~Gene Dobbins, Andreas Holmsen, and Tillmann Miltzow.
\newblock A universality theorem for nested polytopes.
\newblock {\em arXiv}, 1908.02213, 2019.

\bibitem{EfratH06}
Alon Efrat and Sariel Har{-}Peled.
\newblock Guarding galleries and terrains.
\newblock {\em Inf. Process. Lett.}, 100(6):238--245, 2006.

\bibitem{eidenbenz2001inapproximability}
Stephan Eidenbenz, Christoph Stamm, and Peter Widmayer.
\newblock Inapproximability results for guarding polygons and terrains.
\newblock {\em Algorithmica}, 31(1):79--113, 2001.

\bibitem{robustFramework}
Jeff Erickson, Ivor van~der Hoog, and Tillmann Miltzow.
\newblock {Smoothing the Gap Between NP and ER}.
\newblock {\em SIAM Journal on Computing}, 3(2):102--138, 2020.

\bibitem{GeoThicknessER}
Henry F{\"{o}}rster, Philipp Kindermann, Tillmann Miltzow, Irene Parada, Soeren
  Terziadis, and Birgit Vogtenhuber.
\newblock Geometric thickness of multigraphs is {\(\exists\)}r-complete.
\newblock {\em CoRR}, abs/2312.05010, 2023.

\bibitem{PracticalARTMasterFriedrich}
S.~Friedrichs.
\newblock Integer solutions for the art gallery problem using linear
  programming.
\newblock Master Thesis, 2012.

\bibitem{garg2015etr}
Jugal Garg, Ruta Mehta, Vijay~V. Vazirani, and Sadra Yazdanbod.
\newblock {ETR}-completeness for decision versions of multi-player (symmetric)
  {N}ash equilibria.
\newblock In {\em ICALP 2015}, pages 554--566, 2015.

\bibitem{ghosh2010approximation}
Subir~Kumar Ghosh.
\newblock Approximation algorithms for art gallery problems in polygons.
\newblock {\em Discrete Applied Mathematics}, 158(6):718--722, 2010.

\bibitem{PanosQuestion}
Panos Giannopoulos.
\newblock Open problems: guarding problems, 2016.

\bibitem{guibas1987linear}
Leonidas Guibas, John Hershberger, Daniel Leven, Micha Sharir, and Robert
  Tarjan.
\newblock Linear-time algorithms for visibility and shortest path problems
  inside triangulated simple polygons.
\newblock {\em Algorithmica}, 2(1-4):209--233, 1987.

\bibitem{ConvexExpansion}
Simon Hengeveld, Tillmann Miltzow, and Frank Staals.
\newblock Weak visibility by convex expansion.
\newblock in preparation 2020.

\bibitem{HMWW24}
Michael Hoffmann, Tillmann Miltzow, Simon Weber, and Lasse Wulf.
\newblock {R}ecognition of {U}nit {S}egment and {P}olyline {G}raphs is $\exists
  \mathbb{R}$-complete.
\newblock {\em Arxiv}, abs/2401.02172(2401.02172):1--18, 2024.

\bibitem{Hausdorf-UER}
Paul Jungeblut, Linda Kleist, and Tillmann Miltzow.
\newblock The complexity of the hausdorff distance.
\newblock In Xavier Goaoc and Michael Kerber, editors, {\em 38th International
  Symposium on Computational Geometry, SoCG 2022, June 7-10, 2022, Berlin,
  Germany}, volume 224 of {\em LIPIcs}, pages 48:1--48:17. Schloss Dagstuhl -
  Leibniz-Zentrum f{\"{u}}r Informatik, 2022.

\bibitem{kang2011sphere}
Ross~J. Kang and Tobias M{\"u}ller.
\newblock Sphere and dot product representations of graphs.
\newblock In {\em SoCG}, pages 308--314. ACM, 2011.

\bibitem{khodakarami2015fixed}
Farnoosh Khodakarami, Farzad Didehvar, and Ali Mohades.
\newblock A fixed-parameter algorithm for guarding 1.5 d terrains.
\newblock {\em Theoretical Computer Science}, 595:130--142, 2015.

\bibitem{khodakarami20171}
Farnoosh Khodakarami, Farzad Didehvar, and Ali Mohades.
\newblock 1.5 d terrain guarding problem parameterized by guard range.
\newblock {\em Theoretical Computer Science}, 661:65--69, 2017.

\bibitem{ApproXKirkpatrick15}
David~G. Kirkpatrick.
\newblock An {$O(\lg \lg \mathrm{OPT})$}-approximation algorithm for
  multi-guarding galleries.
\newblock {\em Discrete {\&} Computational Geometry}, 53(2):327--343, 2015.

\bibitem{kisfaludi2020gap}
S{\'a}ndor Kisfaludi-Bak, Jesper Nederlof, and Karol W{{e}}grzycki.
\newblock A gap-eth-tight approximation scheme for euclidean tsp.
\newblock {\em arXiv preprint arXiv:2011.03778}, 2020.

\bibitem{LindaPHD}
Linda Kleist.
\newblock {\em Planar graphs and faces areas -- Area-Universality}.
\newblock PhD thesis, Technische Universit\"at Berlin, 2018.
\newblock PhD thesis.

\bibitem{klute2021local}
Fabian Klute, Meghana~M Reddy, and Tillmann Miltzow.
\newblock Local complexity of polygons.
\newblock {\em arXiv preprint arXiv:2101.07554}, 2021.

\bibitem{CGAL-Bad-Windows}
Bernhard Kornberger.
\newblock Why is cgal significantly slower under windows?
\newblock \url{https://stackoverflow.com/questions/58008543/}.

\bibitem{PracticalARTkroller2012}
Alexander Kr{\"o}ller, Tobias Baumgartner, S{\'a}ndor~P. Fekete, and Christiane
  Schmidt.
\newblock Exact solutions and bounds for general art gallery problems.
\newblock {\em Journal of Experimental Algorithmics (JEA)}, 17:2--3, 2012.

\bibitem{LeeLin86}
Der-Tsai Lee and Arthur~K. Lin.
\newblock Computational complexity of art gallery problems.
\newblock {\em {IEEE} Transactions on Information Theory}, 32(2):276--282,
  1986.

\bibitem{AnnaPreparation}
Anna Lubiw, Tillmann Miltzow, and Debajyoti Mondal.
\newblock The complexity of drawing a graph in a polygonal region.
\newblock {\em {A}rxiv}, 2018.
\newblock Graph Drawing 2018.

\bibitem{mcdiarmid2013integer}
Colin McDiarmid and Tobias M{\"u}ller.
\newblock Integer realizations of disk and segment graphs.
\newblock {\em Journal of Combinatorial Theory, Series B}, 103(1):114--143,
  2013.

\bibitem{profiler}
Microsoft.
\newblock {Visual Studio Profiler}.
\newblock
  \url{https://docs.microsoft.com/en-us/visualstudio/profiling/?view=vs-2019}.

\bibitem{MS22}
Tillmann Miltzow and Reinier~F. Schmiermann.
\newblock On classifying continuous constraint satisfaction problems.
\newblock In {\em 2021 {IEEE} 62nd {A}nnual {S}ymposium on {F}oundations of
  {C}omputer {S}cience---{FOCS} 2021}, pages 781--791. IEEE Computer Soc., Los
  Alamitos, CA, [2022] \copyright 2022.

\bibitem{mnev1988universality}
Nicolai~E Mn{\"e}v.
\newblock The universality theorems on the classification problem of
  configuration varieties and convex polytopes varieties.
\newblock In Oleg~Y. Viro, editor, {\em Topology and geometry -- Rohlin
  seminar}, pages 527--543. Springer-Verlag Berlin Heidelberg, 1988.

\bibitem{PerfectGraphApproach}
Rajeev Motwani, Arvind Raghunathan, and Huzur Saran.
\newblock Covering orthogonal polygons with star polygons: The perfect graph
  approach.
\newblock {\em J. Comput. Syst. Sci.}, 40(1):19--48, 1990.

\bibitem{o1987art}
Joseph O'Rourke.
\newblock {\em Art Gallery Theorems and Algorithms}.
\newblock Oxford University Press, 1987.

\bibitem{richter1995realization}
J{\"u}rgen Richter-Gebert and G{\"u}nter~M. Ziegler.
\newblock Realization spaces of 4-polytopes are universal.
\newblock {\em Bulletin of the American Mathematical Society}, 32(4):403--412,
  1995.

\bibitem{Beyond-Worst-Case}
Tim Roughgarden.
\newblock Beyond the worst-case analysis of algorithms (introduction).
\newblock {\em CoRR}, abs/2007.13241, 2020.

\bibitem{Schaefer2010}
Marcus Schaefer.
\newblock Complexity of some geometric and topological problems.
\newblock In {\em Proceedings of the 17th International Symposium on Graph
  Drawing (GD 2009)}, volume 5849 of {\em Lecture Notes in Computer Science
  (LNCS)}, pages 334--344. Springer, 2009.

\bibitem{schaefer2013realizability}
Marcus Schaefer.
\newblock Realizability of graphs and linkages.
\newblock In {\em Thirty Essays on Geometric Graph Theory}, pages 461--482.
  Springer, 2013.

\bibitem{Schaefer-ETR}
Marcus Schaefer and Daniel \v{S}tefankovi\v{c}.
\newblock Fixed points, {N}ash equilibria, and the existential theory of the
  reals.
\newblock {\em Theory of Computing Systems}, 60(2):172--193, 2017.

\bibitem{SchuchardtH95}
Dietmar Schuchardt and Hans{-}Dietrich Hecker.
\newblock Two {NP}-hard art-gallery problems for ortho-polygons.
\newblock {\em Math. Log. Q.}, 41:261--267, 1995.

\bibitem{shitov2016universality}
Yaroslav Shitov.
\newblock A universality theorem for nonnegative matrix factorizations.
\newblock {\em Arxiv 1606.09068}, 2016.

\bibitem{Shitov16a}
Yaroslav Shitov.
\newblock The complexity of positive semidefinite matrix factorization.
\newblock {\em {SIAM} Journal on Optimization}, 27(3):1898--1909, 2017.

\bibitem{shor1991stretchability}
Peter Shor.
\newblock Stretchability of pseudolines is np-hard.
\newblock {\em Applied Geometry and Discrete Mathematics-The Victor Klee
  Festschrift}, 1991.

\bibitem{stade2022complexity}
Jack Stade.
\newblock Complexity of the boundary-guarding art gallery problem.
\newblock {\em arXiv preprint arXiv:2210.12817}, 2022.

\bibitem{cgal:eb-20a}
{The CGAL Project}.
\newblock {\em {CGAL} User and Reference Manual}.
\newblock {CGAL Editorial Board}, 4.1.3 edition, 2020.

\bibitem{tozoni2013practical}
Davi~C. Tozoni, Pedro~J. de~Rezende, and Cid~C de~Souza.
\newblock A practical iterative algorithm for the art gallery problem using
  integer linear programming.
\newblock {\em Optimization Online}, 2013.

\bibitem{Quest-Tozoni}
Davi~C. Tozoni, Pedro~J. de~Rezende, and Cid~C. de~Souza.
\newblock The quest for optimal solutions for the art gallery problem: A
  practical iterative algorithm.
\newblock In {\em Experimental Algorithms}, pages 320--336, 2013.

\bibitem{tozoni}
Davi~C. Tozoni, Pedro J.~de Rezende, and Cid C.~de Souza.
\newblock Algorithm 966: A practical iterative algorithm for the art gallery
  problem using integer linear programming.
\newblock {\em ACM Trans. Math. Softw.}, 43(2), August 2016.

\bibitem{Z92}
Xiao-Dong Zhang.
\newblock Complexity of neural network learning in the real number model.
\newblock In {\em Workshop on Physics and Computation}, pages 146--150, 1992.

\end{thebibliography}

\end{document}